\documentclass[nohyperref]{article}

\usepackage{natbib}

\usepackage{amsmath, amsthm, amssymb}
\usepackage{fullpage}
\usepackage{xcolor}
\usepackage{algorithm, algorithmicx}
\usepackage{amsmath}
\usepackage{appendix}
\usepackage[normalem]{ulem}
\usepackage[noend]{algpseudocode}
\usepackage{thm-restate}

\newcommand{\E}{\mathbb{E}}
\newcommand{\sigp}{\sigma_p^2}
\newcommand{\widehatp}{\widehat{p}}
\newcommand{\widehatpi}{\widehat{p}_i}
\newcommand{\widehatsigi}{\widehat{\sigma}_i^2}
\newcommand{\widehatsigp}{\widehat{\sigma}_p^2}
\newcommand{\bersumP}{\Dee(k)}
\newcommand{\berP}{\Dee}
\newcommand{\initialvariance}{\mathcal{M}}
\newcommand{\berryesseen}{\gamma}
\newcommand{\var}{{\rm Var}}
\newcommand{\Lap}{{\rm Lap}}

\newcommand{\truncatedfactor}{\lambda}
\newcommand{\thirdmoment}{\rho}
\newcommand{\Dee}{\mathcal{D}}

\newcommand{\Ber}{{\rm Ber}}
\newcommand{\Bin}{{\rm Bin}}
\newcommand{\medk}{k_{\rm med}}

\newcommand{\linearest}{\text{\rm \texttt{NLE}}}
\newcommand{\eps}{\epsilon}
\newcommand{\sensitivity}{\Lambda}
\newcommand{\meanest}{\rm \texttt{mean}_{\eps,\delta}}
\newcommand{\varianceest}{\texttt{variance}_{\eps,\delta}}
\newcommand{\expmech}{\texttt{EM}_{\eps}}
\newcommand{\kneighbor}{\kappa}
\newcommand{\variancesamplesize}{L}

\newtheorem{example}{Example}

\newcommand{\bersumPi}[1]{\Dee({#1})}

\newcommand{\pinitial}{\widehat{p}_{\epsilon}^{\rm initial}}
\newcommand{\pnpideal}{\widehat{p}^{\rm ideal}}
\newcommand{\pideal}{\widehat{p}_{\epsilon}^{\rm ideal}}
\newcommand{\pnpinitial}{\widehat{p}^{\rm initial}}
\newcommand{\pnprealistic}{\widehat{p}}
\newcommand{\prealistic}{\widehat{p}_{\epsilon}}

\newcommand{\LS}{\text{\rm \texttt{LS}}}
\newcommand{\unknownki}{private $k$ }
\newcommand{\knownki}{\text{public-size }}
\newcommand{\prealisticunknown}{\hat{p}_{\epsilon}^{\text{\rm priv } k}}

\newcommand{\removeproofs}[1]{\begin{proof} \textcolor{blue}{proof written but hidden}\end{proof}}
\newcommand{\showproofs}[1]{#1}%{\begin{proof} \textcolor{blue}{proof written but hidden}\end{proof}}
\newcommand{\concentrationbound}[5]{f^{#1}_{#2}(#3, #4, #5)}

\definecolor{DarkGreen}{rgb}{0.1,0.5,0.1}

\usepackage{microtype}
\usepackage{graphicx}
\usepackage{subfigure}
\usepackage{booktabs} % for professional tables

\usepackage{hyperref}

\usepackage{amsmath}
\usepackage{amssymb}
\usepackage{mathtools}
\usepackage{amsthm}

\usepackage[utf8]{inputenc} % allow utf-8 input
\usepackage[T1]{fontenc}    % use 8-bit T1 fonts
\usepackage{url}            % simple URL typesetting
\usepackage{booktabs}       % professional-quality tables
\usepackage{amsfonts}       % blackboard math symbols
\usepackage{nicefrac}       % compact symbols for 1/2, etc.
\usepackage{microtype}      % microtypography
\usepackage{xcolor}         % colors

\usepackage[capitalize,noabbrev]{cleveref}

\theoremstyle{plain}
\newtheorem{theorem}{Theorem}[section]

\newtheorem{lemma}[theorem]{Lemma}

\theoremstyle{definition}
\newtheorem{definition}[theorem]{Definition}

\theoremstyle{remark}

\begin{document}

\title{Mean Estimation with User-level Privacy under Data Heterogeneity}

\author{Rachel Cummings\footnotemark[1] \and Vitaly Feldman\footnotemark[2] \and Audra McMillan\footnotemark[2] \and Kunal Talwar\footnotemark[2]}

\renewcommand{\thefootnote}{\fnsymbol{footnote}}
\footnotetext[1]{Columbia University. Part of this work was completed while the author was at Apple. Supported in part by NSF grant CNS-2138834 and an Apple Privacy-Preserving Machine Learning Award.}
\footnotetext[2]{Apple}

\renewcommand{\thefootnote}{\arabic{footnote}}

\date{}

\maketitle

\begin{abstract}
A key challenge in many modern data analysis tasks is that user data are heterogeneous. Different users may possess vastly different numbers of data points. More importantly, it cannot be assumed that all users sample from the same underlying distribution.  This is true, for example in language data, where different speech styles result in data heterogeneity. In this work we propose a simple model of heterogeneous user data that allows user data to differ in both distribution and quantity of data, and provide a method for estimating the population-level mean while preserving user-level differential privacy. We demonstrate asymptotic optimality of our estimator and also prove general lower bounds on the error achievable in the setting we introduce. 
\end{abstract}

\section{Introduction}
Many practical problems in statistical data analysis and machine learning deal with the setting in which each user generates multiple data points. In such settings the distribution of each user's data may be somewhat different and, furthermore, users may possess vastly different numbers of samples. This issue is one the key challenges in federated learning \citep{Kairouz21} leading to considerable interest in models and algorithms that address this issue.

As an example, consider the task of next-word prediction for a keyboard. Different users typing on a keyboard may have different styles of writing or focus on different topics, leading to different distributions. There are aspects of the language that are common to all users, and likely additional aspects of style that are common to large groups of users. Thus while each user has their own data distribution, there are commonalities between the distributions, and additional commonalities amongst distributions corresponding to particular subsets of users. Modeling and learning such relationships between users' distributions is crucial for building a better global model for all users, as well as for personalizing models for users. 

The focus of this work is on differentially private algorithms for such settings. We assume that there is an unknown global meta-distribution $\Dee$. For each user $i$, a personal data distribution $\Dee_i$ is chosen randomly from $\Dee$ (for example, by sampling a set of parameters that define $\Dee_i$). Each user then receives some number $k_i$ of i.i.d.~samples from $\Dee_i$. The goal is to solve an analysis task relative to $\Dee$, with an eye towards better modeling of each $\Dee_i$ even when $k_i$ is small. This abstract setting can model many practical settings where the relationships between the $\Dee_i$'s take different forms. Indeed the standard loss in federated learning is the (unweighted) average over users of a per-user loss function~\citep[Sec. 3.3.2]{Kairouz21}, which corresponds to learning when the underlying distribution is $\Dee$. Little theoretical work has been done in this setting and even the most basic statistical tasks are poorly understood. Thus we start by focusing on the fundamental problem of mean estimation. Specifically, in our model, $\Dee$ is a distribution on the interval $[0,1]$ with unknown mean $p$ and unknown variance $\sigma_p^2$. Further, we assume that $\Dee_i$ is simply a Bernoulli distribution with mean $p_i \sim \Dee$. 

While the general $\mathcal{D}_i$ setting is of interest, the Bernoulli case captures a variety of interesting use cases. For example, each sample from the Bernoulli distribution could represent whether or not the user has clicked on an ad. Another common example is model evaluation, where the user produces a Bernoulli sample by engaging or not engaging with a feature (e.g., phone keyboard next word suggestion, crisis helpline link, search engine knowledge panels, sponsored link in search results, etc.). As a concrete example, a language model is used to make the next word suggestions on a phone keyboard. A new version of this model would be first tested to measure the average suggestion acceptance rate over users. Each user would thus generate a set of independent Bernoulli r.v.’s with each individual mean $p_i$ corresponding to the model accuracy for the specific user. Heterogeneity comes from different users typing differently (and hence model accuracy varying across users) and using the keyboard with different frequency. Note that the distribution of model accuracies among users is the meta distribution $\mathcal{D}$ in our work. More generally, measuring the average accuracy of a classification model among a large group of users is an important task in itself. Such models are deployed in privacy-sensitive applications such as health and finance. The resulting statistics may need to be shared with third parties or other teams within a company, raising potential user privacy concerns.

Our main contribution is a differentially private algorithm that estimates the mean of $\Dee$ in this heterogeneous setting. We first study this question in an idealized setting, where the variance of $\Dee$ is known, and no privacy constraints. Here the optimal non-private estimator for $p_i$ is simple and linear: it is a weighted linear combination of the individual user means with weights that depend on the $k_i$'s and on $\sigma_p$. The variance of this estimate is $\sigma_{ideal}^2 \approx (\sum_i \min(k_i, \sigma_p^{-2}))^{-1}$. This expression has a natural interpretation: this is the variance from using $\min(k_i, \sigma_p^{-2})$ samples from user $i$ and averaging all the Bernoulli samples thus obtained. 
We then design a differentially private estimator for $p$. We show that under mild assumptions, there is no asymptotic price to privacy (and to not knowing $\sigma_p$). That is, our differentially private estimator  has variance $\tilde{O}(\sigma_{ideal}^2)$. For some intuition, note that the restriction on using at most $\sigma_p^{-2}$ samples from each user ensures that the estimator is not too affected by their individual mean $p_i$.
Interestingly, the estimator achieving this bound in the private setting is non-linear. Further, we show that $\sigma_{ideal}^2$ is close to the best achievable variance, under some mild technical conditions. 

Our technical results highlight several of the challenges associated with ensuring user-level privacy when data is heterogeneous. For example, in the heterogeneous setting, the optimal choice of weights for each user contribution depends on properties of $\Dee$ that also need to be estimated from the data. Further, we show a novel approach to proving lower bounds for private statistical estimation in the heterogenous setting. Our approach builds on the proof of the Cram\'er-Rao lower bound in statistics, and we show how privacy terms can be incorporated in this approach to show near optimality of our algorithms for nearly every setting of $k_i$'s. These tools and insights should be useful for modeling and designing algorithms for more involved data analysis tasks. 

We note that the optimal algorithm for this problem was not known prior to this work, even in the special case where all $\Dee_i$'s are identical (or, equivalently, $\sigma_p^2=0$) but users hold different numbers of samples. In the absence of privacy constraints, this setting poses no additional complexity over the case where each user has a single data point, since the data points all come from the same distribution. However, with the requirement of user-level differential privacy, even this special case appears to require many of the technical tools developed in this work (see Section \ref{s.constantpi} for a detailed discussion).

We aim to help foster similar model-driven exploration in other settings. There have been attempts to handle heterogeneity by phrasing the problem as meta-learning or multi-task learning~\citep[Sec 3.3.3]{Kairouz21}. These works rely on implicit assumptions about the different distributions. Our goal is to start with a more principled approach that makes explicit the assumptions on the relationship between different distributions and use that to derive algorithms. For example, if we were to model the $\Dee_i$'s as having means coming from a mixture of Gaussians, the estimation of cluster means would be a necessary step in an EM-type algorithm. 
Our choice of $\Dee_i$'s being Bernoulli is meant to capture discrete distribution learning problems that have been extensively studied in private federated settings. Our techniques are general and extend naturally to real-valued random variables where, e.g., $\Dee_i$ is a Gaussian with mean $p_i$ and known variance. While we make minimal assumptions on $\Dee$, our results asymptotically match the lower bounds for the case of $\Dee$ being Gaussian with known variance. Our techniques also have natural extensions to higher dimensions.

\paragraph{Summary of our results:}
Our main results involve three estimators; an idealized (non-realisable) estimator $\pideal$ that assumes that the mean and variance of $\Dee$ are known to the algorithm,
an estimator $\prealistic$ that is private with respect to the user's samples, but not with respect to each user's number of samples $k_i$, and finally an estimator $\prealisticunknown$ that is private with respect to both the samples \emph{and} the number of samples. Let $\widehatpi$ be the mean of the $k_i$ samples from user $i$. The estimators $\prealistic$ and $\prealisticunknown$ both require as input initial, less accurate $(\eps, \delta)$-DP mean and variance estimators $\meanest$ and  $\varianceest$. The main results of this paper can be (informally) summarised as follows:
\begin{itemize}
    \item \textbf{Near optimality of $\pideal$ [Theorem~\ref{optimalityfisherinfo}].} For any parameterized family of distributions $p\mapsto\Dee_p$, such that the Fisher information of $\widehatpi$ is inversely proportional to the variance of $\widehatpi$ for all $i$, each $\widehatpi$ is sufficiently-well concentrated (e.g. sub-Gaussian) and $p\in[1/3,2/3]$, we have that $\pideal$ is minimax optimal, up to logarithmic (in $n$) factors, among all unbiased estimators of $p$. The estimator $\pideal$ itself is not unbiased,
    but it has very low bias. The proof of this result involves a Cram\'er-Rao style argument which may be of independent interest. This result allows us to use $\pideal$ as a yardstick by which to compare $\prealistic$ and $\prealisticunknown$.
    \item \textbf{Near optimality of $\prealistic$ [Theorem~\ref{metatheorem}].}  Assume there exists mean and variance estimators, $\meanest$ and $\varianceest$, such that when
    run with a constant fraction (say $n/10$) of the users, $\meanest$ returns a sufficiently good estimate of $p$ (roughly  no worse than the estimate from any single user, and implies a constant multiplicative approximation to $p(1-p)$), and when run with $\log n/\epsilon$ users, $\varianceest$ returns a constant multiplicative approximation to $\sigma_p^2$. If the maximum $k_i$ and median $k_i$ are within a factor of $(n\eps/\log n)-1$, then the variance of $\prealistic$, with $\meanest$ and $\varianceest$ as the inputted initial estimators, is within a constant factor of the variance of $\pideal$. The conditions on $\meanest$ and $\varianceest$ are not particularly stringent and such estimators exist, for example, when $\Dee$ is a truncated Gaussian distribution with mean bounded away from 0 or 1 and sufficiently small variance.
   \item \textbf{Near Optimality of $\prealisticunknown$ [Theorem~\ref{thm.privatek}].} 
   Under slightly more stringent conditions on $\Dee$ and the assumption that the maximum $k_i$ and median $k_i$ are within a factor of $O(n\eps^2/\log n)$, we extend the upper bounds to the case when $k_i$'s are also considered private information. The conditions are again satisfied, for example, by truncated Gaussian distributions with mean bounded away from 0 or 1  and sufficiently small variance. 
   \item \textbf{Lower bound in terms of $k_i$ [Corollary~\ref{cor.lower}].} Finally, we show that for any sequence $k_1,\cdots,k_n$ and variance $\sigp$ there exists $k^*$ and a family of distributions $p
   \mapsto\Dee_p$ such that the minimax optimal error among all unbiased estimators of $p$, for $p$ in the range $[1/3,2/3]$, is lower bounded by
   \[\tilde{\Omega}\left(\min\left\{ \sqrt{\tfrac{\frac{k^*}{\epsilon^2}+\sum_{i=1}^n \min\{k_i, k^*\}}{(\sum_{i=1}^n \min\{k_i,\sqrt{k_ik^*}\})^2}}, \frac{\sigma_p}{\sqrt{n}}\right\}\right).\]
    \end{itemize}

We note that our main algorithmic results require concentration of the meta-distribution $\mathcal{D}$. We note that in practice, this is not an unreasonable assumption. For example, in the case of model evaluation, it may be be reasonable to assume that a general model has similar accuracy for the vast majority of users, or formally, that the model accuracy is well-concentrated.

\subsection{Related Work}

Frequency estimation in the example-level privacy model has been well-studied in the central~\citep{Dwork:2006,DworkRoth:14} and local models~\citep{hsu2012distributed, ErlingssonPK14, chen2020breaking, acharya2019communication, Acharya:2019}. Similarly, private mean estimation has been well studied in both central~\citep{Dwork:2006, HardtT10} and local models~\citep{duchi2018minimax, DuchiR19, bhowmick2019protection} of privacy. These works have focused on providing example-level privacy (rather than user-level) in settings with homogeneous data, i.e., i.i.d.~samples. 

\cite{LiSYKR20} recently studied the problem of learning discrete distributions in the homogeneous cases (same distribution and same number of samples per user) with user-level differential privacy, and ~\cite{LevySAKKMS21} extended such results to other statistical tasks. These works also consider the setting with different number of samples per user although only via a reduction to same number of samples by discarding the data of users that have less than the median number of samples and effectively only using the median number of samples from all the other users. This approach can be asymptotically suboptimal for many natural distributions of $k_i$'s and is also likely to be worse in practice. Previously,~\cite{mcsherry2009differentially} showed how to build a (user-level) differentially private recommendation system, and~\cite{McMahanRTZ18} showed how to train a language model with user-level differential privacy.

User-level differential privacy in the context of heterogeneous data distributions has been studied in the constant $k_i$ setting \cite{Ozkara:2022}. Much of the complexity in our setting arises from variation in the $k_i$ values, which makes it challenging to maintain user-level privacy while leveraging the additional data points from users with a large number of data points. 

The challenges to optimization due to data heterogeneity have also been studied; \cite{ZhouC17, hanzely2020federated}, and~\citet{EichnerKMST19} study the approach of using different models for different groups from a convex optimization point-of-view.

Mathematically, similar issues are addressed in meta-analysis \citep{borenstein2021introduction, wiki:meta-analysis}, where the heterogeneity comes from different studies instead of different users. The non-private approach of inverse variance weighting that we recap in Section 3 is standard in that context.

\section{Model and Preliminaries}\label{sec:modelsandprelims}

Let $\Dee$ be a distribution on $[0,1]$ with (unknown) mean $p$ and variance $\sigma_p^2$. We assume a population of $n\in\mathbb{N}$ users, where each user $i \in [n]$ has a hidden variable $p_i\sim \Dee$ and $k_i \in\mathbb{N}$ samples $x_i^1, \ldots, x_i^{k_i} \sim_{i.i.d.} Ber(p_i)$. That is, the samples of user $i$ are i.i.d.~from a Bernoulli distribution with parameter $p_i$, which we will denote $\Dee_i=$Ber$(p_i)$. Assume without loss of generality that individuals are sorted by their $k_i$, so that $k_1 \geq \cdots \geq k_n$. The hidden variables $p_i$ of each user are unknown to the analyst. In the non-private setting, the samples $x_i^j$ and $k_i$ will be accessible to the analyst. In the private setting, access to these data is constrained.

The analyst's goal is to estimate the population mean $p$ with an estimator of minimum variance in a manner that is differentially private with respect to user data ($p_i$ and $\{x_i^j\}$).  
Each user provides their own estimate of their $p_i$ to the analyst based on their data $x_i$:
$\widehatpi = \frac{1}{k_i} \sum_{j=1}^{k_i} x_i^j$.
The analyst can then aggregate these (possibly along with other information) into her estimate of $p$. 

Let us first give some intuition for the distribution of these $\widehatp_i$. Let $\bersumP$ be the distribution that first samples $p_i\sim \Dee$, then samples $x_1, \cdots, x_k\sim Ber(p_i)$ and finally outputs $\widehat{p}_i=\frac{1}{k}\sum_{i=1}^k x_i$. The following lemma (proven in Appendix~\ref{appendix:modelsandprelims}) shows that the variance of $\widehat{p}_i$ is larger than $\sigma_p^2$ and transitions from $p(1-p)$ to $\sigma_p^2$ as $k$ increases (equivalently as $\widehatp_i$ concentrates around $p_i$). 

\begin{restatable}{lemma}{sigmai}
%\begin{lemma}
\label{lem.sigi} 
For all distributions $\Dee$ supported on $[0,1]$ with mean $p$ and variance $\sigp$,
$\sigma_p^2\le p(1-p)$. Further, $\mathbb{E}[\bersumP] = p$ and $\var(\bersumP) = \frac{1}{k}p(1-p)+\left(1-\frac{1}{k}\right)\sigp$.
%\end{lemma}
\end{restatable}

We assume that $k_i$ and $p_i$ are independent, so the amount of data an individual has is independent of her data distribution. This is crucial for the problem setup: in order for learning from the heterogeneous population to be advantageous, there must a common meta-distribution is shared across all individuals in the population, rather than a meta-distribution only for each fixed $k_i$.
If $k_i$ and $p_i$ can be arbitrarily correlated, then the meta-distribution for each value of $k_i$ can be different. Hence, the best solution in that setting is to learn on each sub-population (where the sub-populations are defined by their value of $k_i$) separately. 
While this assumption is natural in some settings, it is unlikely to hold in others -- for example, different writing styles that are more or less verbose. In future work, it may be interesting to explore how various heterogeneity assumptions affect learning algorithms.

\subsection{Differential Privacy}\label{prelim.dp}

Differential privacy (DP) \citep{Dwork:2006} informally limits the inferences that can be made about an individual as a result of computations on a large dataset containing their data. This privacy guarantee is achieved algorithmically by randomizing the computation to obscure small changes in the dataset. 
The definition of differential privacy requires a \emph{neighbouring relation} between datasets. If two datasets $D$ and $D'$ are neighbours under the neighbouring relation, then differences between these two datasets should be hidden by the private algorithm. 

\begin{definition}[ $(\epsilon,\delta)$-Differential Privacy \citep{Dwork:2006}]\label{def:DP} Given $\epsilon\ge0$, $\delta\in[0,1]$ and a neighbouring relation $\sim$,
a randomized mechanism $\mathcal{M}:\mathfrak{D}~\rightarrow~\mathcal{Y}$ from the set of datasets to an output space $\mathcal{Y}$ is $(\epsilon,\delta)$-\emph{differentially private} if for all neighboring datasets $D\sim D' \in \mathfrak{D}$, and all events $E\subseteq\mathcal{Y}$,
\begin{align*} \label{def:dp-with-inputs}
    &
    \Pr[\mathcal{M}(D) \in E]
    \leq e^\epsilon \cdot\Pr[\mathcal{M}(D') \in E] + \delta,
\end{align*}
where the probabilities are taken over the random coins of $\mathcal{M}$. When $\delta=0$, we may refer to this as \emph{$\epsilon$-differential privacy}.
\end{definition}

When each user has a single data point, the neighbouring relation is typically defined as: $D$ and $D'$ are neighbours if they differ on the data of a single individual, i.e., a single data point.
In our setting where users have multiple data points, we must distinguish between \emph{user-level} and \emph{event-level} DP.  The former considers $D$ and $D'$ neighbours if they differ on all data points associated with a single user, whereas the latter considers $D$ and $D'$ neighbours only if they differ on a \emph{single} data point, regardless of the number of data points contributed by that user. Naturally, user-level DP provides substantially stronger privacy guarantees, and is often more challenging to achieve from a technical perspective. In this work, we will provide user-level DP guarantees. 

Further, when defining user-level DP where users have heterogeneous quantities of data, we also need to distinguish between settings where the number of data points held by each user is protected information, and settings where it is publicly known. We'll refer to the former as \emph{\unknownki user-level differential privacy}, where the entry that differs between neighboring databases can have arbitrarily different number of data points, and the latter as \emph{\knownki user-level differential privacy}, where the amount of data held by each user is the same in neighboring databases. Formally, let $D_i=\{x_i^1, \cdots, x_i^{k_i}\}$ be the data of user $i$ for each $i \in [n]$. For \unknownki user-level differential privacy, we say $D$ and $D'$ are neighbours if there exists an index $i$ such that for all $j\in[n]\backslash\{i\}$, $D_j=D_j'$. For \knownki user-level differential privacy, we say $D$ and $D'$ are neighbours if they are neighbours under \unknownki user-level differential privacy and additionally $|D_i|=|D_i'|$ for all $i\in[n]$.

One standard tool for achieving $\epsilon$-differential privacy is the \emph{Laplace Mechanism}. For a given function $f$ to be evaluated on a dataset $D$, the Laplace Mechanism first computes $f(D)$ and then adds Laplace noise which depends on the \emph{sensitivity} of $f$, defined for real-valued functions as \[\Delta f = \max_{D,D' \text{ neighbors}} |f(D) - f(D')|.\] The Laplace Mechanism outputs $\mathcal{M}_L(D,f,\epsilon) = f(D) + \Lap(\Delta f/\epsilon)$, and is $(\epsilon,0)$-DP.

Differential privacy satisfies \emph{robustness to post-processing}, meaning that any function of a DP mechanism will retain the same privacy guarantee. DP also \emph{composes adaptively}, meaning that if an $(\epsilon_1,\delta_1)$-DP mechanism and and an $(\epsilon_2,\delta_2)$-DP mechanism are both applied to the \emph{same dataset}, then the entire process is $(\epsilon_1+\epsilon_2,\delta_1+\delta_2)$-DP. \emph{Parallel composition} of DP mechanisms says that if DP mechanisms are applied to disjoint datasets, then composition is not required. 
That is, if an $(\epsilon_1,\delta_1)$-DP mechanism and and an $(\epsilon_2,\delta_2)$-DP mechanism are each applied to  \emph{disjoint datasets}, then the entire process is $(\max\{\epsilon_1,\epsilon_2\},\max\{\delta_1,\delta_2\})$-DP with respect to both datasets together.

\section{A Non-Private Estimator}\label{s.nonpriv}

We begin by illustrating the procedure for computing an optimal estimator $\widehatp$ in the non-private setting. The general structure of the estimator will be the same in both the private and non-private settings. The analyst will compute the population-level mean estimate $\widehatp$ as a weighted linear combination of the user-level estimates $\widehatpi$.\footnote{In the non-private setting, this restriction is without loss of generality since the optimal estimator takes this form. In the private setting this is still near-optimal; see Section \ref{s.optpideal} for more details.} 
The key question is how to derive the weights so that individuals with more reliable estimates (i.e., larger $k_i$) have more influence over the final result. 

\begin{algorithm}
\caption{Non-private Heterogeneous Mean Estimation $\pnprealistic$}\label{alg.np}
 \textbf{Input:} number of users $n$, number of samples held by each user $(k_1,\ldots, k_n \; s.t. \; k_i\ge k_{i+1}$), user-level estimates $(\widehatp_1, \cdots, \widehatp_n)$.
\begin{algorithmic}[1]
\vspace{0.1in}
\State \textbf{\underline{Initial Estimates}}
\State $\pnpinitial = \sum_{i=9n/10}^n x_i^1$ \Comment{Initial mean estimate}
\State $\widehatsigp = \tfrac{1}{\log n(\log n -1)}\textstyle\sum_{i,j\in[\log n]}(\widehatp_i-\widehatp_j)^2$ \Comment{Initial variance estimate}
\vspace{0.1in}
\State \textbf{\underline{Defining weights}}
\For{$i=\log n$ to $9n/10$}
\State\label{eq.sighati} Compute
$\widehatsigi = \tfrac{1}{k_i} (\pnpinitial - (\pnpinitial)^2) + (1-\tfrac{1}{k_i})\widehatsigp.$ \Comment{Estimate individual variances}
\State $\widehat{w_i} = \tfrac{1/\widehat{\sigma}_i^2}{\sum_{j=\log n}^{9n/10} 1/\widehat{\sigma}_j^2}$ \Comment{Compute normalised weights}
\EndFor
\vspace{0.1in}
\State \textbf{\underline{Final Estimate}}
\State
\Return $\pnprealistic = \textstyle\sum_{i=\log n}^n \widehat{w_i} \widehatp_i$ \Comment{Final estimate}
\end{algorithmic}
\end{algorithm}

Let $\sigma_i^2$ be the variance of $\widehatp_i$. 
In an idealized setting where the $\sigma_i^2$ are all known, the analyst can minimize the variance of the estimator by weighting each user's estimate $\widehatpi$ proportionally to the inverse variance of their estimate. The weights are then normalised to ensure the estimate is unbiased. This approach yields the following estimator, which is optimal in the non-private setting \citep{Hartung:2008}:
\begin{equation}\label{nonprivideal}
\pnpideal = \textstyle\sum_{i=1}^n w_i^* \widehatp_i \text{   where  } w_i^* = \tfrac{1/\sigma_i^2}{\sum_{j=1}^n 1/\sigma_j^2}.
\end{equation}
In practice, the $\sigma_i^2$s are unknown, so the analyst must rely on estimates to assign weights. Fortunately, the user-level variance $\sigma_i^2$ can be expressed as a function of $k_i$ and the population statistics $p$ and $\sigp$, as shown in Lemma~\ref{lem.sigi}:
\begin{equation}\label{formulaforsigmai}
\sigma_i^2 = \tfrac{1}{k_i} (p - p^2) + (1-\tfrac{1}{k_i})\sigp.
\end{equation}
Now, $p$ and $\sigp$ are also unknown but since they are population statistics, we can use simple estimators to obtain initial estimates. These initial statistics can then be used to define the weights, resulting in a refined estimate of the mean $p$. 
Specifically, as outlined in Algorithm~\ref{alg.np}, we split users into three groups. The $\log n$ individuals with the most data are used to produce an estimate of $\var(\bersumPi{k_{\log n}})$, which serves as a proxy for $\sigp$. The $1/10$th of individuals with the least data are used to produce an initial estimate of the mean $p$. The remaining $9n/10-\log n$ individuals are used to produce the final estimate.
We split the individuals into separate groups to ensure the initial estimates and the final estimate are independent so we can easily obtain variance bounds on the final estimate. The specific sizes of the three groups are heuristic; the exact fraction $1/10$ is not necessary.
Under some mild conditions on $\Dee$, and if $n$ is large enough, the error incurred by $\pnprealistic$ is within a constant factor of the error incurred by the ideal estimator $\pnpideal$.\footnote{This can be observed by viewing the non-private setting as a simplified version of the setting studied in Section \ref{s.optpideal}, which proves near-optimality of (truncated) linear estimators for this problem.}

\section{A Framework for Private Estimators}\label{s.privest}

We now turn to our main result, which is a framework for designing differentially private estimators for the mean $p$ of the meta-distribution $\Dee$. 
We discussed in Section~\ref{s.nonpriv} the need for initial estimates of $p$ and $\sigp$ to weight the contributions of the users. In the non-private setting, there are canonical, optimal choices of these estimators; the empirical mean and empirical variance. In the private setting, these choices are not canonical, and different estimators may perform better in different settings. There is a considerable literature exploring various mean and variance estimators for the homogeneous, single-data-point-per-user setting. As such, we leave the choice of the specific initial mean and variance estimators as parameters of the framework. This allows us to focus on the nuances of the heterogeneous setting, not addressed in prior work.
In Section~\ref{instantiation}, we give a specific pair of private mean and variance estimators that provably perform well in our framework.

We will define three estimators: a ideal estimator $\pideal$ (only implementable if all the $\sigma_i^2$ are known), and a realisable estimator $\prealistic$ in the \knownki user-level DP setting, and a realisable estimator $\prealisticunknown$ in the \unknownki user-level DP setting.
The main result in the \knownki user-level DP setting (Theorem \ref{metatheorem}) is that under some mild conditions and assuming $n$ is sufficiently large, there exists an $(\epsilon, \delta)$-DP estimator $\prealistic$ (Algorithm \ref{alg.dp}) such that for some constant $C$,
\[\var(\prealistic)\le C\cdot \var(\pideal).\] 
In Section~\ref{s.privatek}, we extend this result to the case where $k_i$s are private and unknown to the analyst. We will maintain the optimality of the estimator (up to logarithmic factors), under slightly more restrictive conditions (Theorem~\ref{thm.privatek}).

\subsection{The Complete Information Private Estimator}

As in Section~\ref{s.nonpriv}, we begin with a discussion of the ideal estimator if the $\sigma_i$ were known. This ideal private estimator $\pideal$ has a similar form to $\pnpideal$ with some crucial differences. The first main distinction is that Laplace noise is added to achieve DP, where 
the standard deviation of the noise must be scaled to the sensitivity of the statistic. A natural solution would be to add noise directly to the non-private estimator $\pnpideal$, but the sensitivity of this statistic is too high. In fact, the worst case sensitivity of $\pnpideal$ is 1, which would result in the noise that completely masks the signal. Thus, the first change we make is to limit the weight of any individual's contribution by setting 
\[w_i = \tfrac{\min\{1/\sigma_i^2, T/\sigma_i\}}{\sum_{j=1}^n \min\{1/\sigma_j^2, T/\sigma_j\}}\]
for some truncation parameter $T$. Analogous to the weights used in Section \ref{s.nonpriv}, this choice of $w_i$ is still inversely proportional to $\sigma^2_i$ up to an upper limit that depends on the truncation parameter $T$, and then normalized to ensure the weights sum to 1 so the estimator is unbiased.
Intuitively, the parameter $T$ controls the trade-off between variance of the weighted sum of individual estimates (which is minimized by assigning high weight to low variance estimators) and variance of the noise added for privacy (which is minimized by assigning roughly equal weight to all users).

We make one final modification to lower the sensitivity of the statistic. Inspired by the Gaussian mean estimator of \cite{Karwa:2018}, we truncate the individual contributions $\widehatpi$ into a sub-interval of $[0,1]$. The truncation intervals $[a_i, b_i]$ are chosen to be as small as possible (to reduce the sensitivity and hence the noise added for privacy), while simultaneously ensuring that $\widehat{p}_i\in[a_i,b_i]$ with high probability (to avoid truncating relevant information for the estimation). In order to achieve this, we need a tail bound on the distribution $\Dee$. To maintain generality for now, we assume there exists a known function $\concentrationbound{k}{\Dee}{n}{\sigma_p^2}{\beta}$ that gives high-probability concentration guarantees of $\widehatpi$ around $p$, and is defined such that
\[\Pr\left(\forall i, |\widehatp_i-p|\le \concentrationbound{k_i}{\Dee}{n}{\sigma_p^2}{\beta}\right)\ge 1-\beta.\] Appendix \ref{s.truncest} presents a more detailed discussion of the structure of these concentration functions and how they may be estimated if they are unknown to the analyst.

We can now describe the full information, or \emph{ideal} estimator $\pideal$:
\begin{equation}\label{privoptimal1}
\pideal = \textstyle \sum_{i=1}^n w_i^* [\widehatpi]_{a_i}^{b_i}  + \Lap(\tfrac{\max_i w_i^*|b_i-a_i|}{\epsilon}), 
\end{equation}
where $[\widehat{p_i}]_{a_i}^{b_i}$ denotes the projection of $\widehat{p_i}$ onto the interval $[a_i, b_i]$ and
\begin{align}\label{idealtruncation}
a_i &= p-\concentrationbound{k_i}{\Dee}{n}{\sigma_p^2}{\beta}, \quad b_i = p+\concentrationbound{k_i}{\Dee}{n}{\sigma_p^2}{\beta},\;\text{ and }\; w_i^* = \tfrac{\min\{1/\sigma_i^2, T^*/\sigma_i\}}{\sum_{j=1}^n \min\{1/\sigma_j^2, T^*/\sigma_j\}}.
\end{align}
We would like to choose the truncation parameter $T^*$ to minimise the variance of the resulting estimator: 
\begin{align}\label{eq.varideal}
\var(\pideal)
&= \textstyle \sum_{i=1}^n (w_i^*)^2\var([\widehatpi]_{a_i}^{b_i})+\max_i\tfrac{(w_i^*)^2|b_i-a_i|^2}{\epsilon^2}.
\end{align}
Although we do not know $\var([\widehatpi]_{a_i}^{b_i})$ exactly, we do know that $[\widehatpi]_{a_i}^{b_i}=\widehatpi$ with high probability, and thus we can approximate $\var([\widehatpi]_{a_i}^{b_i})$ with $\sigma_i^2$. Throughout the remainder of the paper, we will assume that $\beta$ is chosen such that $\frac{1}{2}\sigma_i^2\le\var([\widehatpi]_{a_i}^{b_i}).$ Thus, we will approximate the optimal truncation parameter by
\begin{align}\label{idealthreshold}
T^*&= \arg\min_T \sum_{i=1}^n (w_i^*)^2\sigma_i^2+\max_i\frac{(w_i^*)^2|b_i-a_i|^2}{\epsilon^2} \notag \\
&=\arg\min_T \tfrac{1}{{(\sum_{j=1}^n \min\{1/\sigma_j^2, T/\sigma_i\})^2}}(\textstyle\sum_{i=1}^n\min\{1/\sigma_i^2, {T^2}\} +\max_i \tfrac{\min\{1/\sigma_i^4, T^2/\sigma_i^2\}|b_i-a_i|^2}{\epsilon^2} ).
\end{align}

We'll show in Section~\ref{s.optpideal} that under some conditions on the Fisher information of $\bersumP$, $\pideal$ is optimal up to logarithmic factors among all private unbiased estimators for heterogeneous mean estimation. 

\begin{example}
As a simple example, suppose that $p \in (\frac 1 3, \frac 2 3)$, $\sigma_p = 1/\sqrt{n}$, and $k_i = \lceil \frac n i\rceil$. In this case, an asymptotically optimal non-private estimator averages all the $\sum k_i = O(n \log n)$ available samples. It can be shown that this gives us an unbiased estimator with standard deviation $\Theta(\frac{1}{\sqrt{n \log n}})$. A naive sensitivity-based noise addition method will give us privacy error $O(\frac{1}{\varepsilon \log n})$, since the weight of the first user in this average is $\Theta(1/\log n)$. Our truncation-based algorithm will truncate the $i$th user's contribution to a range of width $\sqrt{\frac{\log n}{k_i}} \approx \sqrt{\frac{i\log n}{n}}$. Applying our algorithm would then give us privacy error $\Theta(\frac{1}{\varepsilon \sqrt{n \log n}})$. In other words, for constant $\varepsilon$, privacy does not have an asymptotic cost. We remark that in this case, any uniform weighted average will incur asymptotically larger standard deviation $\Omega(\frac{1}{\sqrt{n}})$.
\end{example}

\subsection{Realizable Private Heterogeneous Mean Estimation}\label{s.finalest}

 Our goal in this section is to design a realizable estimator $\prealistic$ that is competitive with the ideal estimator $\pideal$.  As in the non-private setting, we divide the individuals into three groups. The first group, consisting of the $n/10$ individuals with the lowest $k_i$ will be used to compute the initial mean estimate $\pinitial$. The $\variancesamplesize$ individuals with the largest $k_i$ will be used to compute the initial variance estimate $\widehatsigp$. These will respectively be computed using private subroutines $\meanest$ and $\varianceest$, which each provide event-level DP, as they each operate on only a single point from each user. These initial estimates will be plugged into expressions to compute $\widehatsigi$, $\widehat{a_i}$, and $\widehat{b_i}$ for the remaining individuals $\variancesamplesize+1\le i\le 9n/10$. As in the non-private setting, the specific sizes of these groups are heuristic. The important thing is that the size of the first two groups are large enough that the resulting mean and variance estimates are sufficiently accurate, and the last group contains $\Theta(n)$-users whose $k_i$ is above the median.
 
 Since the estimate $\pinitial$ used in $\widehat{a_i}$ and $\widehat{b_i}$ may have additional error up to $\alpha$ (which will depends on the additive accuracy guarantee of $\meanest$), we shift these estimates by an additive $\alpha$ to account for this error. Next, all of these intermediate estimates and the user-level mean estimates $\widehatpi$ from users $\variancesamplesize+1\le i\le 9n/10$ will be used to compute the optimal weight cutoff $\widehat{T}^*$, the optimal weights $\widehat{w}_i^*$ for each user $L +1\le i\le 9n/10$, and finally the estimator $\prealistic$ as a weighted sum of the truncated user-level estimates $[\widehatpi]_{\widehat{a}_i}^{\widehat{b}_i}$ plus Laplace noise. This procedure is presented in full detail in Algorithm \ref{alg.dp}. 

\begin{algorithm}[tbh]
\caption{Private Heterogeneous Mean Estimation $\prealistic$
}\label{alg.dp}
 \textbf{Input parameters:} privacy parameters $\eps>0$, $\delta\in[0,1]$, desired high probability bound $\beta\in[0,1]$, number of users $n$, an $(\eps,\delta)$-DP mean estimator $\meanest$, error guarantee on $\meanest$ $\alpha>0$, an $(\eps,\delta)$-DP variance estimator $\varianceest$, number of samples for variance estimator $\variancesamplesize$, and number of samples held by each user $(k_1,\ldots, k_n \; s.t. \; k_i\ge k_{i+1}$).\\
 \textbf{Input data:} User-level estimates $(\widehatp_1, \cdots, \widehatp_n)$
\begin{algorithmic}[1]
\State \textbf{\underline{Initial Estimates}}
\State $\pinitial = \meanest(x_{9n/10 + 1}^1, \cdots, x_n^1)$ \Comment{Initial mean estimate}
\State $\widehatsigp = \varianceest(\widehatp_1, \cdots, \widehatp_{\variancesamplesize})$ \label{varianceestimaterequired} \Comment{Initial variance estimate}
\vspace{0.1in}
\State \textbf{\underline{Defining weights and truncation}}
\For{$i=\variancesamplesize + 1$ to $9n/10$}
\State\label{hatsigi} Compute
$\widehatsigi = \tfrac{1}{k_i} (\pinitial - (\pinitial)^2) + (1-\tfrac{1}{k_i})\widehatsigp.$ \Comment{Estimate individual variances}
\State $\widehat{a_i} =  \pinitial-\alpha-\concentrationbound{k_{i}}{\Dee}{n}{\widehat{\sigma_p^2}}{\beta}$
\State\label{settingintervalbi1} $\widehat{b_i} =  \pinitial+\alpha+\concentrationbound{k_{i}}{\Dee}{n}{\widehat{\sigma_p^2}}{\beta}$
\Comment{Estimate truncation parameters}
\EndFor
\State\label{settrunc} $\widehat{T}^*=\arg\min_T \frac{\sum_{i=\variancesamplesize+1}^{9n/10}\min\{\frac{1}{\widehatsigi}, {T^2}\}+\max_{{\variancesamplesize+1}\le i\le {9n/10}} \frac{\min\{1/\widehat{\sigma_i}^4, {T^2/\widehatsigi}\}|\widehat{b_i}-\widehat{a_i}|^2}{\epsilon^2}}{{(\sum_{i=\variancesamplesize+1}^{9n/10} \min\{1/\widehat{\sigma}_j^2, T/\widehat{\sigma}_i\})^2}}$
\\\Comment{Compute weight truncation}
\For{$i=\variancesamplesize + 1$ to $9n/10$}
\State\label{setweight} $\widehat{w_i}^* = \frac{\min\{1/\widehatsigi, \widehat{T}^*/\widehat{\sigma_i}\}}{\sum_{j=\variancesamplesize+1}^{9n/10} \min\left\{1/\widehat{\sigma}_j^2, \widehat{T}^*/\widehat{\sigma}_i\right\}}$ \Comment{Compute weights}
\EndFor
\vspace{0.1in}
\State \textbf{\underline{Final Estimate}}
\State\label{sensitivity} $\sensitivity=  \max_{i\in[\variancesamplesize +1, 9n/10]}\frac{\min\{1/\widehatsigi, \widehat{T}^*/\widehat{\sigma_i}\}|\widehat{b_i}-\widehat{a_i}|}{\sum_{j=\variancesamplesize+1}^{9n/10} \min\left\{1/\widehat{\sigma}_j^2, \widehat{T}^*/\widehat{\sigma}_i\right\}}$ \Comment{Compute sensitivity}
\State Sample $Y \sim \Lap\left(\frac{\sensitivity}{\epsilon}\right)$ \Comment{Sample noise added for privacy}
\State\label{finalsum} \textbf{return} $\prealistic = \sum_{i=\variancesamplesize+1}^{9n/10} \widehat{w_i}^* [\widehatpi]_{\widehat{a_i}}^{\widehat{b_i}}  + Y$ \Comment{Final estimate}
\end{algorithmic}
\end{algorithm}

For the remainder of this section, we turn to establishing the accuracy requirements of $\meanest$ and $\varianceest$ that ensure that the variance of $\prealistic$ is within a constant factor of the variance of $\pideal$.

\begin{restatable}{theorem}{metatheorem}
%\begin{theorem}
\label{metatheorem}
For any $\eps>0$, $\delta\in[0,1]$, $\alpha>0$, $\beta\in[0,1]$, $n\in\mathbb{N}$, $0\le L\le 3n/5$, $(\eps,\delta)$-DP mean estimator $\meanest$, $(\eps,\delta)$-DP variance estimator $\varianceest$, and sequence $(k_1,\ldots, k_n \; s.t. \; k_i\ge k_{i+1}$), Algorithm~\ref{alg.dp} is $(\eps,\delta)$-DP. If,
\begin{itemize}
\item $\meanest$ is such that given $n/10$ samples from $\Dee$, with probability $1-\beta$, $|p-\pinitial|\le \concentrationbound{k_i}{\Dee}{n}{\sigma_p^2}{\beta}$ and $\pinitial(1-\pinitial)\in\left[\tfrac{1}{2}p(1-p), \tfrac{3}{2}p(1-p)\right]$,
\item $\varianceest$ is such that given $\variancesamplesize$ samples from $\bersumP$, with probability $1-\beta$, $\widehat{\sigma}_{p}^2\in\left[\var(\bersumPi{k}),~8\var(\bersumPi{k})\right]$,
\item the $k_i$s are such that $\tfrac{k_{1}}{k_{n/2}}\le\tfrac{n/2-\variancesamplesize}{\variancesamplesize}$,
\end{itemize}
then with probability $1-2\beta$, 
$\var(\prealistic)\le C\cdot\var(\pideal)$ for some absolute constant $C$. 
%\end{theorem}
\end{restatable}

The final assumption ensures that the $\variancesamplesize$ users with the most data can not estimate the mean of meta-distribution alone. In the setting where these $L$ users can give a very accurate estimate of the mean, we conjecture that there is little benefit in incorporating the data of the remaining users. If this assumption does not hold, then an estimator that better utilizes only the top log $n$ users may be optimal. The strictness of this condition depends on the sample complexity of estimating the variance of $\Dee(k)$.  We'll see in Section~\ref{s.popvar} that for well-behaved distributions like Gaussians, the sample complexity for obtaining a constant multiplicative approximation of $\var(\Dee(k))$ is $O(\log(1/\beta)/\epsilon)$. Thus for sufficiently well-behaved distributions, up to logarithmic factors, this condition simply requires that the number of data points held by the user with the most data is at most $n$ times the number of data points of the median user. If $n$ is large, then this is unlikely to be a limiting factor. 

The first two conditions of Theorem~\ref{metatheorem} ensure that the mean and variance estimates are sufficiently accurate to use in the remainder of the algorithm. Notice that the initial estimates do not need to be especially accurate. In fact, provided $p$ is not too close to 0 or 1, the DP mean estimator that simply adds noise to the sample mean achieves sufficient accuracy (see Lemma \ref{initialmeanestimate} for details). In Section~\ref{instantiation}, we also give a DP variance estimator that achieves the desired accuracy guarantee using only $L=\log n/\eps$ samples, under some mild conditions (Lemma \ref{multiplicativevariance}). Thus the set of mean and variance estimators that satisfy the accuracy requirements of Theorem \ref{metatheorem} are non-empty.
We note that the constants $1/2$
, $3/2$ and $8$ in Theorem~\ref{metatheorem} are not intrinsic; any constant multiplicative factors will suffice. We also note that the specific sizes of the three groups outlined in Algorithm~\ref{alg.dp} are heuristic and can be varied to ensure that the initial estimator achieves the required accuracy.

A full proof of Theorem \ref{metatheorem} is given in Appendix \ref{appendix.metathm}; we present intuition and a proof sketch here.

The main distinction between $\pideal$ and $\prealistic$ is the use of the output of the estimators $\meanest$ and $\varianceest$ to estimate $\sigma_i^2$, $a_i$ and $b_i$. Thus, the main component of the proof of Theorem~\ref{metatheorem} is to show that the conditions stated in the theorem are enough to ensure that $\widehat{\sigma_i}^2$, $\widehat{a}_i$ and $\widehat{b}_i$ are sufficiently accurate.

\begin{restatable}{lemma}{finalvariance}
%\begin{lemma}
\label{lem.finalvariance}
Given $\pinitial$, $\widehat{\sigma}_{p}^2$, and $k_i$, define $\widehat{\sigma_i}^2 = \frac{1}{k_i}\pinitial(1-\pinitial)+\frac{k_i-1}{k_i}\widehat{\sigma}_{p}^2$.
Under the conditions of Theorem~\ref{metatheorem}, for all $i>\variancesamplesize$, we have 
 $\widehat{\sigma_i}^2\in\left[\frac{1}{2}\sigma_i^2, 9.5\sigma_i^2\right]$ and $|\widehat{b}_i-\widehat{a}_i|\leq 4|b_i-a_i|$.
%\end{lemma}
\end{restatable}

A detailed proof of Lemma~\ref{lem.finalvariance} is presented in Appendix~\ref{appendix.metathm}. Lemma~\ref{lem.finalvariance} implies that the individual variance estimates used in the weights, and the truncation parameters are accurate up to constant multiplicative factors. The main ingredient left then is to show that using only a subset of the population in the final estimate only affects the performance up to a multiplicative factor.
Under the assumption that $\frac{k_{\max}}{\medk}\le\frac{n/2-\variancesamplesize}{\variancesamplesize}$, where $\sigma_{k_{\max}}^2 = \var(\widehatp_1)$ and $\sigma_{\medk}^2 = \var(\widehatp_{n/2})$ then \begin{align}\label{limitonk}
\sigma_{\medk}^2 &= \tfrac{1}{\medk}p(1-p)+(1-\tfrac{1}{\medk})\sigma_p^2 \notag\\
&\le \tfrac{n/2-\variancesamplesize}{\variancesamplesize}\tfrac{1}{k_{\max}}p(1-p)+(1-\tfrac{1}{k_{\max}})\sigma_p^2\notag \\
&\le \tfrac{n/2-\variancesamplesize}{\variancesamplesize} \sigma_{k_{\max}}^2.\end{align}
We use this to show that for any truncation parameter $T$, \[\textstyle \sum_{i=1}^n \min\{\tfrac{1}{\sigma_i^2}, \tfrac{T}{\sigma_i}\} \le 4 \textstyle \sum_{i=\variancesamplesize+1}^{9n/10} \min\{\tfrac{1}{\sigma_i^2}, \tfrac{T}{\sigma_i}\}.\]
Using this, along with the bounds on estimated quantities from Lemma \ref{lem.finalvariance}, we show that with high probability, the variance of the our estimator $\prealistic$ is within a constant factor of $\var(\pideal)$, as given in Equation \eqref{eq.varideal}:
\begin{align}
\var(\prealistic) &= \tfrac{\sum_{i=\variancesamplesize+1}^{9n/10}\min\{\tfrac{1}{\widehat{\sigma_i}^4}, \tfrac{\widehat{T}^{*2}}{\widehat{\sigma_i}^2}\}\sigma_i^2+\max_{i} \tfrac{\min\{\frac{1}{\widehat{\sigma_i}^4}, {\tfrac{\widehat{T}^{*2}}{\widehat{\sigma_i}^2}}\}|\widehat{b_i}-\widehat{a_i}|^2}{\epsilon^2}}{(\textstyle\sum_{j=\variancesamplesize+1}^{9n/10} \min\{1/\widehat{\sigma_j}^2, \tfrac{\widehat{T}^*}{\widehat{\sigma_i}}\})^2} \\
\nonumber &\leq O(\var(\pideal)).
\end{align}

We remark that this framework is amenable to being performed in a federated manner if one has private federated mean and variance estimators. Steps~\eqref{hatsigi}~-~\eqref{settingintervalbi1} and Step~\eqref{setweight} can be performed locally. Steps~\eqref{settrunc} and the final sum in Step \eqref{finalsum} would need to be altered to fit the federated framework. We will see in Section \ref{s.privatek} that it is sufficient to replace 
Step~\eqref{settrunc} with an estimate of $\frac{1}{\sigma_{\variancesamplesize}}$ (the inverse standard deviation of the user with the $\variancesamplesize$-th most data). The final step is then a simple addition with output perturbation, which can be performed in a federated manner (e.g., \cite{mcmahan17a, Kairouz21}).

\subsection{Special Case: The constant $p_i$ case.}\label{s.constantpi}

In the previous section, we considered the setting where there was heterogeneity in both the users' distributions (i.e., the $p_i$s were not constant), as well as the number of data points that they each held (i.e., the $k_i$s were not constant). In the absence of variation in the $p_i$, each user is sampling from the same distribution $\Ber(p)$. When privacy is not a concern, this setting reduces to the single-data-point-per-user setting where the sample size is increased to $\sum_{i=1}^n k_i$. However, under the constraint of user-level differential privacy, this setting is distinct from the single-data-point-per-user setting, since we need to protect the entirety of each users data set. In fact, much of the complexity of Algorithm~\ref{alg.dp} is required even in this simpler case. In particular, the truncated inverse variance weighting is still required in this case when there is variation in the $k_i$. In fact, the only step of Algorithm~\ref{alg.dp} that is not required is Step~\ref{varianceestimaterequired}, since already know that $\sigma_p^2=0$. Since there is no variance in $\Dee$, the high probability bound $\concentrationbound{k_{i}}{\Dee}{n}{\widehat{\sigma_p^2}}{\beta}$ is just due to the randomness in the binomial distribution $\Bin(k_i,p)$, which comes from averaging $k_i$ samples drawn from $\Ber(p)$.

When $\sigma_p^2=0$, $\sigma_i$ has the simple formula $\sigma_i=\frac{\sqrt{p(1-p)}}{k_i}$ and we can directly translate from the truncation threshold $T$ on $\sigma_i$ to a truncation threshold $k$ on $k_i$, $T=\frac{\sqrt{p(1-p)}}{k}$.
Further, if we assume that all the $k_i$ are large enough ($\min k_i\ge 2\ln(1/\delta)/p$) then we also have the simple formula $\concentrationbound{k_{i}}{\Dee}{n}{\widehat{\sigma_p^2}}{\beta}=\sqrt{\frac{3p\ln(2/\beta)}{k_i}}$. We can plug these into Equation~\eqref{idealthreshold} (recall that $T^*$ is defined as the truncation threshold that minimizes the variance of $\pideal$) to obtain the following formula for the variance of $\pideal$, and hence the variance of $\prealistic$ is: 
\begin{align}\label{constantD}
\min_k \tfrac{p(1-p)\textstyle\sum_{i=1}^n\min\{k_i, k\} +6p\ln(2/\beta) \tfrac{k}{\epsilon^2}}{{(\sum_{j=1}^n \min\{k_i, \sqrt{k_ik}\})^2}}.
\end{align}

Even in the private setting, one can reduce to the single-data-point-per-user setting by reducing the sample size by a factor of 2, and forcing the $n/2$ users with the most data points to produce their estimate $\hat{p_i}$ using only $\medk$ (the median $k_i$) data points. Then each estimate $\hat{p_i}$ is a sample from the same distribution and we can compute their mean. To the best of our knowledge, all the prior work in the private literature that handles variations in $k_i$ follows this formula. However, not only does this algorithm reduce the sample size by a factor of 2, it also unnecessarily hinders the contribution of users with many data points. As a simple example, suppose that all the users have a single data point, except for $\sqrt{n}$ users, which have $n$ data points. Then the algorithm which forces $n/2$ of the users to use the median number of data points has an error rate of $\Theta(\frac{1}{n}+\frac{1}{n^2\eps^2})$ assuming that $p$ is bounded away from 0 or 1. Letting $k=n$ in Equation~\ref{constantD} implies that that the truncated inverse variance weighted algorithm in the previous section is better able to utilise the data of the users with high $k_i$s, resulting in an error rate of
$O(\frac{1}{n^{3/2}}+\frac{1}{n^2\eps^2})$.

\subsection{Extension: \unknownki user-level differential privacy setting}\label{s.privatek}

Let us now turn to our problem in the \unknownki user-level differential privacy setting, where the $k_i$s are considered private and require formal privacy protections. We will need to add considerably more machinery to Algorithm~\ref{alg.dp} to make it private under this stronger notion of privacy. Under \knownki user-level privacy, the quantities $\hat{T}^*$ (the weight truncation parameter) and $\Lambda$ (the sensitivity of the final estimate) in Algorithm~\ref{alg.dp} do not pose privacy concerns since they only depend on the private data $\widehatpi$ through the $\pinitial$ and $\widehatsigi$, which are both produced differentially privately. However, both these quantities depend on the $k_i$ directly, and hence care needs to be taken when using them under \unknownki user-level DP.

In Algorithm~\ref{alg.dpk}, we outline the extension of Algorithm~\ref{alg.dp} to satisfy \unknownki user-level differential privacy. It is different to Algorithm~\ref{alg.dp} in two main ways: the method for truncating the weights and the method for computing the scale of the noise needed to maintain privacy.

\begin{algorithm}[t]
\caption{Private Heterogeneous Mean Estimation $\prealisticunknown$}\label{alg.dpk}
 \textbf{Input parameters:} Privacy parameters $\eps>0$, $\delta\in[0,1]$, desired high probability bound $\beta\in[0,1]$, number of users $n$, an $(\eps,\delta)$-DP mean estimator $\meanest$, error guarantee on $\meanest$ $\alpha>0$, an $(\eps,\delta)$-DP variance estimator $\varianceest$, number of samples for variance estimator $\variancesamplesize$, an upper bound on the total number of data points held by a single user $k_{\max}$, an $\eps$-DP estimator of the $\ell$th order statistic $\expmech(\cdot; \ell, k_{\max})$.\\
 \textbf{Input data:} Number of samples held by each user $(k_1,\ldots, k_n \; s.t. \; k_i\ge k_{i+1}$), and user-level estimates $(\widehatp_1, \cdots, \widehatp_n)$.
\begin{algorithmic}[1]
\vspace{0.1in}
\State \textbf{\underline{Initial Estimates}}
\State $\pnpinitial = \meanest(x_{9n/10 +1}^1, \cdots, x_n^1)$\label{initialk} \quad \Comment{Initial mean estimate}
\State $\widehatsigp = \varianceest(\widehatp_1, \cdots, \widehatp_{\variancesamplesize})$\label{initialvark} \quad \Comment{Initial variance estimate}
\vspace{0.1in}
\vspace{0.1in}
\State \textbf{\underline{Compute Sensitivity Proposal}}
\State $\widehat{k_{\variancesamplesize}}=\expmech(k_1, \cdots, k_n; \variancesamplesize, k_{\max})$ \label{EMk} \quad \Comment{Compute $L$-th order statistic}
\For{$i\in[\variancesamplesize+1, 9n/10]$ }
\State $\widetilde{k_i}=\min\{k_i,\widehat{k_{\variancesamplesize}}\}$ \label{startpreprocessing}
\State ${\widetilde{\sigma_i}}^2 = \tfrac{1}{\widetilde{k_i}} (\pinitial - (\pinitial)^2) + (1-\tfrac{1}{\widetilde{k_i}})\widehat{\sigma_p}^2.$
\State $v_i=\frac{1}{{\widetilde{\sigma_i}}^2}$
\quad \Comment{Compute truncated, unnormalised weights}\label{endpreprocessing}
\EndFor
\State ${\widehat{\sigma_{\min}}}^2 = \tfrac{1}{\widehat{k_{\variancesamplesize}}} (\pinitial - (\pinitial)^2) + (1-\tfrac{1}{\widehat{k_{\variancesamplesize}}})\widehat{\sigma_p}^2.$ 
\State $\widehat{N} = \sum_{j=\variancesamplesize+1}^{9n/10} v_i+\Lap\left(\frac{1}{\epsilon{\widehat{\sigma_{\min}}}^2}\right)-\frac{1}{\epsilon{\widehat{\sigma_{\min}}}^2}\ln(2\delta)$ \label{noisynorm} \quad \Comment{Compute noisy normalisation term}
\State $\Lambda=12\frac{\concentrationbound{k_{\max}}{\Dee}{n}{\hat{\sigma_p}^2}{\beta}}{\widehat{\sigma_{\min}}^2\widehat{N}}$\label{senscomp} \quad \Comment{Compute local sensitivity proposal}
\vspace{0.1in}
\State \textbf{\underline{Propose-Test-Release on $\mathcal{M}(\cdot\; ; \widehat{k_{L}}, n, \pinitial, \widehatsigp, \alpha)$}}
\State $D_T=\{(\widehatpi, k_i)\}_{i\in[\variancesamplesize+1:9n/10]}$
\State $\kneighbor^* = \arg\max\{\kneighbor\in\mathbb{N}\;|\; \forall D' \text{ s.t. } D' \text{ is a } \kneighbor \text{-neighbor of } D_T, \LS(\mathcal{M}(\cdot; \widehat{k_{\variancesamplesize}}, 9n/10-\variancesamplesize, \pinitial, \widehatsigp, \alpha); D')\le \Lambda\}$\label{disttosens}\\ \quad \Comment{Compute distance to high sensitivity dataset}
\State $\tilde{\kneighbor}=\kneighbor^*+\Lap(1/\eps)$
\If{$\tilde{\kneighbor}<\frac{\log(1/\delta)}{\eps}$}
\State \textbf{return} $\prealisticunknown=\pinitial$\label{PTRfail} \quad \Comment{Return initial estimate if proposed local sensitivity too small}
\Else{}
\State Sample $Y \sim \Lap\left(\frac{\sensitivity}{\epsilon}\right)$ \quad \Comment{Sample noise added for privacy}
\State \textbf{return} $\prealisticunknown = \mathcal{M}(D_T; \widehat{k_{\variancesamplesize}}, 9n/10-\variancesamplesize, \pinitial, \widehatsigp, \alpha)  + Y$\label{finalk}\label{finalk} \quad \Comment{Final estimate}
\EndIf
\end{algorithmic}
\end{algorithm}

\begin{algorithm}[tbh]
\caption{Truncated weighted mean, $\mathcal{M}(\cdot; k_{\max}, n, \hat{p}, \hat{\sigma_p}^2, \alpha)$}\label{alg.PTR}
\textbf{Input:}  number of users $n$, number of samples held by each user $(k_1,\ldots, k_n)$, user-level estimates $(\widehatp_1, \cdots, \widehatp_n)$, desired upper bound $k_{\max}$, mean estimate $\hat{p}$, variance estimate $\hat{\sigma_p^2}$, accuracy on mean estimate $\alpha$
\begin{algorithmic}[1]
\vspace{0.1in}
\For{$i\in[n]$} 
\State $\widetilde{k_i}=\min\{k_i,k_{\max}\}$ 
\State $\widetilde{a_i} =  \hat{p}-\alpha-\concentrationbound{\widetilde{k_i}}{\Dee}{n}{\hat{\sigma_p}^2}{\beta}$\\
\State\label{settingintervalbi} $\widetilde{b_i} =  \hat{p}+\alpha+\concentrationbound{\widetilde{k_i}}{\Dee}{n}{\hat{\sigma_p}^2}{\beta}$\\
\State ${\widetilde{\sigma_i^2}} = \tfrac{1}{\widetilde{k_i}} (\hat{p} - (\hat{p})^2) + (1-\tfrac{1}{\widetilde{k_i}})\hat{\sigma_p^2}.$\label{truncatedvariance}
\State $v_i = \frac{1}{{\widetilde{\sigma_i^2}}}$\label{unnormalweight}
\EndFor
\State Return $ \frac{\sum_{i\in[n]}v_i [\widehatpi]_{\widehat{a_i}}^{\widehat{b_i}} }{\sum_{i\in[n]} v_i}   $ 
\end{algorithmic}
\end{algorithm}

The first significant change in Algorithm~\ref{alg.dpk} is how the sensitivity parameter $\Lambda$ is chosen. The final statistic is more sensitive under the view of \unknownki user level privacy; the weight of every user can change as a result of a single user changing the amount of data they hold (due to the resulting change in the normalisation constant). Rather than an upper bound on the global sensitivity, $\Lambda$ as defined in Algorithm~\ref{alg.dpk}, is, with high probability, an upper bound on the \emph{local} sensitivity of all databases that lie in a neighbourhood of $D$. Given a function $f$ from the set of databases to $\mathbb{R}$, and a database $D$, the \emph{local sensitivity} of $f$ at $D$ is defined by
$\LS(f; D) = \max_{D' \text{ neighbour of } D}|f(D)-f(D')|.$  We use a standard framework from the differential privacy literature called propose-test-release (PTR) \citep{Lei:2009} to privately verify that $\Lambda$ is indeed an upper bound on the local sensitivity of all databases in a neighbourhood of $D$, which allows us to safely add noise proportional to $
\Lambda$ to privatise the final statistic. A  database $D'$ is said to be a $\kneighbor$-neighbour of $D$ if it differs from $D$ on the data of at most $\kneighbor$ data subjects, and if it contains the same number of data subjects.

Next, the function $\mathcal{M}$ as described in Algorithm~\ref{alg.PTR} incorporates the truncation of weights in a slightly different (but nearly equivalent) manner to Algorithm~\ref{alg.dp}, but is otherwise the same as Algorithm~\ref{alg.dp}, without the addition of noise. Observe that choosing a truncation parameter $T$ is equivalent to choosing an integer $k$ such that $T=1/\var(\Dee(k))$, so $\widehat{k_{\variancesamplesize}}$ plays the role in Algorithm~\ref{alg.dpk} that $T^*$ plays in Algorithm~\ref{alg.dp}. The statistic $\widehat{k_{\variancesamplesize}}$ is a private estimate of the $\variancesamplesize$-th order statistic of the set $\{k_1, \cdots, k_n\}$. Since the only users that participate in the final estimate (and hence have their data truncated) all have $k_i<k_{\variancesamplesize}$, this algorithm attempts to find the smallest truncation parameter such that no data are actually truncated. We will show that provided either $\epsilon$ is not too small or the ratio $k_{\max}/\medk$ is not too large, this level of truncation is sufficient.
There are several existing algorithms in the literature that can be used to privately estimate the $\variancesamplesize$-th  order statistic $\widehat{k_{\variancesamplesize}}$. A simple algorithm \citep{Lei:2009, Thakurta:2013, Johnson:2013, alabi2020differentially, Asi:2020} that estimates the order statistic using standard differential privacy framework called the Exponential Mechanism (EM) \citep{McSherry:2007} is sufficient up to a constant factor. For a full description of this algorithm, as well as its accuracy guarantees, see \citep{Asi:2020}. In order for this algorithm to produce accurate results, we need an upper bound on the maximum number of data points a single user can have; we will call this number $k_{\max}$.

\begin{restatable}{theorem}{privatek}
%\begin{theorem}
\label{thm.privatek}
For any $\eps>0$, $\delta\in[0,1]$, $\beta\in[0,1]$,  $n\in\mathbb{N}$, $\alpha>0$, $\variancesamplesize\in[n]$ $(\eps,\delta)$-DP mean estimator $\meanest$, $(\eps,\delta)$-DP variance estimator $\varianceest$, $k_{\max}\in\mathbb{N}$, $\eps$-DP estimator of the $\ell$th order statistic $\expmech(\cdot; \ell, k_{\max})$, Algorithm~\ref{alg.dpk} is $(3\epsilon,2\delta)$-DP. Let $\Upsilon=\frac{\log(1/\delta)}{\epsilon}+\frac{\ln(1/\delta)\ln(1/\beta)}{\epsilon}$. If the conditions of Theorem~\ref{metatheorem} hold and \begin{itemize}
    \item $\frac{1}{2}\frac{1}{\epsilon}(\ln k_{\max}+\ln (1/\beta))\le \variancesamplesize\le n/4$,
    \item $\frac{k_{\max}}{k_{\text{med}}} \le \min\left\{ \frac{\log\frac{n}{\beta}}{\log\frac{n^{\Upsilon+1}}{\beta}}\frac{n-\Upsilon-1}{2}, \frac{n-1}{2(\Upsilon+1)}, \frac{\epsilon^2(n/2-\variancesamplesize-1)}{\log^2(n/\beta)}, \frac{(n/4-1)\epsilon}{3\ln(2/\delta)}\right\}$, 
    \item for all $k\le k_{\max}$, $\max\{\alpha, \sigma_{k}\}\le \concentrationbound{k}{\Dee}{n}{\hat{\widehat{\sigma}_p}^2}{\beta}\le 2\sigma_k\sqrt{\log(n/\beta)}$, where $\sigma_k^2=\var(\Dee(k))$ 
    \item for any set $I\subset [n]$, with probability $1-\beta$, $\left|\frac{\sum_{i\in I} v_i\widehatpi}{\sum_{i\in I} v_i}-p\right|\le 2\sqrt{\var\left(\frac{\sum_{i\in I} v_i\widehatpi}{\sum_{i\in I} v_i}\right)\log(1/\beta)}$,
\end{itemize}
then with probability $1-4\beta$, $\var(\prealisticunknown)\le \tilde{O}\left(\var(\pnprealistic)\right)$
%\end{theorem}
\end{restatable}

Theorem~\ref{thm.privatek} implies that under some mild conditions, the variance of $\prealisticunknown$ is within a constant factor of the variance of $\pnprealistic$, the non-private realisable estimator. While the conditions of this theorem may seem intimidating, they are not particularly stringent for reasonable parameter settings. 
\begin{itemize}
    \item \textbf{Conditions on L.} In Section~\ref{s.finalest}, when discussing the conditions of Theorem~\ref{metatheorem},  we discussed that $\variancesamplesize=\tilde{O}(1/\epsilon)$ is sufficient for learning a constant multiplicative approximation to $\sigp$ for sufficiently well-behaved distributions. We'll give such an example estimator in Section~\ref{s.popvar}. If we increase $\variancesamplesize$ to $O(\log(n)/\epsilon)$ then the third condition in Theorem~\ref{metatheorem} (which we still need to satisfy) becomes only slightly more restrictive, and we can satisfy the first condition of Theorem~\ref{thm.privatek} provided $k_{\max}$ and $1/\beta$ are both polylogarithmic in $n$. 
    \item \textbf{Conditions on $k_{\max}/\medk$.} Up to logarithmic factors, the required upper bound on the ratio $k_{\max}/\medk$ is $\tilde{O}(\epsilon^2 n)$. For moderate values of $\epsilon$, this condition is unlikely to be prohibitive in practice, although it is more restrictive than the upper bound of $\tilde{O}(\epsilon n)$ that was required in Theorem~\ref{metatheorem}.
    \item \textbf{Concentration bounds.} The final two conditions are concentration bounds, essentially requiring $\Dee(k)$ to be sub-Gaussian. This condition is technically absent from Theorem~\ref{metatheorem}, although a similar condition is required in order to design a private variance estimation algorithm with sufficiently good accuracy.
\end{itemize}

The proof that Algorithm~\ref{alg.dpk} is $(3\epsilon, 2\delta)$-DP is fairly routine, details can be found in the appendix.
There are two main differences between Algorithm~\ref{alg.dpk} and Algorithm~\ref{alg.dp} that affect the utility: the replacement of the optimal truncation with truncation based on $\widehat{k_{\variancesamplesize}}$, and the use of propose-test-release (PTR) to determine the level of noise added to the final estimate. We will control the impact of these two factors separately. 

Let us consider the impact of changing the truncation parameter. Set $T_{\variancesamplesize}=\frac{1}{\widehat{\sigma_{\min}}^2}$.
Assuming the PTR component of the algorithm does not fail, the variance of $\prealisticunknown$ can be written as two terms, namely the variance that exists in the non-private setting, and the additional noise due to privacy:
\[\var(\prealisticunknown) =  \underbrace{\frac{\sum_{i=\variancesamplesize+1}^{9n/10} \min\left\{\frac{T_{\variancesamplesize}^{2}}{\widehatsigi}, \frac{1}{\widetilde{\sigma_i}^4}\right\}\var([\widehatpi]_{\widetilde{a_i}}^{\widetilde{b_i}}])}{\left(\sum_{i=\variancesamplesize+1}^{9n/10} \min\left\{\frac{T_{\variancesamplesize}}{\widetilde{\sigma_i}}, \frac{1}{\widetilde{\sigma^2_i}}\right\}\right)^2}}_{\text{non-private term}} + \underbrace{ \frac{\left(12\frac{\concentrationbound{\widehat{k_{\variancesamplesize}}}{\Dee}{n}{\hat{\sigma_p}^2}{\beta}}{\widehat{\sigma_{\min}}^2\widehat{N}}\right)^2}{\epsilon^2}}_{\text{private term}}. \] 
The truncation has opposite effects on each of these terms. As $T$ decreases, the private variance term decreases while the non-private variance term increases. When we set $T_{\variancesamplesize}=1/\var(\Dee(k_{\variancesamplesize+K}))$, where $K\in [-\frac{1}{2}\variancesamplesize, \frac{1}{2}\variancesamplesize]$ then if $K$ is negative, no truncation occurs and the non-private term is optimal. Even if $K$ is positive, only a small number of data points are truncated so the non-private term is still close to its optimal value. However, setting the truncation parameter this large means that the private term is larger than necessary. We show that even though the private term may be larger than it would be with the optimal truncation, under the conditions of the theorem, the non-private term dominates the variance anyway.

Let us now consider the impact of the use of propose-test-release (PTR).
The two relevant components for the how the PTR component of Algorithm~\ref{alg.dpk} affects the utility are the scale of $\Lambda/\epsilon$ and the probability that the proposed sensitivity is too small resulting in the algorithm ending in line~\eqref{PTRfail}, rather than line~\eqref{finalk}. The impact of the former is easy to analyse since the noise added is simply output perturbation. In order to show that the PTR ends in line~\eqref{finalk} with high probability, we need to show that with high probability (over the randomness in the samples), $\kneighbor^*$ as defined in line~\eqref{disttosens} is large enough. Since this claim is in essence about $\mathcal{M}(\cdot; k_{\max}, n, \hat{p}, \hat{\sigma_p}^2)$, we will state this claim in the notation of Algorithm~\ref{alg.PTR}.

\begin{restatable}{lemma}{PTRfailureratelem}\label{PTRfailurerate} Given $k_{\max}\in\mathbb{N}$, $n\in\mathbb{N}$, $\hat{p}\in[0,1]$, $\hat{\sigma_p}^2\in[0,1]$ and $k_1,\cdots, k_n$, let $\Upsilon=\frac{\log(1/\delta)}{\epsilon}+\frac{\ln(1/\delta)\ln(1/\beta)}{\epsilon}$, if 
the conditions of Theorem~\ref{thm.privatek} hold and $D=\{(\widehatpi, k_i)\}_{i=1}^n$ is a dataset such that $\widehatpi\sim\Dee(k_i)$, then with probability $1-\beta$, 
    for any $D'$ that is a $\kneighbor$-neighbour of $D$ for $0\le \kneighbor\le \Upsilon$, we have 
    \[\LS(\mathcal{M}(\cdot; k_{\max}, m, \hat{p}, \hat{\sigma_p}^2, \alpha); D')\le 12\frac{v_{k_{\max}}\concentrationbound{k_{\max}}{\Dee}{n}{\hat{\sigma_p}^2}{\beta}}{\sum_{i=1}^{n} v_i}.\] 
\end{restatable}

\section{Near Optimality and Lower Bounds}\label{s.optpideal}

In Section \ref{s.privest}, we showed that the variance of our realisable private estimator $\prealistic$ was within a constant of that of the complete information estimator $\pideal$. In this section, we will show that in fact, $\prealistic$ performs as well (up to logarithmic factors) as the true optimal private estimator. 
We'll also give a lower bound on the performance of the optimal estimator in terms of the $k_i$. This will give us some intuition into the types of distributions of $k_i$'s that benefit from this refined analysis.

\subsection{Minimax Optimality of $\prealistic$}
The goal of this section is to show that the estimator $\prealistic$ discussed in Section \ref{s.finalest} is minimax optimal up to logarithmic factors among the class of unbiased estimators. In light of  Theorem~\ref{metatheorem}, it suffices to show that the estimator $\pideal$ defined by Equations~\ref{privoptimal1},~\eqref{idealtruncation}, and~\eqref{idealthreshold} 
is minimax optimal up to logarithmic factors. 
Let $\mathcal{P}$ be a parameterized family of distributions $p\mapsto \Dee_{p}$, where $\mathbb{E}[\Dee_{p}]=p$ and $\Dee_p$ is supported on $[0,1]$. For $p\in[0,1]$ and $k\in\mathbb{N}$, let $\phi_{p,k}$ be the probability density function of $\Dee_p(k)$. In this section, we will return to the known size user-level differential privacy setting. Hence, we will let $k_1,\cdots,k_n$ be fixed.

Our lower bound will show that the estimation error must consist of a statistical term and a privacy term. Such a lower bound thus must generalize a statistical lower bound. We will rely on the Cram\'er-Rao approach to proving statistical lower bounds; as we show, it is particularly amenable to incorporating a privacy term. This approach relates the variance of any unbiased estimator of the mean of a distribution to the inverse of the Fischer information; the proof naturally extends to the case where we are given samples from a set of distributions with the same mean but different variances, as is the case in our setting. For many distributions of interest, e.g., Gaussian and Bernoulli, the Fischer information of a single sample is the inverse of the variance, and we make that assumption for $\Dee_p$. We also assume that the $\Dee_p$ has sub-Gaussian tails. Thus, as long as the set of permissible meta-distributions includes distributions with this property, e.g., includes truncated Gaussians, our lower bound applies.

\begin{theorem}\label{optimalityfisherinfo}
Let $\mathcal{P}$ be a parameterized family of distributions $p\mapsto \Dee_{p}$ and suppose that for all $p\in[0,1]$ and $k\in\mathbb{N}$, the Fisher information of $\phi_{p,k}$ is inversely proportional to the variance, $\var(\Dee_p(k))$:
\begin{equation}\label{condition2} \textstyle\int (\tfrac{\partial}{\partial p}\log \phi_{p,k}(x))^2 \phi_{p,k}(x)dx=O(\tfrac{1}{\var(\Dee_p(k))}),
\end{equation} and for all $p$, $n>0$, $k\in\mathbb{N}$ and  $\beta\in[1/3,2/3]$, $\concentrationbound{k}{\Dee_p}{n}{\sigma_p^2}{\beta}=\tilde{O}(\var(\Dee_p(k)))$,
then \begin{align*}
\min_{M \text{\rm, unbiased} }\max_{p\in[1/3,2/3]}[\var_{\forall i\in[n], x_i\sim\Dee(k_i), M}(M)]&=\tilde{O}\left(\max_{p\in[1/3,2/3]}\left[\var_{\forall i\in[n], x_i\sim\Dee(k_i), M}(\pideal)\right]\right)\\
&=\tilde{O}\left(\min_T \tfrac{\textstyle\sum_{i=1}^n\min\{1/\sigma_i^2, {T^2}\} +\max_i \tfrac{\min\{1/\sigma_i^4, T^2/\sigma_i^2\}|b_i-a_i|^2}{\epsilon^2} )}{{(\sum_{j=1}^n \min\{1/\sigma_j^2, T/\sigma_i\})^2}}\right).
\end{align*}
Further, under the conditions of Theorem~\ref{metatheorem},  \begin{align*}
\max_{p\in[1/3,2/3]}\left[\var_{\forall i\in[n], x_i\sim\Dee(k_i), M}(\prealistic)\right]=
\tilde{O}\left(\min_{M \text{\rm, unbiased} }\max_{p\in[1/3,2/3]}[\var_{\forall i\in[n], x_i\sim\Dee(k_i), M}(M)]\right).
\end{align*}
\end{theorem}

Theorem~\ref{optimalityfisherinfo} says the estimator $\pideal$ has variance only a logarithmic factor worse than the variance of the optimal unbiased estimator. Due to the truncation of the $\hat{p_i}$, the estimator $\pideal$ is not unbiased, although the bias can be made polynomially small by widening the truncation interval so truncation does not occur with high probability. The theorem can also be slightly extended to include estimators with polynomially small bias. This small bias assumption seems to be inherent in the Cramer-Rao style proof that we use.

We will prove Theorem~\ref{optimalityfisherinfo} in three steps. The following class of noisy linear estimators, $\linearest$, will act as an intermediary in our proof. The notation $\sigma_i$ denotes $\var(x_i)$, which accounts for the randomness in generating $x_i$.
 \begin{align*}\linearest = \Big\{M_{\texttt{NL}}(\mathbf{x}; \mathbf{w}) = &\textstyle\sum_{i=1}^n w_ix_i+\Lap(\tfrac{\max_i w_i\sigma_i}{\epsilon})\;\big| \; w_i\in[0,1], \textstyle\sum_{i=1}^n w_i=1\Big\}.\end{align*}
Similar to $\pideal$, this class of estimators is not realizable since we only have access to an estimate of $\sigma_i=\var(\Dee_p(k_i))$. Additionally, the estimators in $\linearest$ are not necessarily $\epsilon$-DP.

To prove Theorem \ref{optimalityfisherinfo}, we will first show that the weights used in $\pideal$ define the optimal weight vector among the estimators in $\linearest$. Then, we'll show that (up to constant factors) the minimax optimal estimator among unbiased estimators lies in $\linearest$. Finally, we'll show that the variance of $\pideal$ is at most a logarithmic factor worse than its not-quite-private counterpart in $\linearest$. This completes the proof of the near minimax optimality of $\pideal$, and hence $\prealistic$.

The first step is shown in Lemma \ref{optimalityofthresholding}, which shows that the weights used in $\pideal$ are optimal (i.e., variance-minimizing) among all estimators in the set $\linearest$.

\begin{restatable}{lemma}{optimalthresh}
%\begin{lemma}
\label{optimalityofthresholding}
Given $\widehatpi \sim \Dee_{p}(k_i)$ with variance $\sigma^2_i$ for all $i\in[n]$ and $w\in[0,1]^n$ such that $\sum_{i=1}^n w_i=1$, let 
$\widehatp = \sum_{i=1}^n w_i \widehatpi  + \Lap(\frac{\max_i w_i\sigma_i}{\epsilon})$. 
The variance of $\widehatp$ is minimized by the following weights: 
\[\tilde{w_i}^* = \frac{\min\{1/\sigma_i^2, T/\sigma_i\}}{\sum_{j=1}^n \min\{1/\sigma_j^2, T/\sigma_j\}}\]
for some $T$.
%\end{lemma}
\end{restatable}

Since the threshold $T^*$ in $\pideal$ was chosen to minimize $\var(\pideal)$, then we know that the weights $w_i^*$ in $\pideal$ are optimal. The proof of Lemma~\ref{optimalityofthresholding} can be found in Appendix~\ref{appendix:opt}. The main component of the proof is showing that under the constraint of differential privacy, no individual's contribution should be too heavily weighted.

Now, let us turn to the second -- and main -- component of the proof of Theorem~\ref{optimalityfisherinfo}. Lemma \ref{optimalityoflinearmain} formalises the statement that an estimator inside the class $\linearest$ is minimax optimal among unbiased estimators. That is, for any unbiased estimator $M$, there exists an estimator $M_{\texttt{NL}}\in \linearest$ with lower worst-case variance.

\begin{restatable}{lemma}{optimallinearmain}
%\begin{lemma}
\label{optimalityoflinearmain}
Let $\mathcal{P}$ be a parameterized family of distributions $p\mapsto \Dee_{p}$ and suppose that $M:[0,1]^n\to[0,1]$ is an $\eps$-DP estimator such that for all $p\in[1/3,2/3]$, if
\begin{enumerate}
\item $M$ is unbiased, $\mu_M(p)=p$ 
\item\label{condition2} the Fisher information of $\phi_{p,k_i}$ is inversely proportional to the variance 
\[\textstyle\int (\tfrac{\partial}{\partial p}\log \phi_{p,k_i}(x_i))^2 \phi_{p,k_i}(x_i)dx_i=O(\tfrac{1}{\var(\Dee_p(k_i))}),\]
\end{enumerate}
then there exists an estimator $M_{\texttt{NL}}\in \linearest$ such that 
\begin{align*}\max_{p\in[1/3,2/3]}&[\var_{\forall i\in[n], x_i\sim\Dee(k_i), M_{\texttt{NL}}}(M_{\texttt{NL}})] \le O\left(\max_{p\in[1/3,2/3]}[\var_{\forall i\in[n], x_i\sim\Dee(k_i), M}(M)]\right).\end{align*}
%\end{lemma}
\end{restatable}

A detailed proof of Lemma~\ref{optimalityoflinearmain} can be found in Appendix~\ref{appendix:opt}, but let us give a brief sketch of the proof here.
Given an estimator $M_{\texttt{NL}}\in\linearest$, the variance of $M_{\texttt{NL}}$ can be written as \begin{align}\label{linearvariance}
\var(M_{\texttt{NL}})&\le\textstyle\sum_{i=1}^n w_i^2\var(\Dee(k_i))+O(\tfrac{\max w_i\sigma_i}{\epsilon})^2.
\end{align} That is, it can be decomposed as the variance contribution of each individual coordinate, and the variance contribution of the additional noise due to privacy. Lemma \ref{variancedecomp} (proved in Appendix~\ref{appendix:opt}) shows that the variance of any estimator $M$ can be lower bounded by a similar decomposition.
Since this involves considering the impact of each coordinate individually, the following notation will be useful.
Given an estimator $M$, vector $\boldsymbol{q}\in[0,1]^n$ and set $I\subset[n]$, let
$\mu_M(x_{[n]\backslash I}; \boldsymbol{q}) = \mathbb{E}_{\forall i\in I, x_i\sim \Dee_{q_i}(k_i), M}[M(x_1,\cdots,x_n)]$
be the expectation over only randomness in $I$ and $M$. 
Note that in this notation, user $i$ is sampling from a meta-distribution with mean $q_i$, which may be different for each user. We will abuse notation slightly to let $\mu_M(\boldsymbol{q})=\mu_M(\emptyset;\boldsymbol{q})$, and for $p\in[0,1]$, we will let $\mu_M(x_{[n]\backslash I}; p)=\mu_M(x_{[n]\backslash I}; (p,\cdots,p))$.
When the estimator $M$ is clear from context, we will omit it.

\begin{restatable}{lemma}{variancedecomp}
%\begin{lemma}
\label{variancedecomp} For any randomised mechanism $M:[0,1]^n\to[0,1]$,
\begin{align}
\nonumber&\var_{{\forall i\in[n], x_i\sim\Dee_p(k_i), M}}(M) 
= \mathbb{E}_{\forall i\in[n], x_i\sim\Dee_p(k_i), M}[(M(x_1,..., x_n)-\mu(p))^2]\\
&\hspace{0.3in}\ge \textstyle\sum_{i=1}^n \mathbb{E}_{x_i\sim \Dee_p(k_i)}[(\mu(x_{i};p)-\mu(p))^2]\label{decomp}+\mathbb{E}_{\forall i\in[n], x_i\sim\Dee_p(k_i),M}[(M(x_1,...,x_n)-\mu(x_1,...,x_n;p))^2]
\end{align}
%\end{lemma}
\end{restatable}

In Equation \eqref{decomp}, the first term is the sum of contributions to the variance of the individual terms $x_i$, and the second term is the contribution to the variance of the noise added for privacy.
Now we want to define a weight vector $\bf{w}$ such that the terms in Equation~\eqref{decomp} are lower bounded by the corresponding terms in Equation~\eqref{linearvariance}. The key component of the proof is the observation that if we
let \begin{equation}\label{wipdef}w_i(p) =\tfrac{\partial}{\partial q_i}\mu(\boldsymbol{q})\;\big\rvert_{\boldsymbol{q}=(p, \cdots, p)}\end{equation} then we can show that there exists a constant $c$ such that 
\begin{equation}\label{eq.varMpart1}\mathbb{E}_{x_i\sim\Dee_p(k_i)}[(\mu(x_i;p )-\mu(p))^2]\ge c\cdot w_i(p)^2\var(\Dee_p(k_i)). \end{equation}
This controls the contribution of each individual coordinate to the variance of $M$. It remains only to control the contribution of the noise due to privacy. We show that there exists $x_i$, $x_i'$ such that
\begin{equation*}
|\mu(x_i;p)-\mu(x_i';p)|\ge \Omega(w_i(p)\cdot\textstyle\sqrt{\var(\Dee_p(k_i))}),\end{equation*}
 which we show implies that, 
\begin{align}\label{eq.varMpart2}\mathbb{E}&_{\forall i\in[n], x_i\sim\Dee_p(k_i),M}[(M(x_1,\cdots,x_n)-\mu(x_1,\cdots,x_n;p))^2]\ge \Omega(\tfrac{w_i(p)^2\var(\Dee_p(k_i))}{\eps^2}).\end{align}
Intuitively, the worst-case $|\mu(x_i;p)-\mu(x_i';p)|$ plays an analogous role to the sensitivity, since it captures the impact of changing one user's data. Since $M$ is an $\eps$-DP mechanism and $|\mu(x_i;p)-\mu(x_i';p)|$ is at least $\Omega(w_i(p)\cdot\sqrt{\var(\Dee_p(k_i))})$, we show that it must include noise with standard deviation of at least this magnitude over $\eps$. This is consistent with, e.g.,  the Laplace Mechanism that adds noise with standard deviation $\Theta(\Delta f /\epsilon)$.

Combining Lemma \ref{variancedecomp}
with Equations \eqref{eq.varMpart1} and \eqref{eq.varMpart2} gives that the variance of $M$ is at least,
\[\var_{\forall i\in[n], x_i\sim\Dee_p(k_i), M}(M)
 \ge \textstyle\sum_{i=1}^n c\cdot w_i(p)^2\var(\Dee_p(k_i)) + \Omega(\tfrac{w_i(p)^2\var(\Dee_p(k_i))}{\eps^2}).\]

Finally, we must create a corresponding $M_{\texttt{NL}}\in \linearest$ for comparison, using the same weights. Since $\sum_{i=1}^nw_i(p)$ as defined in Equation \eqref{wipdef} need not equal 1, these weights will need to be normalized to sum to 1 to create an estimator in $\linearest$. We need to show this normalisation does not substantially increase the variance of the resulting estimator.  In order to show this, we show that there exists a $p^*\in[1/3,2/3]$ such that 
$\sum_{i=1}^n w_i(p^*)\ge 1$, since normalizing the estimator by a factor of $\frac{1}{\sum_{i=1}^n w_i(p^*)}$ will affect the variance by a factor of $\frac{1}{(\sum_{i=1}^n w_i(p^*))^2}$, and thus if $\sum_{i=1}^n w_i(p^*)\ge 1$, then this will decrease variance. This desired fact follows from the definition of $w_i$, and the fact that $M$ is unbiased. Now, if we 
define 
\[M_{\texttt{NL}}(\mathbf{x}) = \tfrac{\textstyle\sum_{i=1}^n w_i(p^*)x_i+\Lap(\tfrac{\max_i w_i(p^*)\sqrt{\var(\Dee_p(k_i))}}{\epsilon})}{\textstyle\sum_{i=1}^n w_i(p^*)},\]
then $M_{\texttt{NL}}\in\linearest$ and $\var_{\forall i\in[n], x_i\sim\Dee_p(k_i),M_{\texttt{TNL}}}(M_{\texttt{NL}})=\Theta\left(\var_{\forall i\in[n], x_i\sim\Dee_p(k_i),M}(M)\right)$.

The final component needed for the proof of Theorem~\ref{optimalityfisherinfo} is a translation from the estimators in $\linearest$, which are not $\epsilon$-DP to the corresponding $\epsilon$-DP estimator.
For any weight vector $\bf{w}$, we can define an $\epsilon$-DP estimator by truncating the data point $x_i$ and calibrating the noise appropriately:
\begin{align}\label{linearestimator}
M_{\texttt{TNL}}(x_1, \cdots, x_n; \mathbf{w}) = \textstyle\sum_{i=1}^n w_i[x_i]_{p-\concentrationbound{k_i}{\Dee}{n}{\sigma_p^2}{\beta}}^{p+\concentrationbound{k_i}{\Dee}{n}{\sigma_p^2}{\beta}}+\Lap(\tfrac{\max_i 2w_i\concentrationbound{k_i}{\Dee}{n}{\sigma_p^2}{\beta}}{\epsilon}).\notag
\end{align}
Provided $\concentrationbound{k_i}{\Dee}{n}{\sigma_p^2}{\beta}\approx\var(\Dee(k_i))$, the estimators $M_{\texttt{TNL}}$ have approximately the same variance as the corresponding element of $\linearest$, but are slightly biased. This is formalized in the following lemma.

\begin{restatable}{lemma}{optthreshold}
%\begin{lemma}
\label{optimalitythresholdingdoesntmatter}
For any distribution $\Dee$, $n>0$ and  $\beta\in[0,1]$, if for all $k_i$, $\concentrationbound{k_i}{\Dee}{n}{\sigma_p^2}{\beta}=\tilde{O}(\var(\Dee(k_i))$ then
for any $\mathbf{w}\in[0,1]^n$ such that $\sum_{i=1}^n w_i=1$, we have $\var(M_{\texttt{TNL}}(\cdot\;; \mathbf{w}))= \tilde{O}(\var(M_{\texttt{NL}}(\cdot\;; \mathbf{w})))$. Further, the bias of $M_{\texttt{TNL}}$ is at most $\beta$.
%\end{lemma}
\end{restatable}
Finally, we have the tools to prove the main theorem in this section, Theorem~\ref{optimalityfisherinfo}:
\begin{align*}
\min_{M \text{ unbiased} }\max_{p\in[1/3,2/3]}[\var_{\Dee_p}(M)] &=\Omega(\min_{M\in\texttt{NLE}}\max_{p\in[1/3,2/3]}[\var_{\Dee_p}(M)])\\
&= \Omega(\max_{p\in[1/3,2/3]}[\var_{\Dee_p}({p_{\epsilon}^{\texttt{NLE}}})])\\
&= \tilde{\Omega}(\max_{p\in[1/3,2/3]}[\var_{\Dee_p}(\pideal)])\\
&= \tilde{\Omega}(\max_{p\in[1/3,2/3]}[\var_{\Dee_p}(\prealistic)]) 
\end{align*}
where $p_{\epsilon}^{\texttt{NLE}}\in\texttt{NLE}$ has the same weights as $\pideal$. The equalities follow from Lemmas~\ref{optimalityoflinearmain}, \ref{optimalityofthresholding}, \ref{optimalitythresholdingdoesntmatter}, and Theorem \ref{metatheorem}, respectively.

\subsection{Minimax Lower Bound on Estimation Rate}
In addition to establishing the near optimality of $\prealistic$, we will also give a lower bound on minimax rate of estimation in terms of the parameters $k_1,\cdots, k_n$ and $\sigma_p^2$. Note that we can view the truncation of the weights $w_i$ as establishing an effective upper bound on $k_i$.
Given $k_1, \cdots, k_n\in\mathbb{N}$, and $\epsilon>0$, let 
\begin{equation}\label{kstar}
k^* = \arg\min_k \tfrac{\tfrac{k}{\epsilon^2}+\sum_{i=1}^n \min\{k_i, k\}}{(\sum_{i=1}^n \min\{k_i,k\})^2}.
\end{equation}
Intuitively, in the case that $\sigma_p=0$, we want to use as many samples as possible, but one user contributing many samples leads to larger sensitivity and thus privacy cost. Limiting the number of samples per user to $k_{\max}$ allows us to limit the sensitivity to be about $w_{\max}(1/\sqrt{k_{\max}})$. Since $w_i$ is proportional to the number of samples used, the variance of the estimator when using at most $k^*$ samples per user is akin to choosing a threshold that minimises the variance.

\begin{restatable}{corollary}{lowerb}
%\begin{corollary}
\label{cor.lower}
Given $k_1, \cdots, k_n\in\mathbb{N}$, and $\sigma_p$, there exists a family of distributions $\Dee_p$ such that \begin{align*}
\min_{M\text{\rm, unbiased}} \max_{p\in[1/3,2/3]}\var_{\forall i\in[n], x_i\sim\Dee_p(k_i)}[M(x_1, \cdots, x_n)]\ge\tilde{\Omega}\left(\min_{k^*}\left\{ \tfrac{\frac{k^*}{\epsilon^2}+\sum_{i=1}^n \min\{k_i, k^*\}}{(\sum_{i=1}^n \min\{k_i,\sqrt{k_ik^*}\})^2}, \frac{\sigma_p^2}{n}\right\}\right).\end{align*}
%\end{corollary}
\end{restatable}

Corollary~\ref{cor.lower} is proved in two parts, using two different families of distributions $\Dee_p$. The first family is where $\sigp=0$, so $\Dee_p(k)=\Bin(k,p)$ for all $k\in[n]$. For this family, we know that the minimax error is obtained by the mechanism $\pideal$. Calculating the variance of $\pideal$ on this family, we obtain the first term of the minimum. The second family is the family of truncated Gaussian distributions (truncated so that $\Dee$ is supported on $[0,1]$). The variance of the optimal estimator for this family would be lower bounded by $\sigp/n$, even if each user was given a sample directly from $\Dee$, rather than from $\Dee(k)$. Thus, using a reduction to the case of simply estimating $p$ given $n$ samples from $\Dee$, we obtain the second term in the minimum.

\section{Example Initial Estimators}\label{instantiation}

In this section we give example initial mean and variance estimation procedures that can be used in the framework described in Section~\ref{s.privest}. For both estimators, we show that they satisfy the conditions of Theorem \ref{metatheorem}, and thus can be used as initial estimators in Algorithm \ref{alg.dp}, assuming all other technical conditions are satisfied. This also immediately implies that the set of initial mean and variance estimators which satisfy the conditions of Theorem \ref{metatheorem} is non-empty.

We note again that the estimators described in this section are examples of estimators that achieve the conditions of Theorem \ref{metatheorem}, and that any private mean and variance estimators that satisfy these conditions  could be used instead. As discussed in Section~\ref{s.finalest}, one may choose to use different estimators of these initial quantities in different settings (for example, if local differential privacy is required or if different distributional assumptions are known).

\subsection{Initial Mean Estimation}\label{s.initmean}

We will begin with the initial mean estimation procedure $\meanest$ to computed $\pinitial$. We consider the simplest mean estimation subroutine, where the analyst collects a single data point from the $n/10$ users with the smallest $k_i$, then privately computes the empirical mean of these points using the Laplace Mechanism. The following lemma shows that this process is differentially private and satisfies the accuracy conditions of Theorem \ref{metatheorem}, i.e., that with high probability, $\pinitial$ is close to $p$ and $\pinitial(1-\pinitial)$ is close to $p(1-p)$.

\begin{restatable}{lemma}{initmean}
%\begin{lemma}
\label{initialmeanestimate} 
Fix any $\epsilon>0$ and let $\pinitial = \meanest(x^1_{(9n/10)+1}, \cdots, x^1_{n}) = \frac{1}{n/10}\sum_{i=(9n/10)+1}^{n} x_i^1 +\Lap\left(\frac{10}{\epsilon n}\right)$. Then $\meanest$ is $(\epsilon,0)$-differentially private, $\mathbb{E}[\pinitial]= p$ and if $p\ge \frac{20\log(1/\beta)}{n}$, then for $n$ sufficiently large, 
\[\Pr[|\pinitial-p|\le \alpha]\le\beta \text{ for }
 \alpha = 2\max\{\sqrt{\tfrac{12\pinitial\log(4/\beta)}{n/10}+\tfrac{36\log^2(4/\beta)}{n^2/100}}+\tfrac{6\log(4/\beta)}{n/10},\tfrac{\log(2/\beta)}{\epsilon n/10}\} \leq \concentrationbound{k_i}{\Dee}{n}{\sigp}{\beta}. \]
Further, if $\min\{p,1-p\}\ge 12 \max\left\{\frac{3\log(4/\beta)}{n/10}, \frac{\log(2/\beta)}{\epsilon n/10}\right\}$ then with probability $1-\beta$, $\pinitial\in[\frac{1}{2}p, \frac{3}{2}p]$ and $\pinitial(1-\pinitial)\in [\frac{p(1-p)}{2}, \frac{3p(1-p)}{2}]$.
%\end{lemma}
\end{restatable}
 
The concentration bound follows from noticing that $\berP=\Ber(p)$ and using the concentration of binomial random variables. The full proof is in Appendix \ref{appendix.private}. 

Note that the expression of $\alpha$ depends only on quantities known to the analyst -- including $\pinitial$, which will be observed as output -- so that $\alpha$ can be computed directly for use in Algorithm \ref{alg.dp}. Although our presentation of Algorithm \ref{alg.dp} requires $\alpha$ to be specified up front as input to the algorithm, it could equivalently be computed internally by the algorithm as a function of $\pinitial$ and other input parameters.

\subsection{Initial Variance Estimation}\label{s.popvar}

We now turn to our variance estimation procedure $\varianceest$ for estimating $\sigma_p^2$. Let us first provide some background on privately estimating the standard deviation of well-behaved distributions. Lemma \ref{highprobstd} guarantees the existence of a differentially private algorithm for estimating standard deviation within a small constant factor with high probability, as long as the sample size is sufficiently large. The following is a slight generalisation of the estimation of the standard deviation of a Gaussian given by \cite{Karwa:2018}. 

\begin{lemma}[DP standard deviation estimation]\label{highprobstd} 
For all $n\in\mathbb{N}$, $\sigma_{\min}<\sigma_{\max}\in[0, \infty], \epsilon>0, \delta\in(0,\frac{1}{n}], \beta\in(0,1/2), \zeta>0,$ there exists an $(\epsilon, \delta)$-differentially private algorithm $\mathcal{M}$ that satisfies the following: if $x_1, \ldots, x_n$ are i.i.d.~draws from a distribution $P$ which has standard deviation $\sigma\in[\sigma_{\min}, \sigma_{\max}]$ and absolute central third moment $\rho=\mathbb{E}[|x-\mu(P)|^3]$ such that $\frac{\rho}{\sigma^3}\le \zeta$, then if $n\ge c \zeta^2\min\{\frac{1}{\epsilon}\ln(\frac{\ln\frac{\sigma_{\max}}{\sigma_{\min}})}{\beta}), \frac{1}{\epsilon}\ln(\frac{1}{\delta\beta})\}$, 
(where $c$ is a universal constant), then $\mathcal{M}$ produces an estimate $\widehat{\sigma}$ of the standard deviation such that $\Pr_{x_1,\ldots,x_n\sim P, \mathcal{M}} (\sigma^2\le\widehat{\sigma}^2\le 8\sigma^2)\ge 1-\beta$.
\end{lemma}

The proof of Lemma \ref{highprobstd} is given formally in Appendix~\ref{appendix:highprobstd}, along with a detailed description of the algorithm $\mathcal{M}$. The remaining omitted proofs in this section are in Appendix  \ref{appendix.private}. We note that the interval $[\sigma_{\min}, \sigma_{\max}]$ can be set fairly large without much impact on the sample complexity, in the case that little is known about $\sigma$ a priori.

In order to estimate $\sigma_p^2$, we will use the estimator promised by Lemma~\ref{highprobstd} on the data of the $L=\log n/\eps$ users with the largest $k_i$.  Let $k=k_{\log n/\eps}$, so the top $\log n/\eps$ individuals all have at least $k$ data points. We will have these individuals report $\widehat{p}_i^k := \frac{1}{k}\sum_{j=1}^k x_j^i$, which is the empirical mean of their first $k$ data points. Thus, we are running the estimator promised in Lemma~\ref{highprobstd} on $\bersumP$ with $\log n/\eps$ data points. 
In order to utilise Lemma~\ref{highprobstd}, we first need to ensure that $\bersumP$ satisfies the moment condition that $\rho/\sigma^3$ is bounded, which is shown in Lemma \ref{lem.thirdmoment}.

\begin{restatable}{lemma}{lemthirdmoment}
%\begin{lemma}
\label{lem.thirdmoment}
For $k\in\mathbb{N}$, suppose $p\in[\frac{1}{k}, 1-\frac{1}{k}]$, $\sigma_p\ge \frac{1}{k}$, $k\ge 2$, and there exists $\gamma>0$ such that $\frac{\rho_{\Dee}}{\sigma_p^3}\le\gamma$ where $\rho_{\Dee}$ denotes the absolute central third moment of $\Dee$. Then $\frac{\rho_{\bersumP}}{\var(\bersumP)^{3/2}}\le 8(3\sqrt{3}+\gamma)$.
%\end{lemma}
\end{restatable}

With this result, we can apply Lemma~\ref{highprobstd} to our setting to privately achieve an estimate $\widehat{\sigma}_{p,k}^2$ that is close to the true population-level variance $\sigp$, as shown in Lemma \ref{multiplicativevariance}. Note that as $k$ grows large, the allowable range for $p$ approaches the full support $[0,1]$ and the allowable standard deviation $\sigma_p$ approaches any non-negative number.

Lemma \ref{multiplicativevariance} combines the two previouse results to show that Lemma \ref{highprobstd} can be applied to the individual reports $\widehat{p}_i^k$ from the top $\log n$ users, and the resulting variance estimate will satisfy the accuracy conditions of Theorem \ref{metatheorem}.

\begin{restatable}{lemma}{lemmultvariance}
%\begin{lemma}
\label{multiplicativevariance}
Given $\sigma_{\min}<\sigma_{\max}\in[0, \infty], \epsilon>0, \delta\in(0,\frac{1}{n}], \beta\in(0,1/2)$, and $\zeta>0,$ let $\initialvariance$ be the $(\epsilon,\delta)$-differentially private mechanism given by Lemma~\ref{highprobstd}, and let $\widehat{\sigma}_{p,k}^2 = \mathcal{M}(\widehatp^k_1, \cdots, \widehatp^k_{\log n/\eps})$,
where $\widehatp^k_1, \cdots, \widehatp^k_{\log n/\eps}\sim\bersumP$. If there exists $\zeta>0$ such that $\frac{\rho_{\Dee}}{\sigma_p^3}\le \zeta$ where $\rho_{\Dee} = \mathbb{E}_{x\sim\Dee}[|x-p|^3]$, $\sqrt{\frac{1}{k}p(1-p)+\frac{k-1}{k}\sigma_p^2}\in[\sigma_{\min}, \sigma_{\max}]$, $\sigma_p>\frac{1}{k}$, $p\in \left[\frac{1}{k}, 1-\frac{1}{k}\right]$,
and $\log n\ge c (8(3\sqrt{3}+\zeta))^2\min\{\ln(\frac{\ln(\frac{\sigma_{\max}}{\sigma_{\min}})}{\beta}), \ln(\frac{1}{\delta\beta})\}$,
 then with probability $1-\beta$,
 $\widehat{\sigma}_{p,k}^2\in [\var(\Dee(k)), 8\var(\Dee(k))]$.
%\end{lemma}
\end{restatable}

%%%%%%%%%%%%%%%%%%%%%%%%%%%%%%%%%%%%%%%%%%%%%%%%%%%%%%%%%%%%

\bibliography{refs.bib}

\bibliographystyle{abbrvnat}

\newpage
\appendix
\onecolumn

\section{Proofs from Section~\ref{sec:modelsandprelims}}\label{appendix:modelsandprelims}

\sigmai*

\showproofs{
\begin{proof}[Proof of Lemma~\ref{lem.sigi}] 
Firstly, note that,
\[\sigma_p^2 = \mathbb{E}_{x \sim \Dee}[x^2]-p^2 \le \mathbb{E}_{x \sim \Dee}[x]-p^2 = p(1-p),\]
where the inequality follows from the fact that $\Dee$ is supported on $[0,1]$.

Next, \[\mathbb{E}[x_i] = \int_{x=0}^1 \Pr(p_i=x)\Pr(\Ber(x)=1)dx = \int_{x=0}^1 \Pr(p_i=x)x dx = p,\] which by linearity of expectation implies that $\mathbb{E}[\bersumP] = p$.

By the Law of Total Variation, the variance of $\widehatpi$ is: 
\begin{align*} \var(\widehatpi) &= \E_{p_i}[\var_{x_i}(\widehatpi| p_i)] + \var_{p_i}(\E_{x_i}[\widehatpi | p_i]) \\
&= \E_{p_i}[\frac{1}{k_i} p_i (1-p_i)] + \var_{p_i}(p_i)\\
&= \frac{1}{k_i}(p - \sigp - p^2) + \sigp \\
&=  \frac{1}{k_i} (p - p^2) + (1-\frac{1}{k_i})\sigp.\\
&= \frac{1}{k_i} \var(\Ber(p))+(1-\frac{1}{k_i})\sigp.
\end{align*}
\end{proof}}

\section{Proofs from Section~\ref{s.finalest}}\label{appendix.metathm}

First, let us show that the conditions of Theorem~\ref{metatheorem} imply that the variance and truncation parameter estimates of each individual data subject are correct up to constant factors.

\finalvariance*

\begin{proof}[Proof of Lemma~\ref{lem.finalvariance}]
Note that $\widehatsigp$ is actually an estimate of the variance of $\Dee(k_{\variancesamplesize})$ since it has access to samples from this distribution rather than $\Dee$ itself. Therefore, $\widehat{\sigma}_p^2\in\left[\var(\bersumPi{k_{\variancesamplesize}}), 8\cdot \var(\bersumPi{k_{\variancesamplesize}}) \right]$ implies $\widehat{\sigma}_{p}^2\in\left[\sigma_p^2, 8\left(\frac{1}{k_{\variancesamplesize}}p(1-p)+\sigma_p^2\right)\right]$. Then for every $i\geq \variancesamplesize$ (i.e., with $k_i\le k_{\variancesamplesize}$),
\begin{align*}
\widehatsigi &= \frac{1}{k_i}\pinitial(1-\pinitial)+\frac{k_i-1}{k_i}\widehat{\sigma}_{p}^2 \\
&\ge \frac{1}{k_i}\frac{1}{2}p(1-p)+\frac{k_i-1}{k_i}\sigma_p^2\\
&\ge \frac{1}{2}\left( \frac{1}{k_i}p(1-p)+\frac{k_i-1}{k_i}\sigma_p^2\right)\\
&= \frac{1}{2} \sigma_i^2,
\end{align*}
where the first inequality follows from the accuracy conditions on $\meanest$ and $\varianceest$ in Theorem \ref{metatheorem}, and the last equality follows from the definition of $\sigma_i^2$ in Lemma \ref{lem.sigi}. 
Also, 
\begin{align*}
\widehatsigi & =\frac{1}{k_i}\pinitial(1-\pinitial)+\frac{k_i-1}{k_i}\widehat{\sigma}_{p}^2 \\
&\le \frac{1}{k_i}\frac{3}{2}p(1-p)+8\frac{k_i-1}{k_i}\left(\frac{1}{k_{\log n}}p(1-p)+\sigma_p^2\right)\\
&= \left(\frac{3}{2}+8\frac{k_i-1}{k_{\log n}}\right)\frac{1}{k_i}p(1-p)+8\frac{k_i-1}{k_i}\sigma_p^2\\
&\le 9.5 \left( \frac{1}{k_i}p(1-p)+\frac{k_i-1}{k_i}\sigma_p^2\right)\\
&= 9.5\sigma_i^2,
\end{align*}
where again, the first inequality follows from the accuracy conditions on $\meanest$ and $\varianceest$ in Theorem \ref{metatheorem}, 
and the last equality follows from the definition of $\sigma_i^2$ in Lemma \ref{lem.sigi}. The intermediate steps are simply algebraic manipulations.  These two facts give us the desired bounds on $\widehatsigi$.

Next we turn to the truncation parameters $\widehat{a_i}$ and $\widehat{b_i}$. Using the definition of $\widehat{a_i}$ in Algorithm~\ref{alg.dp}, we have,
\begin{align*}
\widehat{a_i} &= \pinitial-\alpha-f^{k_i}_{\Dee}(n, \widehat{\sigp}, \beta/2) \\
&\le p-f_{\Dee}^{k_i}(n, \widehat{\sigp}, \beta/2))\\
&\le p-\concentrationbound{k_i}{\Dee}{n}{\sigma_p^2}{\beta}\\
&= a_i,
\end{align*}
where the two inequalities respectively follow from the accuracy conditions on $\meanest$ and $\varianceest$ in Theorem \ref{metatheorem}. A symmetric result that $\widehat{b_i} \geq b_i$ follows similarly.

Finally,  
\begin{align*}
|\widehat{b_i}-\widehat{a_i}| &= 2\alpha+2\concentrationbound{k_i}{\Dee}{n}{\widehat{\sigma_p^2}}{\beta}\\
&\le 2\concentrationbound{k_i}{\Dee}{n}{\sigma_p^2}{\beta}+2\concentrationbound{k_i}{\Dee}{n}{\sigma_p^2}{\beta}\\
&=4\concentrationbound{k_i}{\Dee}{n}{\sigma_p^2}{\beta}\\
&= 4|b_i-a_i|.
\end{align*}
The inequalities again follows from the accuracy conditions on $\meanest$ and $\varianceest$ in Theorem \ref{metatheorem}.
\end{proof}

\metatheorem*

\begin{proof}[Proof of Theorem~\ref{metatheorem}]

To see that Algorithm \ref{alg.dp} is differentially private, consider the three cohorts into which users are placed. The first cohort, containing the $n/10$ users with the smallest $k_i$ will have their data used in $\meanest$, which is $(\eps,\delta)$-DP. Similarly, the second cohort containing the $\variancesamplesize$ users with the largest $k_i$ will have their data used in $\varianceest$, which is also $(\eps,\delta)$-DP. The intermediate estimators of $\widehatsigi$, $\hat{T}^*$, $\hat{a}_i$, $\hat{b}_i$, and sensitivity $\Lambda$ are all computed as post-processing on the private outputs of these initial estimation subroutines and on the public $k_i$s, and thus do not incur any additional privacy cost. The third cohort contains the middle users $i \in [\variancesamplesize + 1, 9n/10]$. These users' data are only used in the final estimate, which is an $(\eps,0)$-DP instantiation of the Laplace Mechanism \citep{Dwork:2006}. 

Since these cohorts are disjoint and private algorithms are applied to each cohort's data separately, parallel composition applies, and the overall privacy parameters are the maximum of those experienced by any cohort, so the overall algorithm is $(\eps,\delta)$-DP.

For accuracy of the $\prealistic$ estimator produced by Algorithm \ref{alg.dp}, first notice that under the assumption that $\frac{k_{\max}}{\medk}\le\frac{n/2-\variancesamplesize}{\variancesamplesize}$, if $\sigma_{k_{\max}}^2 = \var(\widehatp_1)$ and $\sigma_{\medk}^2 = \var(\widehatp_{n/2})$ then \[\sigma_{\medk}^2 = \frac{1}{\medk}p(1-p)+\left(1-\frac{1}{\medk}\right)\sigma_p^2\le \frac{n/2-\variancesamplesize}{\variancesamplesize}\frac{1}{k_{\max}}p(1-p)+\left(1-\frac{1}{k_{\max}}\right)\sigma_p^2\le \frac{n/2-\variancesamplesize}{\variancesamplesize} \sigma_{k_{\max}}^2.\]
Therefore, for any truncation parameter $T$,
\begin{align}
\nonumber\frac{1}{2} \sum_{i=1}^n \min\left\{\frac{1}{\sigma_i^2}, \frac{T}{\sigma_i}\right\} &\le \sum_{i=1}^{n/2} \min\left\{\frac{1}{\sigma_i^2}, \frac{T}{\sigma_i}\right\}\\
\nonumber&= \sum_{i=1}^{\variancesamplesize} \min\left\{\frac{1}{\sigma_i^2}, \frac{T}{\sigma_i}\right\} + \sum_{i=\variancesamplesize+1}^{n/2} \min\left\{\frac{1}{\sigma_i^2}, \frac{T}{\sigma_i}\right\} \\
\nonumber&\le \variancesamplesize \cdot \min\left\{\frac{1}{\sigma_{k_{\max}}^2}, \frac{T}{\sigma_{k_{\max}}}\right\} + \sum_{i=\variancesamplesize+1}^{n/2} \min\left\{\frac{1}{\sigma_i^2}, \frac{T}{\sigma_i}\right\}\\
\nonumber&\le (n/2-\variancesamplesize) \cdot \min\left\{\frac{1}{\sigma_{\medk}^2}, \frac{T}{\sigma_{\medk}}\right\} + \sum_{i=\variancesamplesize+1}^{n/2} \min\left\{\frac{1}{\sigma_i^2}, \frac{T}{\sigma_i}\right\}\\
\nonumber&\le 2 \sum_{i=\variancesamplesize+1}^{n/2} \min\left\{\frac{1}{\sigma_i^2}, \frac{T}{\sigma_i}\right\}\\
&\le 2 \sum_{i=\variancesamplesize+1}^{9n/10} \min\left\{\frac{1}{\sigma_i^2}, \frac{T}{\sigma_i}\right\},\label{removingdata}
\end{align}
where the first, second, and fourth inequalities follow from our assumed ordering on the $k_i$s. The third inequality comes from our assumption on $k_{\max}$ and $\medk$, and the final inequality follows from the fact that the summands $\min\{\frac{1}{\sigma_i^2}, \frac{T}{\sigma_i}\}$ are positive so adding more terms only increases the sum.

Therefore, 
\begin{align*}
\var(\prealistic) &= \frac{1}{{(\sum_{j=\variancesamplesize+1}^{9n/10} \min\{1/\widehat{\sigma_j}^2, \frac{\widehat{T}^*}{\widehat{\sigma_i}}\})^2}}\left(\sum_{i=\variancesamplesize+1}^{9n/10}\min\{\frac{1}{\widehat{\sigma_i}^4}, \frac{\widehat{T}^{*2}}{\widehat{\sigma_i}^2}\}\sigma_i^2+\max_{i} \frac{\min\{\frac{1}{\widehat{\sigma_i}^4}, {\frac{\widehat{T}^{*2}}{\widehat{\sigma_i}^2}}\}|\widehat{b_i}-\widehat{a_i}|^2}{\epsilon^2} \right)\\
&\le \frac{1}{{(\sum_{j=\variancesamplesize+1}^{9n/10} \min\{1/\widehat{\sigma_j}^2, \frac{\widehat{T}^*}{\widehat{\sigma_i}}\})^2}}\left(\sum_{i=\variancesamplesize+1}^{9n/10}\min\{\frac{1}{\widehat{\sigma_i}^4}, \frac{\widehat{T}^{*2}}{\widehat{\sigma_i}^2}\}2\widehat{\sigma_i}^2+\max_{i} \frac{\min\{\frac{1}{\widehat{\sigma_i}^4}, \frac{\widehat{T}^{*2}}{\widehat{\sigma_i}^2}\}|\widehat{b_i}-\widehat{a_i}|^2}{\epsilon^2} \right)\\
&\le 2\frac{1}{{(\sum_{j=\variancesamplesize+1}^{9n/10} \min\{1/\widehat{\sigma_j}^2, \frac{\widehat{T}^*}{\widehat{\sigma_i}}\})^2}}\left(\sum_{i=\variancesamplesize+1}^{9n/10}\min\{\frac{1}{\widehat{\sigma}_i^2}, \widehat{T}^{*2}\}+\max_{i} \frac{\min\{\frac{1}{\widehat{\sigma_i}^4}, \frac{\widehat{T}^{*2}}{\widehat{\sigma_i}^2}\}|\widehat{b_i}-\widehat{a_i}|^2}{\epsilon^2} \right)\\
&\le 2\frac{1}{{(\sum_{j=\variancesamplesize+1}^{9n/10} \min\{1/\widehat{\sigma_j}^2, \frac{T^*}{\widehat{\sigma_i}}\})^2}}\left(\sum_{i=\variancesamplesize+1}^{9n/10}\min\{\frac{1}{\widehat{\sigma}_i^2}, T^{*2}\}+\max_{i} \frac{\min\{\frac{1}{\widehat{\sigma_i}^4}, \frac{T^{*2}}{\widehat{\sigma_i}^2}\}|\widehat{b_i}-\widehat{a_i}|^2}{\epsilon^2} \right)\\
&\le 2\frac{1}{{(\sum_{j=\variancesamplesize+1}^{9n/10} \min\{1/10\sigma_j^2, \frac{\sqrt{2}T^*}{\sigma_i}\})^2}}\left(\sum_{i=\variancesamplesize+1}^{9n/10}\min\{\frac{2}{\sigma_i^2}, T^{*2}\}+\max_{i} \frac{\min\{\frac{4}{\sigma_i^4}, \frac{2T^{*2}}{\sigma_i^2}\}6|b_i-a_i|^2}{\epsilon^2} \right)\\
&\le 240\frac{1}{{(\sum_{j=\variancesamplesize+1}^{9n/10} \min\{1/\sigma_j^2, \frac{T^*}{\sigma_i}\})^2}}\left(\sum_{i=\variancesamplesize+1}^{9n/10}\min\{\frac{1}{\sigma_i^2}, T^{*2}\}+\max_{i} \frac{\min\{\frac{1}{\sigma_i^4}, \frac{T^{*2}}{\sigma_i^2}\}|b_i-a_i|^2}{\epsilon^2} \right)\\
&\le 240\frac{1}{{\frac{1}{16}(\sum_{j=1}^n \min\{1/\sigma_j^2, \frac{T^*}{\sigma_i}\})^2}}\left(\sum_{i=1}^n\min\{\frac{1}{\sigma_i^2}, T^{*2}\}+\max_{i} \frac{\min\{\frac{1}{\sigma_i^4}, \frac{T^{*2}}{\sigma_i^2}\}|b_i-a_i|^2}{\epsilon^2} \right)\\
&= 3840 \cdot\var(\pideal)
\end{align*}

The first equality simply follows from the definition of the estimator and basic properties of the variance, as well as the fact that $\var([\widehatpi]_{a_i}^{b_i})\le\sigma_i$. The first inequality follows from the fact that $\sigma_i^2\le 2\hat{\sigma_i}^2$, which was shown in Lemma~\ref{lem.finalvariance}. The second inequality is simply pulling out the constant to the front. The third inequality follows from the definition of $\hat{T}$ as the optimiser of the variance using the approximations $\widehat{\sigma_i}^2$, $\hat{b_i}$ and $\hat{a_i}$. The fourth inequality follows from the fact that $\widehat{\sigma_i}^2\in\left[\frac{1}{2}\sigma_i^2, 10\sigma_i^2\right]$ and $|\widehat{b}_i-\widehat{a}_i|\leq 4|b_i-a_i|$, as shown in Lemma~\ref{lem.finalvariance}, and will hold with probability $1-2\beta$, by taking a union bound over the $\beta$ failure probabilities from each of the $\meanest$ and $\varianceest$ subroutines. The fifth inequality simply pulls out the constants (240=10*4*6). The final inequality follows from Equation~\eqref{removingdata} above. The final equality follows from definition of $\pideal$ and the assumption that $\frac{1}{2}\sigma_i^2\le\var([\widehatpi]_{a_i}^{b_i})$.

\end{proof}

\section{Proofs from Section~\ref{s.privatek}}

\begin{proof}[Proof of privacy claim in Theorem~\ref{thm.privatek}]
Let us begin with the privacy proof. The population is broken into three cohorts. Let us consider each cohort individually. First, consider the $\variancesamplesize$ individuals with the most data. They participate in private releases in lines~\eqref{initialvark} ($(\epsilon, \delta)$-DP), and \eqref{EMk} ($\epsilon$-DP). Using the simple composition rule of differential privacy \citep{Dwork:2006}, Algorithm~\ref{alg.dpk} is $(2\eps,\delta)$-DP with respect to these users.

Next, consider the $1/10$th of users with the least data. These users participate in lines~\eqref{initialk} ($(\epsilon, \delta)$-DP) and \eqref{EMk} ($\epsilon$-DP). Again using the simple composition rule of differential privacy, Algorithm~\ref{alg.dpk} is $(2\eps,\delta)$-DP with respect to these users.

Finally, let us consider the the group consisting of users $i\in[\variancesamplesize+1, 9n/10]$. These users first participate in line~\eqref{EMk} ($\epsilon$-DP). The post-processing guarantee of differential privacy states that we can now use these statistics in the subsequent computations without paying additionally for their privacy. Lines~\eqref{startpreprocessing}~-~\eqref{endpreprocessing} are pre-processing for the computation of $\tilde{N}$. The algorithm releasing $\tilde{N}$ is a simple application of the Laplace mechanism since each $v_i\in[0,\frac{1}{\widehat{\sigma_{\min}}^2}]$, and hence is $\epsilon$-differentially private. The computation of $\Lambda$ in line~\eqref{senscomp} does not additionally touch the users data. 
The final estimate $\prealistic$ is an application of the propose-test-release framework on the function $\mathcal{M}(\cdot\; ;  \widehat{k_{T}}, n, \pinitial, \widehatsigp)$ with proposed sensitivity $\Lambda$. This is a generic application of the propose-test-release framework, so we refer the reader to \citep{Lei:2009} for a proof that this final step of the algorithm is $(\epsilon, \delta)$-differentially private. Therefore, again using the composition theorem, Algorithm~\ref{alg.dpk} is $(3\epsilon, 2\delta)$-DP with respect to this final set of users.
\end{proof}

\PTRfailureratelem*

\begin{proof}[Proof of Lemma \ref{PTRfailurerate}] Let $\sigma_{\max}^2 = \tfrac{1}{k_{\min}} \hat{p}(1 - \hat{p}) + (1-\tfrac{1}{k_{\min}})\hat{\sigma_p^2}$, $\sigma_{\min}^2 = \tfrac{1}{k_{\max}} \hat{p}(1 - \hat{p}) + (1-\tfrac{1}{k_{\max}})\hat{\sigma_p^2}$, $v_{\max} = 1/\sigma_{\min}^2$ and $v_{\min} = 1/\sigma_{\max}^2$.
Note that as in Equation~\eqref{limitonk}, the condition that $k_{\max}/k_{\min}\le A$ implies that $\sigma_{\max}^2\le A\sigma_{\min}^2$ and, equivalently, $v_{\max}\le A v_{\min}$. 

Let $D=\{(\widehatpi, k_i)\}_{i=1}^n$ be a dataset of size $n$ where each $\widehatpi\sim \Dee(k_i)$ where $\Dee$ has mean $p$ and variance $\sigma_p^2$. 
It suffices to show that for any database $D'$, which is a $\kneighbor$-neighbour of $D$ where $0\le \kneighbor\le \Upsilon+1$, and any $j\in[n]$, if $D'_{-j}$ is $D'$ where the data of the $j$th data subject has been removed, then, \begin{equation}\label{LSDsample} \left|\mathcal{M}(D'; k_{\max}, n, \hat{p}, \hat{\sigma_p}^2, \alpha)-\mathcal{M}(D'_{-j}; k_{\max}, n, \hat{p}, \hat{\sigma_p}^2, \alpha) \right|\le 6\frac{v_{k_{\max}}\concentrationbound{k_{\max}}{\Dee}{n}{\hat{\sigma_p}^2}{\beta}}{\sum_{i=1}^{n} v_i}.\end{equation}
The final result is then a simple application of the triangle inequality.

Our proof that Equation~\eqref{LSDsample} holds with high probability for all $\kneighbor$-neighbours of $D$ relies on the fact that with probability $1-\beta$, $D$ is such that all subsets $S$ of $D$ of size at least $m\ge n-\Upsilon-1$, $\mathcal{M}(S; k_{\max}, m, \hat{p}, \hat{\sigma_p}^2)$ is concentrated around $p$. Let $I$ be a subset of $[n]$ of size $n-\kneighbor$ where $\kneighbor\le\Upsilon+1$.
Then
\begin{align*}
\var(\mathcal{M}(S; k_{\max}, n-\kappa, \hat{p}, \hat{\sigma_p}^2, \alpha))&\le\var\left(\frac{\sum_{i\in I} v_i[\widehatpi]_{\widetilde{a_i}}^{\widetilde{b_i}}}{\sum_{i\in I} v_i}\right) \\
&= \frac{\sum_{i\in I} \frac{1}{({\widetilde{\sigma_i^2}})^2}\sigma_i^2}{\left(\sum_{i\in I} \frac{1}{{\widetilde{\sigma_i^2}}}\right)^2}\\
&\le2\frac{1}{\sum_{i\in I} \frac{1}{{\widetilde{\sigma_i^2}}}}\\
&\le 2\frac{1}{\frac{n-\kneighbor}{A} \frac{1}{\sigma_{\min}^2}}\\
&= \frac{2A}{n-\kneighbor}\sigma_{\min}^2\\
\end{align*}
where the first inequality follows from $\sigma_i^2\le 2{\widetilde{\sigma_i^2}}$ by Lemma~\ref{lem.finalvariance} and the definition of $\widetilde{\sigma_i^2}$ from line~\eqref{truncatedvariance} of Algorithm~\ref{alg.PTR}. The second from the fact that ${\widetilde{\sigma_i^2}}\le A\sigma_{\min}^2$ for all $i\in[n]$. Let $\Gamma=\sum_{\kneighbor=0}^{\Upsilon +1}\binom{n}{\kneighbor}$ be the number of subsets of $D$ of size greater than $n-\Upsilon-1$.
By the concentration assumption on $\mathcal{M}(S; k_{\max}, n-\kappa, \hat{p}, \hat{\sigma_p}^2)$, with probability $1-\frac{\beta}{\Gamma}$, \begin{equation}\label{concentrationofsubsets}\left|\mathcal{M}(S; k_{\max}, n-\kappa, \hat{p}, \hat{\sigma_p}^2, \alpha))-p\right|\le \sigma_{\min} \sqrt{\log \frac{\Gamma}{\beta}\frac{2A}{n-\kappa}} \le \sigma_{\min}\sqrt{\log\frac{n}{\beta}} \end{equation}
Note that $\Gamma\le n^{\Upsilon+1}$ so the second inequality follows from the conditions on $A$.
Applying a union bound, with probability $1-\beta$, eqn~\eqref{concentrationofsubsets} holds simultaneously for all subsets of $D$ of sufficiently large size. For the remainder of the proof, let us assume that this holds.

Let $D'$ be a $\kneighbor$-neighbour of $D$ where $0\le \kneighbor\le \Upsilon+1$. Without loss of generality, assume that $D'~=~\{(\widehatpi',k_i
')\}_{i=1}^{n}$ where $(\widehatpi', k_i')=(\widehatpi, k_i)$ for $i\in[n-\kappa]$. In order to use this simplification, we will not assume that the $k_i'$ are in descending order. Let the $v_i$ be the un-normalised weights corresponding to $D'$, as defined in line~\eqref{unnormalweight} of Algorithm~\ref{alg.PTR}. Note that the $v_i$ depends only on the data of user $i$, not the data of any other individual in the data set. Then
\begin{align}\nonumber\Big|\mathcal{M}(D'; k_{\max}, n, \hat{p}, \hat{\sigma_p}^2, \alpha)-&\mathcal{M}(D'_{-j}; k_{\max}, n-1, \hat{p}, \hat{\sigma_p}^2, \alpha) \Big|=
\left|\frac{\sum_{i=1}^n v_i\widehatpi'}{\sum_{i=1}^n v_i}-\frac{\sum_{i=1, i\neq j}^{n} v_i\widehatpi'}{\sum_{i=1, i\neq j}^{n} v_i}\right|\\
\nonumber&= \frac{v_j}{\sum_{i=1}^{n} v_i}\left|\widehatp_j'-\frac{\sum_{i=1, i\neq j}^{n} v_i\widehatpi'}{\sum_{i=1, i\neq j}^{n} v_i}\right|\\
&\le\frac{v_j}{\sum_{i=1}^{n} v_i}\left(\left|\widehatp_j'-\frac{\sum_{i=1}^{n-\kneighbor} v_i\widehatpi}{\sum_{i=1}^{n-\kneighbor} v_i}\right|+\left|\frac{\sum_{i=1}^{n-\kneighbor} v_i\widehatpi}{\sum_{i=1}^{n-\kneighbor} v_i}-\frac{\sum_{i=1, i\neq j}^{n} v_i\widehatpi'}{\sum_{i=1, i\neq j}^{n} v_i}\right|\right).\label{normalsensitivity}
\end{align}

We will bound the two terms separately. 
For the first term in Equation~\eqref{normalsensitivity}, we will use the fact that $\frac{\sum_{i=1}^{n-\kappa} v_i\widehatpi'}{\sum_{i=1}^{n-\kappa} v_i}$ is concentrated around $p$, and $\hat{p}_j'$ is truncated to within $\alpha+\concentrationbound{\widetilde{k_{j}}}{\Dee}{n}{\sigma_p^2}{\beta}$ of $p$. So
\[\left|\widehatp_j'-\frac{\sum_{i=1}^{n-\kappa} v_i\widehatpi}{\sum_{i=1}^{n-\kappa} v_i}\right|\le\max\left\{2(\alpha+\concentrationbound{\widetilde{k_{j}}}{\Dee}{n}{\sigma_p^2}{\beta}), \sigma_{\min}\sqrt{\log\frac{n}{\beta}}\right\}\le 4\concentrationbound{\widetilde{k_{j}}}{\Dee}{n}{\sigma_p^2}{\beta},\] where the second inequality follows since $\max\{\alpha, \frac{1}{2}\sigma_{\min}\sqrt{\log\frac{n}{\beta}}\}\le \concentrationbound{\widetilde{k_{j}}}{\Dee}{n}{\sigma_p^2}{\beta}$ and Eqn~\eqref{concentrationofsubsets}.

Next, let us handle the second term in Equation~\eqref{normalsensitivity}. Assume that $j=n$ to simplify notation:
\begin{align*}
    \left|\frac{\sum_{i=1}^{n-1} v_i\widehatpi'}{\sum_{i=1}^{n-1} v_i} - \frac{\sum_{i=1}^{n-\kappa} v_i\widehatpi'}{\sum_{i=1}^{n-\kappa} v_i }\right| &= \left|\frac{\sum_{i=n-\kappa+1}^{n-1} v_i}{\sum_{i=1}^{n-1} v_i}\left(\frac{\sum_{i=n-\kappa+1}^{n-1} v_i\widehatpi'}{\sum_{i=n-\kappa+1}^{n-1} v_i}- \frac{\sum_{i=1}^{n-\kappa} v_i\widehatpi'}{\sum_{i=1}^{n-\kappa} v_i}\right)\right|\\
    &\le\left(\frac{\sum_{i=n-\kappa+1}^{n-1} v_i}{\sum_{i=1}^{n-1} v_i}\right) \left(\frac{\sum_{i=n-\kappa+1}^{n-1} v_i\left|\widehatpi'- \frac{\sum_{i=1}^{n-\kappa} v_i\widehatpi'}{\sum_{i=1}^{n-\kappa} v_i}\right|}{\sum_{i=n-\kappa+1}^{n-1} v_i}\right)\\
    &\le \left(\frac{\sum_{i=n-\kappa+1}^{n-1} v_i }{\sum_{i=1}^{n-1} v_i}\right)\left( \frac{\sum_{i=n-\kappa+1}^{n-1} v_i \max\{(2\alpha+2\concentrationbound{\widetilde{k_i}}{\Dee}{n}{\widehatsigp}{\beta}), \sigma_{\min}\sqrt{\log\frac{n}{\beta}}\}}{\sum_{i=n-\kappa+1}^{n-1} v_i}\right)\\
    &\le \left(\frac{\sum_{i=n-\kappa+1}^{n-1} v_i}{\sum_{i=1}^{n-1} v_i}\right)\left( \frac{\sum_{i=n-\kappa+1}^{n-1} v_i 4\concentrationbound{\widetilde{k_i}}{\Dee}{n}{\widehatsigp}{\beta}}{\sum_{i=n-\kappa+1}^{n-1} v_i}\right)\\
    &\le \left(\frac{4\sum_{i=n-\kappa+1}^{n-1} v_i \concentrationbound{\widetilde{k_i}}{\Dee}{n}{\widehatsigp}{\beta}}{\sum_{i=1}^{n-1} v_i}\right)\\
    &\le 4\kappa \left(\frac{\max_i v_i \concentrationbound{\widetilde{k_i}}{\Dee}{n}{\widehatsigp}{\beta}}{\sum_{i=1}^{n-1} v_i}\right)\\
\end{align*}
By assumption, $\max_i v_i\concentrationbound{\widetilde{k_i}}{\Dee}{n}{\sigma_p^2}{\beta} \le v_{k_{\max}}\concentrationbound{k_{\max}}{\Dee}{n}{\sigma_p^2}{\beta}$. Also, 
$\sigma_{k_i'}^2 \le A \sigma_{k_{\max}}^2$ so we have $\sum_{i=1}^{n-1} v_i\ge \frac{n-1}{A}v_{k_{\max}}$. Therefore, 

\begin{align*}
    \left|\frac{\sum_{i=1}^{n-1} v_i\widehatpi'}{\sum_{i=1}^{n-1} v_i} - \frac{\sum_{i=1}^{n-\kappa} v_i\widehatpi'}{\sum_{i=1}^{n-\kappa} v_i }\right| 
    &\le 4\kappa \frac{v_{k_{\max}}\concentrationbound{k_{\max}}{\Dee}{n}{\sigma_p^2}{\beta}}{\frac{n-1}{A}v_{k_{\max}}}\\
    &\le \frac{4\kappa A }{n-1} \concentrationbound{k_{\max}}{\Dee}{n}{\sigma_p^2}{\beta}. \\
    &\le 2 \concentrationbound{k_{\max}}{\Dee}{n}{\sigma_p^2}{\beta}
\end{align*}
where the second inequality follows from  $A\le\frac{n-1}{2(\Upsilon+1)}\le\frac{n-1}{2\kappa}$, which holds by assumption.
Therefore, 
\begin{align*}
\Big|\mathcal{M}(D'; k_{\max}, n, \hat{p}, \hat{\sigma_p}^2, \alpha)-\mathcal{M}(D'_{-j}; k_{\max}, n-1, \hat{p}, \hat{\sigma_p}^2, \alpha) \Big|
&\le \frac{v_j}{\sum_{i=1}^{n} v_i} \cdot  (4\concentrationbound{\widetilde{k_j}}{\Dee}{n}{\hat{\sigma_p}^2}{\beta}+2\concentrationbound{k_{\max}}{\Dee}{n}{\hat{\sigma_p}^2}{\beta})\\
&\le 6 \frac{v_j}{\sum_{i=1}^{n} v_i} \cdot  \concentrationbound{\widetilde{k_j}}{\Dee}{n}{\hat{\sigma_p}^2}{\beta}
\end{align*}
Taking the max over $j$, we again have that $\max_j v_j\concentrationbound{\widetilde{k_j}}{\Dee}{n}{\hat{\sigma_p}^2}{\beta} \le v_{k_{\max}}\concentrationbound{k_{\max}}{\Dee}{n}{\hat{\sigma_p}^2}{\beta}$ so for all $j$,
\begin{align*}
\Big|\mathcal{M}(D'; k_{\max}, n, \hat{p}, \hat{\sigma_p}^2, \alpha)-\mathcal{M}(D'_{-j}; k_{\max}, n-1, \hat{p}, &\hat{\sigma_p}^2, \alpha) \Big|
\le 6 \frac{v_{k_{\max}}\concentrationbound{k_{\max}}{\Dee}{n}{\hat{\sigma_p}^2}{\beta}}{\sum_{i=1}^{n} v_i}
\end{align*}
\end{proof}

\begin{lemma}\label{exponentialmedian}
Given $\epsilon>0$, $\delta\in[0,1]$, $\beta\in[0,1]$, $k_{\max}\in\mathbb{N}$, and $\variancesamplesize\in[n]$, there exists a mechanism $\expmech(k_1, \cdots, k_n; \variancesamplesize, k_{\max})$ which is $(\epsilon, \delta)$-DP, and with probability $1-\beta$, and outputs $\widehat{k_{\variancesamplesize}}$ such that \[k_{\variancesamplesize+\frac{1}{\epsilon}(\ln k_{\max}+\ln (1/\beta))}\le \widehat{k_{T}}\le k_{\variancesamplesize-\frac{1}{\epsilon}(\ln k_{\max}+\ln (1/\beta))}\]
\end{lemma}

\begin{proof}
There are several existing algorithms in the literature that can be used to privately estimate the $\variancesamplesize$-th  order statistic $\widehat{k_{\variancesamplesize}}$ with the desired accuracy. A simple algorithm \citep{Lei:2009, Thakurta:2013, Johnson:2013, alabi2020differentially, Asi:2020} that estimates the order statistic using the common differential privacy framework called the exponential mechanism \citep{McSherry:2007} is sufficient up to a constant factor. For a full description of this algorithm, as well as its accuracy guarantees see \citep{Asi:2020}.
\end{proof}

\privatek*

\begin{proof}[Proof of Theorem~\ref{thm.privatek}]
The main component remaining to prove is that truncating at $T=\frac{1}{\widehat{\sigma_{\min}}^2}$ rather than the optimal truncation does not affect the utility by more than a constant factor, under the assumptions of the theorem. Let $k_1\ge k_2\ge\cdots\ge k_n$. Firstly, we need to show that $\widehat{k_{\variancesamplesize}}$ is a sufficiently good estimate of $k_{\variancesamplesize}$. Lemma~\ref{exponentialmedian} provides us with a $\epsilon$-DP estimator of the $\variancesamplesize$-th order statistic that has the guarantee that with probability $1-\beta$, $k_{\variancesamplesize+\frac{1}{\epsilon}(\ln k_{\max}+\ln (1/\beta))}\le \widehat{k_{\variancesamplesize}}\le k_{\variancesamplesize-\frac{1}{\epsilon}(\ln k_{\max}+\ln (1/\beta))}$. Since by assumption $2\variancesamplesize\ge\frac{1}{\epsilon}(\ln k_{\max}+\ln (1/\beta))$, this implies that with probability $1-\beta$, $k_{\frac{1}{2}\variancesamplesize}\le \widehat{k_{\variancesamplesize}}\le k_{\frac{3}{2}\variancesamplesize}$. That is, only $\frac{1}{2}\variancesamplesize$ more data points than desired will be truncated. 

Next, we need to show that truncating at any point within this range provides an estimator with accuracy competitive with the optimal truncation. Assume the PTR component of the algorithm does not fail, the variance of $\prealisticunknown$ can be written as two terms, the variance that exists in the non-private setting, and the additional noise due to privacy;
\[\var(\prealisticunknown) =  \underbrace{\frac{\sum_{i=\variancesamplesize+1}^{9n/10} \min\left\{\frac{T^{2}}{\widehatsigi}, \frac{1}{\widetilde{\sigma_i}^4}\right\}\var([\widehatpi]_{\widetilde{a_i}}^{\widetilde{b_i}}])}{\left(\sum_{i=\variancesamplesize+1}^{9n/10} \min\left\{\frac{T}{\widetilde{\sigma_i}}, \frac{1}{\widetilde{\sigma^2_i}}\right\}\right)^2}}_{\text{non-private term}} + \underbrace{ \frac{\left(12\frac{\concentrationbound{\widehat{k_{\variancesamplesize}}}{\Dee}{n}{\hat{\sigma_p}^2}{\beta}}{\widehat{\sigma_{\min}}^2\widehat{N}}\right)^2}{\epsilon^2}}_{\text{private term}}. \] 
The truncation has opposite effects on each of these terms. As $T$ decreases, the private term decreases while the non-private term increases. When we set $T_{\variancesamplesize}=1/\var(\Dee(k_{\variancesamplesize+K}))$, where $K\in [-\frac{1}{2}\variancesamplesize, \frac{1}{2}\variancesamplesize]$ then if $K$ is negative, no truncation occurs and the non-private term is optimal. Even if $K$ is positive, only a small number of data points are truncated so the non-private term is still close to it's optimal value: 
\begin{align*}
\frac{\sum_{i=\variancesamplesize+1}^{9n/10} \min\left\{\frac{T_{\variancesamplesize}^{2}}{\widehatsigi}, \frac{1}{\widehat{\sigma_i}^4}\right\}\var([\widehatpi]_{\widehat{a_i}}^{\widehat{b_i}}])}{\left(\sum_{i=\variancesamplesize+1}^{9n/10} \min\left\{\frac{T_{\variancesamplesize}}{\widehat{\sigma_i}}, \frac{1}{\widehatsigi}\right\}\right)^2}
&\le O\left(\frac{\sum_{i=\variancesamplesize+K}^{9n/10}  \frac{1}{\widehat{\sigma_i}^4}\var([\widehatpi]_{\widehat{a_i}}^{\widehat{b_i}}])}{\left(\sum_{i=\variancesamplesize+K}^{9n/10} \frac{1}{\widehatsigi}\right)^2}\right)\\
&\le O\left(\frac{\sum_{i=\variancesamplesize+1}^{9n/10}  \frac{1}{\widehat{\sigma_i}^4}\var([\widehatpi]_{\widehat{a_i}}^{\widehat{b_i}}])}{\left(\sum_{i=\variancesamplesize+1}^{9n/10} \frac{1}{\widehatsigi}\right)^2}\right)\\
\end{align*}
where the first inequality follows from the same proof as Theorem~\ref{metatheorem}given $A\le\frac{n/2-3L/2}{3L/2}$. The truncation disappears on the right hand side of the first inequality since $T_{\variancesamplesize}^2\le \frac{1}{\widehat{\sigma_i^2}}$ for $i\ge \variancesamplesize+K$. The second inequality follows from the fact that adding more high quality data points only improves the variance of the estimator. Therefore, the non-private term in the variance is within a constant factor of optimal.

Next, we will show that under the conditions outlined in the theorem, the non-private term dominates the variance. The normalisation term also appears in the private term but as an approximation: \[\widehat{N} = \sum_{j=\variancesamplesize+1}^{9n/10} \min\left\{\frac{T_{\variancesamplesize}}{\widehat{\sigma_i}}, \frac{1}{\widehat{\sigma_i}^2}\right\}+\Lap\left(\frac{1}{\epsilon{\widehat{\sigma_{\min}}}^2}\right)-\frac{1}{\epsilon{\widehat{\sigma_{\min}}}^2}\ln(2\delta).\] With probability $1-\delta$, \begin{align*}
\widehat{N}&\ge \sum_{j=\variancesamplesize+1}^{9n/10} \min\left\{\frac{T_{\variancesamplesize}}{\widehat{\sigma_i}}, \frac{1}{\widehat{\sigma_i}^2}\right\}-2\frac{1}{\epsilon{\widehat{\sigma_{\min}}}^2}\ln(2\delta)\\
&\ge\sum_{j=\variancesamplesize+1}^{n/4} \min\left\{\frac{T_{\variancesamplesize}}{\widehat{\sigma_i}}, \frac{1}{\widehat{\sigma_i}^2}\right\}+\sum_{j=n/4+1}^{n/2} \min\left\{\frac{T_{\variancesamplesize}}{\widehat{\sigma_i}}, \frac{1}{\widehat{\sigma_i}^2}\right\}-2\frac{1}{\epsilon{\widehat{\sigma_{\min}}}^2}\ln(2\delta)\\
&\ge \sum_{j=\variancesamplesize+1}^{n/4} \min\left\{\frac{T_{\variancesamplesize}}{\widehat{\sigma_i}}, \frac{1}{\widehat{\sigma_i}^2}\right\}+ (n/4-1)\frac{1}{\widehat{\sigma}_{k_{\text{med}}}^2}-2\frac{1}{\epsilon{\widehat{\sigma_{\min}}}^2}\ln(2\delta)
\\
&\ge \sum_{j={\variancesamplesize}+1}^{n/4} \min\left\{\frac{T_{\variancesamplesize}}{\widehat{\sigma_i}}, \frac{1}{\widehat{\sigma_i}^2}\right\}\\
&\ge \frac{1}{2} \sum_{j={\variancesamplesize}+1}^{9n/10} \min\left\{\frac{T_{\variancesamplesize}}{\widehat{\sigma_i}}, \frac{1}{\widehat{\sigma_i}^2}\right\}
\end{align*} 
where the first inequality comes from high probability bounds on the Laplacian distribution, the second inequality is simply separating the sum into two pieces and removing the contribution of users $i\in[n/2+1,9n/10]$, the third inequality comes from the fact that any user with more than $k_{\text{med}}$ data points has weight larger than $1/\hat{\sigma}_{k_{\text{med}}}^2$. The fourth inequality follows from $\frac{\widehat{\sigma_{k_{\text{med}}}}}{\widehat{\sigma_{\min}}^2} \le \frac{(n/4-1)\epsilon}{3\ln(2/\delta)}.$
Now, let us turn to the proof that the non-private noise is dominant when $\epsilon$ is not too large. To see this note that the non-private term satisfies
\begin{align*}
\frac{\sum_{i=\variancesamplesize+1}^{9n/10} \min\left\{\frac{T_{\variancesamplesize}^{2}}{\widehatsigi}, \frac{1}{\widehat{\sigma_i}^4}\right\}\var([\widehatpi]_{\widetilde{a_i}}^{\widetilde{b_i}}])}{\left(\sum_{i={\variancesamplesize}+1}^{9n/10} \min\left\{\frac{T_{\variancesamplesize}}{\widehat{\sigma_i}}, \frac{1}{\widehatsigi}\right\}\right)^2} 
&\ge \Omega\left( \frac{\sum_{i={\variancesamplesize}+1}^{9n/10}  \min\{T_{\variancesamplesize}^2, \frac{1}{\sigma_i^2}\}}{\left(\sum_{i={\variancesamplesize}+1}^{9n/10} \min\{\frac{T_{\variancesamplesize}}{\sigma_i}, \frac{1}{\sigma_i^2}\}\right)^2}\right)
\end{align*}
where the inequality comes from noting that $[\widehat{a_i},\widehat{b_i}]\subset[\widetilde{a_i},\widetilde{b_i}]$, which implies $\var([\widehatpi]_{\widetilde{a_i}}^{\widetilde{b_i}}])\ge \var([\widehatpi]_{\widehat{a_i}}^{\widehat{b_i}}])\ge \frac{1}{2}\sigma_i^2$ and $\widehatsigi$ is within a constant multiplicative factor of $\sigma_i^2$. Further, let $N=\sum_{i={\variancesamplesize}+1}^{9n/10} \min\{\frac{T_{\variancesamplesize}}{\sigma_i}, \frac{1}{\sigma_i^2}\}$ so the private term satisfies
\begin{align*}
\frac{\left(12\frac{\concentrationbound{\widehat{k_{\variancesamplesize}}}{\Dee}{n}{\hat{\sigma_p}^2}{\beta}}{\widehat{\sigma_{\min}}^2\widehat{N}}\right)^2}{\epsilon^2} 
&= O\left(\frac{\log(n/\beta)}{\widehat{\sigma_{\min}^2} N^2 \epsilon^2} \right)\\
&= O\left(\frac{\log(n/\beta)}{\sigma_{\min}^2 \epsilon^2 (\sum_{i=\variancesamplesize+1}^{9n/10} \min\{\frac{T_{\variancesamplesize}}{\sigma_i}, \frac{1}{\sigma_i^2}\})^2 }\right)\\
\end{align*}
Now, comparing these two terms we can see that the non-private term dominates when: \[\frac{\sum_{i=\variancesamplesize+1}^{9n/10}  \min\{T^2_{\variancesamplesize}, \frac{1}{\sigma_i^2}\}}{\left(\sum_{i=\log n+1}^{9n/10} \min\{\frac{T_{\variancesamplesize}}{\sigma_i}, \frac{1}{\sigma_i^2}\}\right)^2}=\Omega\left(\frac{\log(n/\beta)}{\sigma_{\min}^2 \epsilon^2 (\sum_{i=\variancesamplesize+1}^{9n/10} \min\{\frac{T_{\variancesamplesize}}{\sigma_i}, \frac{1}{\sigma_i^2}\})^2 }\right).
\]
That is, when: \[\sum_{i=\variancesamplesize+1}^{9n/10}  \min\left\{T^2_{\variancesamplesize}, \frac{1}{\sigma_i^2}\right\}\ge \Omega\left(\frac{\log(n/\beta)}{\sigma^2_{\min} \epsilon^2 } \right).\] This condition is satisfied since \[\sum_{i=\variancesamplesize+1}^{9n/10} \min\left\{T^2_{\variancesamplesize}, \frac{1}{\sigma_i^2}\right\}\ge (n/2-\variancesamplesize-1)\frac{1}{\sigma_{\text{med}}^2} \ge \Omega\left(\frac{\log(n/\beta)}{\sigma^2_{\min} \epsilon^2 } \right)\] where the first inequality is simply because more than $(n/2-\variancesamplesize-1)$ of the user have weight larger than the median weight, and the second inequality follows from the assumption that $\frac{k_{\max}}{k_{\text{med}}}\le \frac{\epsilon^2(n/2-\log n-1)}{\log(n/\beta)}$. Therefore, with high probability (based on the accuracy of $\widehat{k_{\variancesamplesize}}$), truncating at $1/\sigma_{\min}^2$ rather than the optimal truncation $T$ does not affect the variance of the estimator by more than a constant factor. 

Now that we have established that the noise added for privacy is not too large, the only remaining potential point of failure for the algorithm is that the PTR component fails and the algorithm outputs $\pinitial$ rather than the more accurate weighted estimate. The fact that this does not happen with high probability is a direct corollary of Lemma~\ref{PTRfailurerate}.
\end{proof}

\section{Proofs from Section~\ref{s.optpideal}}\label{appendix:opt}

\optimalthresh*

\showproofs{\begin{proof}[Proof of Lemma~\ref{optimalityofthresholding}]
Let \[w^* = \arg\min_{\substack{w\in[0,1]^n\\ \sum_{i=1}^n w_i=1}} \var(\widehatp) =  \arg\min_{\substack{w\in[0,1]^n\\ \sum_{i=1}^n w_i=1}} \sum_{i=1}^n w_i^2\sigma_i^2+\frac{\max_k w_k^2\sigma_k^2}{\epsilon^2} \] be an optimal weight vector that minimizes variance of $\widehatp$. We start with a few observations on structural properties of the optimal weights.  Let $M = \{\arg\max_k w^*_k\sigma_k \}$ be the set of all users with maximum weighted-variance contribution to the estimate $\widehatp$. 

First, notice that for all $i,j\in[n]$, if $w_i^*>w_j^*$ then $\sigma_i^2\le\sigma_j^2$. This follows since if $\sigma_i^2>\sigma_j^2$ then $w_i^*\sigma_j^2+w_j^*\sigma_i^2<w_i^*\sigma_i^2+w_j^*\sigma_j^2$ and $\max\{w_i^*\sigma_j^2, w_j^*\sigma_i^2\}\le w_i^*\sigma_i$ which implies that swapping the weights of $i$ and $j$ would result in an estimator with lower variance. This is a contradiction given the definition of $w^*$.

Next, we show that if $i,j\notin M$ then $w_i^*\sigma_i^2=w_j^*\sigma_j^2$. Suppose towards a contradiction that $w_i^*\sigma_i^2<w_j^*\sigma_j^2$. Let $\alpha = \min\{\frac{w_j^*\sigma_j^2-w_i^*\sigma_i^2}{\sigma_i^2+\sigma_j^2}, \frac{\max_k w_k^*\sigma_k-w_i^*\sigma_i}{\sigma_i}, w_j^*\}$. Then $\alpha>0$, and $(w_j^*-\alpha)\sigma_j, (w_i^*+\alpha)\sigma_i\in[0,\max_k w_k^*\sigma_k]$. Also, \begin{align*}
(w_j^*-\alpha)^2\sigma_j^2+(w_i^*+\alpha)^2\sigma_i^2 &= {w_j^*}^2\sigma_j^2+{w_i^*}^2\sigma_i^2+\alpha^2(\sigma_i^2+\sigma_j^2)-2\alpha(w_j^*\sigma_j^2-w_i^*\sigma_i^2)\\
&= {w_j^*}^2\sigma_j^2+{w_i^*}^2\sigma_i^2+\alpha(\alpha(\sigma_i^2+\sigma_j^2)-2(w_j^*\sigma_j^2-w_i^*\sigma_i^2))\\
&< {w_j^*}^2\sigma_j^2+{w_i^*}^2\sigma_i^2.
\end{align*}
This implies that shifting $\alpha$ weight from $w^*_i$ to $w^*_j$ would reduce the variance of the estimator $\widehatp$ without changing the maximum weighted-variance,
which is a contradiction of the optimality of $w^*$. 

Define $H=\max_k w_k^*\sigma_k$ and note that there exists $R>0$ such that $w_i^* = R/\sigma_i^2$ for all $i \notin M$.  From these observations, there must exist some threshold $T$ such that if $\sigma_i \geq 1/T$, then $w_i^* = R/\sigma_i^2$, and if $\sigma_i<1/T$, then $w_i^*=H/\sigma_i$.  By continuity, $H=RT$, and we can write the optimal weights as: $w_i^* = \min\{1/\sigma^2_i, T/\sigma_i\}R$.  Since the weights $w^*_i$ must sum to 1, we know that $R=\frac{1}{\sum_{j=1}^n \min\{1/\sigma_j^2, T/\sigma_j\}}$.  

Thus the optimal weights are:
\[w_i^* = \frac{\min\{1/\sigma_i^2, T/\sigma_i\}}{\sum_{j=1}^n \min\{1/\sigma_j^2, T/\sigma_j\}}, \]
for some appropriate threshold $T$. 

\end{proof}}

Let us recall some notation. Let $\mathcal{P}$ be a parameterized family of distributions $p\mapsto \Dee_{p}$, so $\mathbb{E}[\Dee_{p}]$.
Given an estimator $M$, vector $\boldsymbol{q}\in[0,1]^n$ and set $I\subset[n]$, let \[\mu_M(x_{[n]\backslash I}; \boldsymbol{q}) = \mathbb{E}_{\forall i\in I, x_i\sim \Dee_{q_i}(k_i), M}[M(x_1,\cdots,x_n)]\] be the expectation taken only over the randomness of $I$ and $M$. Note that in this notation, user $i$ is sampling from a meta-distribution with mean $q_i$, which may be different for each user. We will abuse notation slightly and for $p\in[0,1]$, we will let $\mu_M(x_{[n]\backslash I}; p)=\mu_M(x_{[n]\backslash I}; (p,\cdots,p))$. 
Let $\mu_M(\boldsymbol{q})=\mu_M(\emptyset;\boldsymbol{q})$. When the estimator $M$ is clear from context, we will simply use the notation $\mu(x_{[n]\backslash I}; \boldsymbol{p})$. Recall that for $p\in[0,1]$ and $k\in\mathbb{N}$, $\phi_{p,k}$ is the probability density function of $\Dee_p(k)$. We will prove Lemma~\ref{variancedecomp} first since this lemma is required for the proof of Lemma~\ref{optimalityoflinearmain}. 

\variancedecomp*

\begin{proof}[Proof of Lemma~\ref{variancedecomp}]
Let $M:[0,1]^n\to [0,1]$ be a randomised mechanism and suppose that each $x_i\sim \Dee(p_i, k_i)$ where $p_i\sim \Dee$. Now, our goal is to decompose the variance of $M$ into the variance conditioned on each coordinate, and the variance inherent in the mechanism itself. Let $\mu=\mathbb{E}_{x_1\sim \Dee(k_1), \cdots, x_n\sim \Dee(k_n),M}[M(x_1, \cdots, x_n)]$  be the expectation and for any $I\subset[n]$, let $\mu(x_{[n]\backslash I}) = \mathbb{E}_{\forall i\in I, x_i\sim \Dee(k_i), M}[M(x)]$ be the expectation conditioned only on the randomness in $I$. So,
\begin{align*}
    \var(M) &= \mathbb{E}_{x_1\sim \Dee(k_1), \cdots, x_n\sim \Dee(k_n),M}[(M(x_1, \cdots, x_n)-\mu)^2]\\
    &= \mathbb{E}_{x_1\sim \Dee(k_1)}\mathbb{E}_{x_2\sim \Dee(k_2), \cdots, x_n\sim \Dee(k_n),M}[(M(x_1, \cdots, x_n)-\mu_1(x_1)+\mu_1(x_1)-\mu)^2]\\
    &= \mathbb{E}_{x_1\sim \Dee(k_1)}\mathbb{E}_{x_2\sim \Dee(k_2), \cdots, x_n\sim \Dee(k_n),M}[(M(x_1, \cdots, x_n)-\mu_1(x_1))^2\\
    &\hspace{1in}+2(M(x_1, \cdots, x_n)-\mu_1(x_1))(\mu_1(x_1)-\mu)+(\mu_1(x_1)-\mu)^2]\\
    &= \mathbb{E}_{x_1\sim \Dee(k_1)}[(\mu_1(x_1)-\mu)^2]+\mathbb{E}_{x_1\sim \Dee(k_1)}\mathbb{E}_{x_2\sim \Dee(k_2), \cdots, x_n\sim \Dee(k_n),M}[(M(x_1, \cdots, x_n)-\mu_1(x_1))^2].
\end{align*}
Now, by induction we obtain the following decomposition of the variance of $M$, 
\begin{align*}
    \var(M) &= \sum_{i=1}^n \mathbb{E}_{x_1\sim \Dee(k_1),\cdots, x_i\sim \Dee(k_i)}[(\mu(x_{j\le i})-\mu(x_{j<i}))^2]\\
    &\hspace{1in}+\mathbb{E}_{x_1\sim \Dee(k_1), \cdots, x_n\sim \Dee(k_n),M}[(M(x_1,\cdots,x_n)-\mu(x_1,\cdots,x_n))^2]\\
    &\ge \sum_{i=1}^n \mathbb{E}_{x_i\sim \Dee(k_i)}[(\mu(x_{i})-\mu)^2]+\mathbb{E}_{x_1\sim \Dee(k_1), \cdots, x_n\sim \Dee(k_n),M}[(M(x_1,\cdots,x_n)-\mu(x_1,\cdots,x_n))^2]
\end{align*}
where the second inequality follows from Jensen's inequality:
\begin{align*}
\mathbb{E}_{x_1\sim \Dee(k_1),\cdots, x_i\sim \Dee(k_i)}[(\mu(x_{j\le i})-\mu(x_{j<i}))^2]&\ge \mathbb{E}_{x_i\sim \Dee(k_i)}[(\mathbb{E}_{x_1\sim \Dee(k_1),\cdots, x_{i-1}\sim \Dee(k_i)}[\mu(x_{j\le i})-\mu(x_{j<i})])^2] \\
&= \mathbb{E}_{x_i\sim \Dee(k_i)}[(\mu(x_{i})-\mu)^2].
\end{align*}
\end{proof}

\optimallinearmain*

\begin{proof}[Proof of Lemma~\ref{optimalityoflinearmain}] 

We first apply Lemma \ref{variancedecomp} to decompose the variance of the estimate computed by $M$ as:
\[ \var_{{\forall i\in[n], x_i\sim\Dee_p(k_i), M}}(M) \ge \sum_{i=1}^n \mathbb{E}_{x_i\sim \Dee_p(k_i)}[(\mu(x_{i};p)-\mu(p))^2]+\mathbb{E}_{\forall i\in[n], x_i\sim\Dee_p(k_i),M}[(M(x_1,\cdots,x_n)-\mu(x_1,\cdots,x_n;p))^2]
\]
The first term is the sum of contributions to the variance of the individual terms $x_i$, and the second term is the contribution to the variance of the noise added for privacy. We will proceed by bounding these terms separately, starting with the first term.

First note that by definition, \[\int (\mu(x_i; \boldsymbol{q})-\mu(\boldsymbol{q}))\phi_{q_i, k_i}(x_i) dx_i = \mathbb{E}_{x_i\sim \Dee_{q_i}(k_i)} [\mu(x_i; \boldsymbol{q})]-\mu(\boldsymbol{q})=0.\] Therefore, by taking the partial derivative with respect to $q_i$ we have 
\[\int \left[\left(\frac{\partial}{\partial q_i}(\mu(x_i; \boldsymbol{q})-\mu(\boldsymbol{q}))\right) \phi_{q_i,k_i}(x_i)+(\mu(x_i; \boldsymbol{q})-\mu(\boldsymbol{q}))\frac{\partial}{\partial q_i} \phi_{q_i, k_i}(x_i)\right] dx_i=0.\]
Note that $\mu(x_i; \boldsymbol{q})$ is constant in $q_i$ so
rearranging, and noting that $\frac{\partial}{\partial q_i} \phi_{q_i, k_i}(x_i)=\phi_{q_i,k_i}(x_i)\left(\frac{\partial}{\partial q_i}\log \phi_{q_i,k_i}(x_i)\right)$ we have, 
\begin{align}
\nonumber\int \left(\frac{\partial}{\partial q_i}\mu(\boldsymbol{q})\right)\phi_{q_i,k_i}(x_i)dx_i &= \int (\mu(x_i; \boldsymbol{q})-\mu(\boldsymbol{q}))\phi_{q_i,k_i}(x_i)\left(\frac{\partial}{\partial q_i}\log \phi_{q_i,k_i}(x_i)\right)dx_i \\
&\le \sqrt{\left(\int(\mu(x_i; \boldsymbol{q})-\mu(\boldsymbol{q}))^2\phi_{q_i,k_i}(x_i)dx_i\right)\left(\int \left(\frac{\partial}{\partial q_i}\log \phi_{q_i,k_i}(x_i)\right)^2 \phi_{q_i,k_i}(x_i)dx_i\right)}.\label{eqntorearrange}
\end{align}
Let \[w_i(p) = \int \left(\frac{\partial}{\partial q_i}\mu(\boldsymbol{q})\right)\phi_{q_i,k_i}(x_i)dx_i \;\Bigg\rvert_{\boldsymbol{q}=(p, \cdots, p)}=\frac{\partial}{\partial q_i}\mu(\boldsymbol{q})\;\Bigg\rvert_{\boldsymbol{q}=(p, \cdots, p)}\] and note that by assumption there exists a constant $c$ such that for all $i\in[n]$ and $q_i\in[1/3,2/3]$, \[\int \left(\frac{\partial}{\partial q_i}\log \phi_{q_i,k_i}(x_i)\right)^2 \phi_{q_i,k_i}(x_i)dx_i\le \frac{1}{c\cdot\var(\Dee_{q_i}(k_i)}).\]
Then evaluating both sides of Equation~\eqref{eqntorearrange} at the constant vector $\boldsymbol{q}=(p, \cdots, p)$, we have
\[\left(\int(\mu(x_i; p )-\mu(p))^2\phi_{p,k_i}(x_i)dx_i\right)\ge \frac{w_i(p)^2}{\int \left(\frac{\partial}{\partial p}\log \phi_{p,k_i}(x_i)\right)^2 \phi_{p,k_i}(x_i)dx_i}\ge c\cdot w_i(p)^2\var(\Dee_p(k_i)). \]

Now we have controlled the contribution of each individual coordinate to the variance of $M$, and it remains to control the contribution of the noise due to privacy. 

We will show that for two independent samples $x_i$, $x_i'$  drawn from $\Dee_p(k_i)$, 
\begin{equation}
\E[(\mu(x_i;p)-\mu(x_i';p))^2]\ge \Omega\Big(w_i(p)^2\cdot\var(\Dee_p(k_i))\Big).\end{equation}

Letting \[\alpha=\sqrt{\E[(\mu(x_i;p)-\mu(x_i';p))^2]},\] we can write
\begin{align*}
    w_i(p) &= \frac{\partial\mu(\boldsymbol{q}) }{\partial q_i}\;\Bigg\rvert_{\boldsymbol{q}=(p, \cdots, p)} \\
    &= \frac{\partial(\mu(\boldsymbol{q})- \mu(x_i'; \boldsymbol{q})) }{\partial q_i} \;\Bigg\rvert_{\boldsymbol{q}=(p, \cdots, p)}\\
    &= \frac{\partial}{\partial q_i} \int_{x_i} (\mu(x_i; \boldsymbol{q}) - \mu(x_i'; \boldsymbol{q})) \phi_{q_i,k_i}(x_i) dx_i \;\Bigg\rvert_{\boldsymbol{q}=(p, \cdots, p)}\\
    &=  \int_{x_i} (\mu(x_i; p) - \mu(x_i'; p)) \left(\frac{\partial \phi_{q_i,k_i}(x_i)}{\partial q_i}\;\Bigg\rvert_{\boldsymbol{q}=(p, \cdots, p)}\right) dx_i\\
    &=  \int_{x_i} (\mu(x_i; p) - \mu(x_i'; p)) \left(\frac{\partial \log \phi_{q_i,k_i}(x_i)}{\partial q_i}\;\Bigg\rvert_{\boldsymbol{q}=(p, \cdots, p)}\right)\phi_{p,k_i}(x_i) dx_i\\
    &\leq  \sqrt{\left(\int_{x_i} (\mu(x_i; p) - \mu(x_i'; p))^2 \phi_{p,k_i}(x_i) dx_i \right) \left(\int_{x_i}\left(\frac{\partial \log \phi_{q_i,k_i}(x)}{\partial q_i}\;\Bigg\rvert_{\boldsymbol{q}=(p, \cdots, p)}\right)^2 \phi_{p,k_i}(x) dx_i\right)}\\
    &\leq \alpha \cdot \sqrt{\int_{x_i}\left(\frac{\partial \log \phi_{p_i,k_i}(x_i)}{\partial p_i}\;\Bigg\rvert_{\boldsymbol{p}=(p, \cdots, p)}\right)^2 \phi_{p_i,k_i}(x_i) dx_i}\\
    &\le \alpha\cdot \sqrt{\frac{1}{c\cdot \var(\Dee_p(k_i))}}
\end{align*}
The first equality is by definition. The second equality follows from the fact that $\mu(x_i';\mathbf{q})$ is constant with respect to $q_i$, so its derivative is 0. The third inequality simply expands out the definition of $\mu(\mathbf{q})$. The fourth equality follows from the linearity of derivatives, the fact that $\mu(x_i;\mathbf{q})-\mu(x_i',\mathbf{q})$ is constant with respect to $q_i$, and the fact that $(\mu(x_i;\mathbf{q})-\mu(x_i',\mathbf{q}))|_{\mathbf{q}=(p,\cdots,p)} = (\mu(x_i;p)-\mu(x_i',p))$. The fifth equality follows from the formula $\frac{\partial}{\partial x}\ln f(x) = \frac{\frac{\partial}{\partial x} f(x)}{f(x)}$, which holds for any differentiable function $f$. The first inequality is a result of the Cauchy-Schwarz inequality. The second inequality follows from the definition of $\alpha$, and the final inequality follows from Assumption 2 of Lemma \ref{optimalityoflinearmain}. 

Therefore, \[\alpha\ge  w_i(p)\cdot\sqrt{c\cdot\var(\Dee_p(k_i))}. \]

We now argue that any $(\eps, \eps^2/100)$-differentially private mechanism should have variance $\Omega(\alpha^2\log\frac{1}{\eps}/10\eps^2)$. Suppose that we had a mechanism that violated this property. Then by running this mechanism $\frac{1}{\eps^2\log\frac{1}{\eps}}$ times and averaging, the advanced composition theorem implies that this average is $(1, 1/100)$-DP. This averaged output however has variance $O(\alpha^2/10)$. Thus given samples $x_i$, and $x'_i$ such that $|\mu(x_i;p)-\mu(x_i';p)| \geq \alpha/2$, if the noise had variance $O(\alpha^2/10)$ on $(x_i, x_{-i})$ as well as on $(x'_i, x_{-i})$ (when $x_{-i}$ is drawn randomly), then these two inputs would be distinguishable with probability at least $9/10$. This however violates the $(1, 1/100)$-DP of the averaged algorithm. This implies that for random $x_i$, the noise added by the DP algorithm is at least $\Omega(\alpha^2\log\frac{1}{\eps}/20\eps^2)$

Thus the variance of $M$ is, \begin{equation}\label{eq.varofM} \var_{{\forall i\in[n], x_i\sim\Dee_p(k_i), M}}(M) \ge \sum_{i=1}^n c\cdot w_i(p)^2\var(\Dee_p(k_i)) + \Omega\Big(\frac{w_i(p)^2\var(\Dee_p(k_i))}{\eps^2}
\Big) \end{equation}

Finally, since the weights $w_i(p)$ that we defined need not sum to 1, they will need to be normalized to sum to 1 to satisfy the conditions of $\linearest$.  We need to show this normalisation does not substantially increase the variance of the estimator in $\linearest$ defined by these weights. This is equivalent to showing that the normalisation term, $\sum_{i=1}^n w_i(p)$ is large for some $p$.
For 
$p\in[1/3,2/3]$, let $\gamma:[1/3,2/3]\to[0,1]^n$, defined by $\gamma(p) = (p, \cdots, p)$, be a path in $[0,1]^n$ then by the fundamental theorem of line integrals, 
\begin{align*}
3\int_{1/3}^{2/3} \sum_{i=1}^n w_i(p) dp &= 3\int_{1/3}^{2/3} \left(\sum_{i=1}^n  \left(\frac{\partial}{\partial q_i}\mu(\boldsymbol{q})\right) \;\Bigg\rvert_{\boldsymbol{q}=(p, \cdots, p)}\right) dp\\ 
&=3\int_{\gamma}  \nabla \mu(\boldsymbol{q}) \cdot \boldsymbol{1} d \boldsymbol{q}\\
&= 3( \mu(2/3, \cdots, 2/3) - \mu(1/3,\cdots, 1/3))\\
&= 1
\end{align*}

This implies that there exists $p^*\in[1/3,2/3]$ such that $\sum_{i=1}^n w_i(p^*)\ge 1$.
Define
\begin{align*} M_{\texttt{NL}}(x_1,\cdots,x_n) &= \sum_{i=1}^n \frac{w_i(p^*)}{\sum_{j=1}^n w_j(p^*)x_j}+\Lap\left(\frac{\max_i \frac{w_i(p^*)}{\sum_{j=1}^n w_j(p^*)} \sqrt{\var(\Dee_p(k_i))}}{\epsilon}\right) \\
&= \frac{1}{\sum_{i=1}^n w_i(p^*)} \left(\sum_{i=1}^n w_i(p^*)x_i+\Lap\left(\frac{\max_i w_i(p^*) \sqrt{\var(\Dee_p(k_i))}}{\epsilon}\right)\right),\end{align*}
where the second equality follows from properties of the Laplace distribution.
Now, 
\begin{align*}
\var_{\Dee_p}(M_{\texttt{NL}}) &\le \frac{1}{(\sum_{i=1}^n w_i(p^*))^2} \left(\sum_{i=1}^n w_i(p^*)^2\var(\Dee_p(k_i))+O\left(\frac{\max_i w_i(p^*)^2\var(\Dee_p(k_i))}{\epsilon^2}\right)\right)\\
&\le \sum_{i=1}^n w_i(p^*)^2\var(\Dee_p(k_i))+ O\left(\frac{ \max_i w_i(p^*)^2\var(\Dee_p(k_i))}{\epsilon^2}\right),
\end{align*}
where the second inequality comes from the fact that $\sum_{i=1}^n w_i(p^*)\ge 1$. Comparing this with Equation \ref{eq.varofM}, we see that
specifically, at $p=p^*$, \[\var_{\Dee_{p^*}}(M_{\texttt{NL}})\le O\left(\var_{\Dee_{p^*}}(M)\right).\]
Now, if $p,p^*\in[1/3,2/3]$ then $\var(\Dee_{p}(k_i))=\Theta\left(\var(\Dee_{p^*}(k_i))\right)$ so $\var_{D_p}(M_{\texttt{NL}})=\Theta(\var_{D_{p^*}}(M_{\texttt{NL}}))$. Therefore, the worst case variance of $M_{\texttt{NL}}$ is less than the worst case variance of $M$ over all $p\in[1/3,2/3]$, as required.
\end{proof}

\optthreshold*

\begin{proof}[Proof of Lemma \ref{optimalitythresholdingdoesntmatter}]
The variance claim follows immediately from noting that $\var\left([x_i]_{p-\concentrationbound{k_i}{\Dee}{n}{\sigma_p^2}{\beta}}^{p+\concentrationbound{k_i}{\Dee}{n}{\sigma_p^2}{\beta}}\right)\le \var(x_i)$, and the assumption that $\concentrationbound{k_i}{\Dee}{n}{\sigma_p^2}{\beta}=\tilde{O}(\var(\Dee(k_i))$. The bias claim follows from noting that with probability $1-\beta$, $[x_i]_{p-\concentrationbound{k_i}{\Dee}{n}{\sigma_p^2}{\beta}}^{p+\concentrationbound{k_i}{\Dee}{n}{\sigma_p^2}{\beta}}=x_i$. This implies that $M_{\texttt{TNL}}$ is within $\beta$ in total variation distance to an unbiased estimator. Since $M_{\texttt{TNL}}$ takes values in $[0,1]$, this implies the mean is in $[p-\beta, p+\beta]$. 
\end{proof}

\lowerb*

\begin{proof}[Proof of Corollary \ref{cor.lower}]
Firstly, suppose that $\sigma_p=0$, so the meta-distribution is constant, and $\Dee_p(k_i)=\Bin(k_i,p)$. Then the Fisher information of $\phi_{p,k_i}$ is $\int \left(\frac{\partial}{\partial p}\log \phi_{p,k_i}(x_i)\right)^2 \phi_{p,k_i}(x_i)dx_i=\frac{k_i}{p(1-p)}$ and $\var(\Dee_p(k_i))=\frac{p(1-p)}{k_i}$, so $\Dee_p(k_i)$ satisfies Condition~\ref{condition2} of Lemma~\ref{optimalityoflinearmain}. Additionally, \begin{align*}
\min_{M\text{, unbiased}} \max_{p\in[1/3,2/3]}\var_{\Dee_p}[M] &= \tilde{\Omega}\left(\max_{p\in[1/3,2/3]} \var_{\Dee_p} [\pideal]\right) \quad \text{(under conditions of Thm \ref{metatheorem})}
\end{align*} We can view the truncation as simply choosing a maximum $k^*$ so that $T=\sqrt{\frac{k^*}{p(1-p)}}$.
Now, the un-normalised weights of $\pideal$ are \[\min\left\{\frac{1}{\var(\Dee_p(k_i))}, \frac{T}{\sqrt{\var(\Dee_p(k_i))}}\right\}=\min\left\{\frac{k_i}{p(1-p)}, \frac{\sqrt{k_i k^*}}{p(1-p)}\right\}.\] Further, $\var([\widehatpi]_{a_i}^{b_i})\le\var(\Dee(k_i))$ and we assume throughout this paper that $\var([\widehatpi]_{a_i}^{b_i})\ge (1/2)\var(\Dee(k_i))$. So, $\var([\widehatpi]_{a_i}^{b_i})=\Theta(\var(\Dee(k_i))) = \Theta(\frac{p(1-p)}{k_i})$.  
Finally, since binomials are highly concentrated, $|b_i-a_i|=\Omega(\sigma_i)$, which implies that $\tfrac{\max_i w_i^*|b_i-a_i|}{\epsilon}$ as defined in Equation~\eqref{privoptimal1} is achieved at $k_i=k^*$. Thus,
\begin{align*}
\min_{M\text{, unbiased}} \max_{p\in[1/3,2/3]}\var_{\Dee_p}[M]
&= \max_{p\in[1/3,2/3]} \frac{\Omega\left(\frac{k^*}{p(1-p)\epsilon^2}\right)+ \sum_{i=1}^n \left(\min\left\{\frac{k_i}{p(1-p)}, \frac{\sqrt{k_i k^*}}{p(1-p)}\right\}\right)^2\frac{1}{2}\frac{p(1-p)}{k_i}}{\left(\sum_{i=1}^n \min\left\{\frac{k_i}{p(1-p)}, \frac{\sqrt{k_i k^*}}{p(1-p)}\right\}\right)^2 }\\
&= \tilde{\Omega}\left(\max_{p\in[1/3,2/3]}p(1-p) \frac{\frac{k^*}{\epsilon^2}+\sum_{i=1}^n \min\{k_i, k^*\}}{(\sum_{i=1}^n \min\{k_i,\sqrt{k_ik^*}\})^2}\right)\\
&= \tilde{\Omega}\left( \frac{\frac{k^*}{\epsilon^2}+\sum_{i=1}^n \min\{k_i, k^*\}}{(\sum_{i=1}^n \min\{k_i,\sqrt{k_ik^*}\})^2}\right),
\end{align*}
where the first equality comes from Theorem~\ref{optimalityfisherinfo}, the second equality pulls out common factors, and the third equality is because $p$ is bounded away from 0 and 1.

For the other component of the bound we will let $\Dee_p$ be a truncated Gaussian distribution. Let $\phi$ and $\Phi$ respectively be the probability density function and cumulative density function of the standard Gaussian $\mathcal{N}(0,1)$. Let $W$ be such that $\gamma:=\Phi(W)-\Phi(-W)\ge 9/10$ and $\truncatedfactor := \frac{2W\phi(W)}{\Phi(W)-\Phi(-W)}\le 1/2$.
Define the truncated Gaussian $\Dee_p$ with mean $p$ on $[p-\frac{\sigma_p}{\sqrt{1-\truncatedfactor}} W, p+\frac{\sigma_p}{\sqrt{1-\truncatedfactor}} W]$ by the probability density function: \[\phi_p(q) = \begin{cases}
\frac{1}{\gamma}\phi\left((q-p)\frac{\sqrt{1-\truncatedfactor}}{\sigma_p}\right) & q\in[p-\frac{\sigma_p}{\sqrt{1-\truncatedfactor}} W, p+\frac{\sigma_p}{\sqrt{1-\truncatedfactor}} W]\\
0 & \text{otherwise.}
\end{cases}.\]
Now, the variance of $\Dee_p$ is $\sigma_p^2$ and the Fisher information of $\Dee_p$ is given by \cite{Mihoc:2003} \begin{equation}\label{cramerraotrungauss}\frac{1}{\sigma_p^2}\left(1-\truncatedfactor\right)^2\in\left[\frac{1}{4\sigma_p^2}, \frac{1}{\sigma_p^2}\right].\end{equation}
Since any sample from $\Dee$ can be post-processed into a sampling from $\Dee(k)$ for any $k \in \mathbb{N}$, we have
\begin{align*}
\min_{M\text{, unbiased}} \max_{p\in[1/3,2/3]}\var_{\forall i\in[n], x_i\sim\Dee_p(k_i)}[M(x_1, \cdots, x_n)]&\ge \min_{M\text{, unbiased}} \max_{p\in[1/3,2/3]}\var_{p_1, \cdots, p_n\sim\Dee_p}[M(p_1, \cdots, p_n)]\\
&\ge \max_{p\in[1/3,2/3]}O\left(\frac{\sigma_p^2}{n}\right)\\
&= O(\frac{\sigma_p^2}{n}),
\end{align*}
where the second inequality follows from the Cram\'er-Rao bound \citep{Nielsen:2013} and Equation~\eqref{cramerraotrungauss}.
\end{proof}

\section{Proofs from Section~\ref{instantiation}}\label{appendix.private}

\initmean*

\showproofs{\begin{proof}[Proof of Lemma~\ref{initialmeanestimate}] 
Firstly, the privacy guarantees follows immediately from the Laplace Mechanism in differential privacy \citep{Dwork:2006} noting that $\frac{10}{n}\sum_{i=(9n/10)+1}^n x^1_i$ has sensitivity $\frac{10}{n}$.

Now, let us turn to the two accuracy guarantees. We will start with the guarantee that $\pinitial$ is close to $p$ with high-probability.
Note that $\berP$ is simply a Bernoulli random variable with mean $p$ so since each sample is independent, $\frac{10}{n}\sum_{i=(9n/10)+1}^n x^1_i=\Bin(n/10, p)$. Thus, if $n\ge \frac{20\log(1/\beta)}{p}$, a Chernoff bound gives \[\Pr\left[\left|\frac{10}{n}\sum_{i=(9n/10)+1}^n x^1_i-p\right|\ge \sqrt{\frac{3\min\{p, 1-p\}\log(4/\beta)}{n/10}}\right]\le\beta/2.\] Therefore, combining with a high probability bound on the Laplace distribution, \[\Pr\left[\left|\pinitial-p\right|\ge \sqrt{\frac{3\min\{p,1-p\}\log(4/\beta)}{n/10}}+\frac{\log(2/\beta)}{\epsilon n/10}\right]\le\beta.\]
We will condition on the following event for the remainder of the proof, which will occur with probability $1-\beta$: \[\left|\pinitial-p\right|\le 2\max\left\{\sqrt{\frac{3\min\{p,1-p\}\log(4/\beta)}{n/10}},\frac{\log(2/\beta)}{\epsilon n/10}\right\}.\]
Now if $\left|\pinitial-p\right|\le 2\sqrt{\frac{3\min\{p,1-p\}\log(4/\beta)}{n/10}}$. Since we need $\alpha$ in terms of $\pinitial$ rather than $p$ (since $\pinitial$ is known to the algorithm), we need to rework this formula. Squaring both sides and bringing all the terms to the same side, we obtain 
\[p^2-2\left(\pinitial+\frac{6\log(4/\beta)}{n/10}\right)p+(\pinitial)^2\le 0.\] Completing the square we obtain
\[ \left(p-\pinitial-\frac{6\log(4/\beta)}{n/10}\right)^2+(\pinitial)^2-\left(\pinitial+\frac{6\log(4/\beta)}{n/10}\right)^2\le 0. \]
Now, rearranging and taking the square root, we obtain
\[\left|p-\pinitial-\frac{6\log(4/\beta)}{n/10}\right|\le \sqrt{\left(\pinitial+\frac{6\log(4/\beta)}{n/10}\right)^2-(\pinitial)^2} \]
then by squaring both sides, using the fact that
$\min\{p,1-p\} \leq p$,
and rearranging we have 
\begin{align*}
|\pinitial-p|\le \sqrt{\frac{12\pinitial\log(4/\beta)}{n/10}+\frac{36\log^2(4/\beta)}{n^2/100}}+\frac{6\log(4/\beta)}{n/10}
\end{align*}
which implies that, \[\left|\pinitial-p\right|\le 2\max\left\{\sqrt{\frac{12\pinitial\log(4/\beta)}{n/10}+\frac{36\log^2(4/\beta)}{n^2/100}}+\frac{6\log(4/\beta)}{n/10},\frac{\log(2/\beta)}{\epsilon n/10}\right\}.\]

We need to show that this expression is less than or equal to $\concentrationbound{k_i}{\Dee}{n}{\sigma_p^2}{\beta}$ because $\alpha = O(1/\sqrt{n})$. To see this, note that $\alpha = O(1/\sqrt{n})$ and $\concentrationbound{k_i}{\Dee}{n}{\sigma_p^2}{\beta}$ is increasing towards 1 as $n$ grows large. Thus for $n$ sufficiently large, $\alpha \leq \concentrationbound{k_i}{\Dee}{n}{\sigma_p^2}{\beta}$ will be satisfied.

Next we turn to proving the second accuracy claim, that $\pinitial(1-\pinitial)$ is concentrated around $p(1-p)$. 
Let $\mathcal{E} = \pinitial-p$ so
\begin{align*}
    \pinitial(1-\pinitial) &= (p+\mathcal{E})(1-p-\mathcal{E})
    = p(1-p)+(1-2p)\mathcal{E}-\mathcal{E}^2
\end{align*}
Now, if $\min\{p,1-p\}\ge K \max\left\{\frac{3\log(4/\beta)}{n/10}, \frac{\log(2/\beta)}{\epsilon n/10}\right\}$ for some constant $K$, then
\begin{align*}
    |\mathcal{E}|&\le \sqrt{\frac{3\min\{p,1-p\}\log(4/\beta)}{n/10}}+\frac{\log(2/\beta)}{\epsilon n/10}\\
    &\le \sqrt{\frac{\min\{p,1-p\}\min\{p, (1-p)\}}{K}}+\frac{\min\{p,(1-p)\}}{K}\\
    &\le \frac{2\min\{p,1-p\}}{K}.
\end{align*}
Thus, combining this with the fact that $1-2p \leq \max\{p,1-p\}$ for $p\in[0,1]$,
\begin{align*}
    |(1-2p)\mathcal{E}-\mathcal{E}^2|&\le \max\{p,1-p\}\frac{2\min\{p,1-p\}}{K}+\left(\frac{2\min\{p,1-p\}}{K}\right)^2\\
    &\le \frac{6p(1-p)}{K}
\end{align*}
Finally, choosing $K=12$ gives, \[\pinitial(1-\pinitial)\in \left[\frac{p(1-p)}{2}, \frac{3p(1-p)}{2}\right].\]

\end{proof}}

\lemthirdmoment*

\showproofs{\begin{proof}[Proof of Lemma~\ref{lem.thirdmoment}] Note that $\mathbb{E}[\bersumP]=p$. Then we can bound the absolute third central moment as follows,
\begin{align*}
\mathbb{E}_{x\sim\bersumP}[|x-p|^3] &= \mathbb{E}_{p_i\sim \Dee}\mathbb{E}_{y\sim \Bin(k, p_i)} [|(\frac{1}{k}y-p_i)-(p-p_i)|^3]\\
&\le 4\left(\mathbb{E}_{p_i\sim \Dee}\mathbb{E}_{y\sim \Bin(k, p_i)}[|\frac{1}{k}y-p_i|^3]+\mathbb{E}_{p_i\sim \Dee}[|p-p_i|^3]\right)\\
&\le 4\left(\frac{1}{k^3}\mathbb{E}_{p_i\sim \Dee}\left[\sqrt{\mathbb{E}_{y\sim \Bin(k, p_i)}[|y-k\cdot p_i|^2]\mathbb{E}_{y\sim \Bin(k, p_i)}[|y-p_i|^4]}\right]
+\gamma\sigma_p^3\right)\\
&\hspace{2.5in} \text{(by Cauchy-Schwarz inequality)}\\
&\le 4\left(\frac{1}{k^3}\mathbb{E}_{p_i\sim \Dee}\left[\sqrt{k^2(p_i(1-p_i))^2(1+3kp_i(1-p_i))}\right]
+\gamma\sigma_p^3\right)\\
&\le 4\left(\frac{1}{k^3}\mathbb{E}_{p_i\sim \Dee}[k(p_i(1-p_i))]+\frac{1}{k^3}\mathbb{E}_{p_i\sim \Dee}[\sqrt{3k^3(p_i(1-p_i))^3}]
+\gamma\sigma_p^3\right)\\
&\le 4\left(\frac{1}{k^2}p(1-p)+\frac{\sqrt{3}}{k^{3/2}}\mathbb{E}_{p_i\sim \Dee}[\sqrt{(p_i(1-p_i))^3}]
+\gamma\sigma_p^3\right)\\
&\hspace{2.5in} \text{ (by Jensen's inequality)}\\
&\le 4\left(\frac{1}{k^{3/2}}\sqrt{(p(1-p))^3}+\frac{\sqrt{3}}{k^{3/2}}\mathbb{E}_{p_i\sim \Dee}[\sqrt{(p_i(1-p_i))^3}]
+\gamma\sigma_p^3\right),\\
\end{align*}
where the first inequality follows from the following inequality that holds for all real valued $a$ and $b$: $|a-b|^3\le 4(|a|^3+|b|^3)$.
The second to last inequality follows from Jensen's inequality since $h(x)=x(1-x)$ is concave, and the last inequality follows since $\frac{1}{\sqrt{k}}\le\sqrt{p(1-p)}$. Now, we will use a generalised form of Jensen's inequality to bound $\mathbb{E}_{p_i\sim \Dee}[\sqrt{(p_i(1-p_i))^3}]$. Let $h(x)=(x(1-x))^{3/2}$ and \[\phi(x) = \frac{h(x)-h(p)}{(x-p)^2}-\frac{h'(p)}{x-p}.\] Since $p\in[\frac{1}{k}, 1-\frac{1}{k}]$, \[\max_{x\in[\frac{1}{2k}, 1-\frac{1}{2k}]}\phi(x) \le (1/2)\max_{x\in[\frac{1}{2k}, 1-\frac{1}{2k}]}h''(x) \le h''\left(\frac{1}{2k}\right) =  \frac{3 (8 (\frac{1}{2k})^2 - 8 (\frac{1}{2k}) + 1)}{4 \sqrt{(1 - \frac{1}{2k}) \frac{1}{2k}}} = \frac{3(8-16k+4k^2)}{8k\sqrt{(2k-1)}}\le \frac{3}{2}\sqrt{k}.\] If $x\notin[\frac{1}{2k}, 1-\frac{1}{2k}]$ then $|x-p|\ge \frac{1}{2k}$ and $h(x)<h(p)$, so \[\phi(x)\le \frac{|h'(p)|}{|x-p|} = \frac{3 |1 - 2 p| \sqrt{p(1-p)}}{2|p-x|} \le \frac{3}{2}\frac{\sqrt{p(1-p)}}{|p-x|}\le \max\left\{\frac{3}{2}\frac{\sqrt{\frac{1}{k}(1-\frac{1}{k})}}{|\frac{1}{k}-x|}, \frac{3}{2}\frac{\sqrt{\frac{1}{k}(1-\frac{1}{k})}}{|1-\frac{1}{k}-x|}\right\} \le 3\sqrt{k-1}\le 3\sqrt{k}. \]
Therefore, by the generalised Jensen's inequality, \[\mathbb{E}_{p_i\sim \Dee}[\sqrt{(p_i(1-p_i))^3}]\le \sqrt{(p(1-p))^3}+\sigma_p^2\cdot 3\sqrt{k}\le \sqrt{(p(1-p))^3}+\sigma_p^2\cdot 3\sqrt{k}.\] 
Continuing to bound the absolute central third moment as above, 
\begin{align*}
\mathbb{E}_{x\sim\bersumP}[|x-p|^3]
&\le 4\left(\frac{1}{k^{3/2}}\sqrt{(p(1-p))^3}+\frac{\sqrt{3}}{k^{3/2}}\mathbb{E}_{p_i\sim \Dee}[\sqrt{(p_i(1-p_i))^3}]
+\gamma\sigma_p^3\right)\\
&\le 4\left(\frac{1}{k^{3/2}}\sqrt{(p(1-p))^3}+\frac{\sqrt{3}}{k^{3/2}}\sqrt{(p(1-p))^3}+3\sqrt{3}\frac{\sigma_p^2}{k}
+\gamma\sigma_p^3\right)\\
&\le 4\left(\frac{1}{k^{3/2}}\sqrt{(p(1-p))^3}+\frac{\sqrt{3}}{k^{3/2}}\sqrt{(p(1-p))^3}+3\sqrt{3}\sigma_p^3
+\gamma\sigma_p^3\right)\\
&\le 4(3\sqrt{3}+\gamma)\left(\frac{1}{k^{3/2}}\sqrt{(p(1-p))^3}
+\sigma_p^3\right)\\
&\le 4(3\sqrt{3}+\gamma)\left(\frac{1}{k}p(1-p)
+\sigma_p^2\right)^{3/2}\\
&\le 8(3\sqrt{3}+\gamma)\left(\frac{1}{k}p(1-p)
+\frac{k-1}{k}\sigma_p^2\right)^{3/2},
\end{align*}
where the first and second inequalities follow from above, 
the third inequality follows because $k\geq 1$, the fourth is simply rearranging the terms, the fifth follows from the fact that for all positive, real numbers $a$ and $b$: $a^{3/2}+b^{3/2}<(a+b)^{3/2}$, and the last inequality follows since if $k\ge 2$ then $(k-1)/k>1/2$.
\end{proof}}

\lemmultvariance*

\showproofs{\begin{proof}[Proof of Lemma~\ref{multiplicativevariance}] Note that the conditions are sufficient to ensure from Lemma~\ref{lem.thirdmoment} that $\frac{\rho_{\bersumP}}{\var(\bersumP)^{3/2}}\le 8(3\sqrt{3}+\gamma)$. Then Lemma~\ref{highprobstd} and Lemma~\ref{lem.sigi} imply that \[\var(\Dee(k)) = \frac{1}{k}p(1-p)+\frac{k-1}{k}\sigma_p^2\le\widehat{\sigma}_{p,k}^2\le 8\left(\frac{1}{k}p(1-p)+\frac{k-1}{k}\sigma_p^2\right) = 8\var(\Dee(k)).\] 
\end{proof}}

\subsection{Proof of Lemma~\ref{highprobstd}}\label{appendix:highprobstd}

In this section we slightly generalise the algorithm and analysis given by \cite{Karwa:2018} beyond Gaussian distributions. We will show that their algorithm provides accurate estimates of the mean of sufficiently nice exponential families. This algorithm first estimates the variance of the distribution, then estimates the mean. Both steps of the estimation are performed using differentially private histogram queries.

Let $\rho = \mathbb{E}_{P}[|X-\mathbb{E}_P(x)|^3]$ be the absolute third central moment of $P$, and let $\sigma$ be the standard deviation. Since the algorithm of \cite{Karwa:2018} is designed for Gaussian distributions we will use the following lemma that describes the rate of convergence of the central limit theorem.

\begin{lemma}[Berry-Esseen theorem]\label{BE}
Let $X_1, \cdots, X_n$ be iid samples from a distribution $P$ and $\rho = \mathbb{E}_{P}[|X-\mathbb{E}_P(x)|^3]$. Set $S_n =\frac{1}{n} \sum_{j=1}^n X_j$, $\mu=\mathbb{E}_P[x]$ and $\sigma^2=\var(P)$, and let $Y\sim \mathcal{N}(\mu, \frac{\sigma^2}{n})$ then for some absolute constant $\berryesseen>0$,
\begin{itemize}
\item (Uniform) \[|\mathbb{P}[S_n\le a]-\mathbb{P}[Y\le a]|\le \frac{\berryesseen\rho}{\sigma^3\sqrt{n}}\]
\item (Non-uniform) For all $a>0$, \[|\mathbb{P}[S_n\le a]-\mathbb{P}[Y\le a]|\le \frac{\berryesseen\rho}{(1+|a|)^3\sigma^3\sqrt{n}}.\]
\end{itemize}
\end{lemma}

\begin{lemma}[Histogram Learner \citep{Dwork:2006, Bun:2016, Vadhan:2016}]\label{histlearn} For all $K\in\mathbb{N}$ and any domain $\Omega$, for any collection of disjoint bins $B_1, \cdots, B_K$ defined on $\Omega, n\in\mathbb{N}, \epsilon\ge0, \delta\in(0,1/n), \lambda>0$ and $\beta\in(0,1)$ there exists an $(\epsilon,\delta)$-DP algorithm $M:\Omega^n\to\mathbb{R}^K$ such that for every distribution $D$ on $\Omega$, if 
\begin{enumerate}
\item $X_1, \cdots, X_N\sim D$ and $p_k=\mathbb{P}(X_i\in B_k)$
\item $(\tilde{p_1}, \cdots, \tilde{p_K})=M(X_1, \cdots, X_n)$ and 
\item \[n\ge \max\left\{\min\left\{\frac{8}{\epsilon\lambda}\ln\left(\frac{2K}{\beta}\right), \frac{8}{\epsilon\lambda}\ln\left(\frac{4}{\beta\delta}\right)\right\}, \frac{1}{2\lambda^2}\ln\left(\frac{4}{\beta}\right)\right\}\]
\end{enumerate}
then, \[\mathbb{P}_{X\sim D, M}(\max_k|\tilde{p_k}-p_k|\le\lambda)\ge 1-\beta\;\;\;\text{ and },\]
\[\mathbb{P}_{X\sim D, M}(\arg\max_k\tilde{p_k}=j)\le \begin{cases}
np_j+2e^{-(\epsilon n/8)\cdot(\max_kp_k)} & \text{ if } K< 2/\delta\\
np_j & \text{ if } K\ge 2/\delta
\end{cases}\]
where the probability is taken over the randomness of $M$ and the data $X_1, \cdots, X_n$.
\end{lemma}

\begin{algorithm} \caption{Variance estimator}\label{algo:varianceestimate}
\textbf{Input:} Sample $X = (x_{1},\dots,x_{n})\sim P, \epsilon, \delta, \sigma_{\min}, \sigma_{\max}, \beta, \thirdmoment$.
\begin{algorithmic}[1]
\State Let $\phi= \lceil(600\berryesseen\thirdmoment)^2\rceil$, where $\berryesseen$ is the absolute constant from Lemma~\ref{BE}.
\State If \[n<c\phi\min\left\{\frac{1}{\epsilon}\ln\left(\frac{\ln\left(\frac{\sigma_{max}}{\sigma_{min}}\right)}{\beta}\right), \frac{1}{\epsilon}\ln\left(\frac{1}{\delta\beta}\right)\right\},\] where $c$ is an absolute constant whose existence is ensured by Lemma~\ref{histlearn}, then output $\bot$.
\State Divide $[\sigma_{min}, \sigma_{max}]$ into bins of exponentially increasing length. The bins are of the form $B_j = (2^j, 2^{j+1}]$ for $j=j_{min}, \cdots, j_{max},$ where $j_{max} = \lceil \ln_2\frac{\sigma_{max}}{\sqrt{\phi}}\rceil+1$ and $j_{min}=\lfloor\ln_2\frac{\sigma_{min}}{\sqrt{\phi}}\rfloor-2.$
\State Let $Z_i = \frac{1}{\phi}\sum_{j=1}^{\phi} x_{(i-1)\phi+j}$ for $i=1, \cdots, \lfloor n/\phi\rfloor$.
\State Let $Y_i = Z_{2i}-Z_{2i-1}$ for $i=1, \cdots, \lfloor n/2 \rfloor$
\State Run the histogram learner of Lemma~\ref{histlearn} with privacy parameters $(\epsilon, \delta)$ and bins $B_{j_{min}}, \cdots, B_{j_{max}}$ on input $|Y_1|, \cdots, |Y_n|$ to obtain noisy estimates $\tilde{p_{j_{min}}}, \cdots, \tilde{p_{j_{max}}}$. Let \[\widehat{l}=\arg\max\tilde{p_j}\]
\State Output $\widehat{\sigma} = 2^{\widehat{l}+2}\sqrt{\phi}$.
\end{algorithmic}
\end{algorithm}

Note in particular that the use of approximate $(\epsilon, \delta)$-DP allows us to set the $K=\infty$, while the sample complexity remains finite. 
The following lemma states that provided $\rho/\sigma^3$ is bounded, Algorithm~\ref{algo:varianceestimate} can estimate the standard deviation up to a multiplicative constant.

\begin{lemma}\label{highprobstdduplicate} For all $n\in\mathbb{N}$, $\sigma_{min}<\sigma_{max}\in[0, \infty], \epsilon>0, \delta\in(0,\frac{1}{n}], \beta\in(0,1/2), \thirdmoment>0,$ Algorithm~\ref{algo:varianceestimate} is $(\epsilon, \delta)$-DP and satisfies that if $X_1, \cdots, X_n$ are iid draws from $P$, where $P$ has standard deviation $\sigma\in[\sigma_{min}, \sigma_{max}]$ and  \textcolor{black}{$\frac{\rho}{\sigma^3}\le \thirdmoment$} then if \[n\ge c \thirdmoment^2\min\left\{\frac{1}{\epsilon}\ln\left(\frac{\ln\left(\frac{\sigma_{max}}{\sigma_{min}}\right)}{\beta}\right), \frac{1}{\epsilon}\ln\left(\frac{1}{\delta\beta}\right)\right\},\] (where $c$ is a universal constant), we have \[\mathbb{P}_{X\sim P, M} (\sigma\le\widehat{\sigma}\le 8\sigma)\ge 1-\beta.\]
\end{lemma}

\begin{proof}[Proof of Lemma~\ref{highprobstdduplicate}]
This proof follows almost directly from Theorem 3.2 of \cite{Karwa:2018}. Note that each $Y_i$ is sampled from a distribution with mean 0 and variance $\frac{2\sigma^2}{\phi}$, and in addition is the sum of $\phi$ independent random variables. As in \citep{Karwa:2018}, there exists a bin $B_l$ with label $l\in(\lfloor \ln_2\frac{\sigma_{min}}{\sqrt{\phi}}\rfloor-1, \lceil \ln_2\frac{\sigma_{max}}{\sqrt{\phi}}\rceil)$ such that $\frac{\sigma}{\sqrt{\phi}}\in(2^l, 2^{l+1}]=B_l$. Define, \[p_j = \mathbb{P}(|Y_i|\in B_j).\] Sort the $p_j$s as $p_{(1)}\ge p_{(2)}\ge\cdots$, and let $j_{(1)}, j_{(2)}, \cdots$ be the corresponding bins. Then the following two facts imply the result (as in \citep{Karwa:2018}).

\noindent \textbf{Fact 1:} The bins corresponding to the largest and second largest mass $p_{(1)}, p_{(2)}$ are $(j_{(1)}, j_{(2)})\in\{(l,l-1), (l, l+1), (l+1,l)\}$.

\noindent \textbf{Fact 2:} $p_{(1)}-p_{(3)}>1/300$.

Now, let $W_i\sim N(0, 2\frac{\sigma^2}{\phi})$ and let $q_i, q_{(i)}$ be the corresponding probabilities for $W_i$. Then \cite{Karwa:2018} showed that: 
\begin{itemize}
\item The bins corresponding to the largest and second largest mass $q_{(1)}, q_{(2)}$ are $(j_{(1)}, j_{(2)})\in\{(l,l-1), (l, l+1), (l+1,l)\}$.
\item $q_{(1)}-q_{(3)}>1/100$.
\end{itemize}

By Lemma~\ref{BE}, since $\phi= \lceil(600\berryesseen\thirdmoment)^2\rceil$, for all $j$, $|p_j-q_j|\le 1/300$. Therefore, $\{p_{(1)}, p_{(2)}\}=\{q_{(1)}, q_{(2)}\}$, which implies both Fact 1 and Fact 2. 
\end{proof}

\section{Interpretation and Estimation of Concentration Functions}\label{s.truncest}

Recall that $f_{\Dee}^{k_i}(n, \sigp, \beta)$ describes the concentration of $\widehatpi \sim \bersumPi{k_i}$ and is defined as,
\[\concentrationbound{k_i}{\Dee}{n}{\sigma_p^2}{\beta} = \arg\inf\{\alpha\;|\;\Pr_{\widehatp_1, \cdots, \widehatp_n\sim\bersumPi{k_i}  }\left(\max_i |\widehatp_i-p|\ge \alpha\right)\le \beta\}.\] In the main body of the paper, we assumed that this function was known to the analyst, even if the input value $\sigp$ was unknown and had to be estimated. In this appendix, we interpret the structure of this concentration function and show that even when this informational assumption is relaxed, our Algorithm \ref{alg.dp} can still be implemented with some minor modifications.

We start by introducing two additional functions: $f_{\Dee}(n, \sigp, \beta)$, which describes the concentration of $p_i \sim \Dee$, and $f_{\Bin}(k_i,p_i,\beta)$, which describes the high probability tail bound on the binomial $\Bin(k_i,p_i)$: 
\[ f_{\Dee}(n, \sigp, \beta) = \arg\inf\{\alpha\;|\; \Pr_{p_1, \cdots, p_n\sim \Dee}(\max_i|p-p_i|\ge \alpha)\le \beta\}. \]
\[ f_{\Bin}(k_i,p_i,\beta) = \arg\inf\{\alpha\;|\; \Pr_{x\sim \Bin(k_i,p_i)}(|\frac{1}{k_i}x-p_i|\ge \alpha)\le \beta\} \]
In this appendix, we will assume that only the function $f_{\Dee}(n,\cdot,\beta)$ is known to the analyst, but the input variance parameter $\sigp$ of the distribution is not known. For example, the analyst may know that $\Dee$ is Gaussian with unknown mean and variance, and thus she can express the concentration of $p_i$ as a function of the variance. Also note that for any values $k_i,p_i$ and $\beta$, we can empirically compute $f_{\Bin}(k_i,p_i,\beta)$. 

The following lemma shows how we can translate high probability bounds on $\Dee$ to high probability bounds on $\bersumP$, using this binomial tail bound of $\Bin(k_i,p_i)$. Specifically, it shows that our quantity of interest $f_{\Dee}^{k_i}(n, \sigp, \beta)$ of the $\widehat{p_i}$s can be upper and lower bounded by concentration of the $p_i$s (as described by $f_{\Dee}(n, \sigp, \beta)$) plus a binomial tail bound.

\begin{lemma}\label{concentrationimpliesconcentration}\label{anticoncentration} Suppose that $\Dee$ is supported on $[0,1/2]$. Given $k_i,n\in\mathbb{N}$, $\sigma_p^2$, and $\beta\in[0,1]$, define $\beta'=2\sqrt{1-\sqrt[n]{1-\beta}}=\Theta(\sqrt{\beta/n})$ and assume that \textcolor{black}{for all $p_i$ in the support of $\Dee$, \[\Pr_{\widehatp_i\sim\Bin(k_i,p_i)}(p_i-\widehatp_i\ge f_{\Bin}(k_i, p_i, \beta'))\ge \frac{1}{2}\beta' \;\;\text{ and }\;\; \Pr_{\widehatp_i\sim\Bin(k_i,p_i)}(\widehatp_i-p_i\ge f_{\Bin}(k_i, p_i, \beta'))\ge \frac{1}{4}\beta'.\]}
Then for all $\beta\in[0,1]$,
for all $i\in[n]$, 
\[\concentrationbound{k_i}{\Dee}{n}{\sigma_p^2}{\beta} \le f_{\Dee}(n, \sigp, \beta/2)+f_{\Bin}(k_i,p_{\max},\beta/n),\] where $p_{\max}=\min\{1/2, p+f_{\Dee}(n, \sigma_p, \beta/2)\}$.
Further, 
for all $i\in[n]$, 
\[\concentrationbound{k_i}{\Dee}{n}{\sigma_p^2}{\beta} \ge f_{\Dee}(1, \sigp, \beta')+f_{\Bin}(k_i, p_{\max}, \beta').\]
\end{lemma}

We note that the conditions on $\Dee$ and $\Bin(k_i,p_i)$ are mild. The condition on the tails of $\Bin(k_i,p_i)$ is intuitively claiming that $\Bin(k_i,p_i)$ is symmetric. This occurs whenever $k_i$ is large enough, and $p_i$ is bounded away from 0 or 1. We conjecture that the condition that $\Dee$ is supported on $[0, 1/2]$ can be relaxed but leave the relaxation to future work.

\begin{proof}[Proof of Lemma~\ref{concentrationimpliesconcentration}]
Notice that if $p<q<1/2$ then $f_{\Bin}(k_i,p,\beta)\le f_{\Bin}(k_i,q,\beta)$. Let us consider the upper bound first. 
With probability $1-\frac{\beta}{2}$,
\begin{equation}\label{concentrationofP}
\text{for all $i$,      } 
|p-p_i|\le f_{\Dee}(n, \sigma_p, \beta/2).
\end{equation}
Further, if Equation~\eqref{concentrationofP} holds then 
we have that with probability $1-\frac{\beta}{2n}$, \[|\widehat{p_i}-p_i|\le f_{\Bin}(k_i,p_i,\frac{\beta}{2n})\le f_{\Bin}(k_i, p_{\max}, \frac{\beta}{2n}).\]
Thus, for all $i$, 
\[|p-p_i|\le f_{\Dee}(n, \sigma_p, \beta/2)+f_{\Bin}(k_i, p_{\max}, \frac{\beta}{2n}).\]

Now, for the lower bound,
let $\beta'=\sqrt{8}\sqrt{1-\sqrt[n]{1-\beta}}$ and $\alpha=f_{\Dee}(1, \sigp, \beta')$.
Note that either \textcolor{black}{\[\Pr_{p_i\sim \Dee}\left(p_i-p\ge f_{\Dee}(1, \sigp, \beta')\right)\ge\frac{1}{2}\beta'\;\;\text{ or } \Pr_{p_i\sim \Dee}\left(p-p_i\ge f_{\Dee}(1, \sigp, \beta')\right)\ge\frac{1}{2}\beta'.\]} Assume without loss of generality that $\Pr_{p_i\sim \Dee}\left(p_i-p\ge f_{\Dee}(1, \sigp, \beta')\right)\ge\frac{1}{2}\beta'$.
Then by assumption,
\[\Pr_{\widehat{p_i}\sim \Bin(k_i,p_i)}\left(\widehatp_i-p_i\ge f_{\Bin}(k_i, p_i, \beta')\right)\ge\frac{1}{4}\beta'\]
Then
\begin{align*}
\Pr&\left(\max_i |\widehatp_i-p|\ge f_{\Dee}(1, \sigp, \beta')+f_{\Bin}(k_i, p+\alpha, \beta')\right)\\
&\ge \Pr\left(\exists i \text{ s.t. } p_i-p\ge f_{\Dee}(1, \sigp, \beta') \text{ and }\widehatp_i-p_i\ge f_{\Bin}(k_i, p_i, \beta')\right)\\
&= 1-\Pr\left(\forall i, p_i-p\le f_{\Dee}(1, \sigp, \beta') \text{ or }\widehatp_i-p_i\le f_{\Bin}(k_i, p+\alpha, \beta')\right)\\
&= 1-\left(\Pr\left(p_i-p\le f_{\Dee}(1, \sigp, \beta') \text{ or }\widehatp_i-p_i\le f_{\Bin}(k_i, p+\alpha, \beta')\right)\right)^n.
\end{align*}
Now, 
\begin{align*}
\Pr&\left(p_i-p\le f_{\Dee}(1, \sigp, \beta') \text{ or }\widehatp_i-p_i\le f_{\Bin}(k_i, p+\alpha, \beta')\right)\\
&= 1-\Pr\left(p_i-p\ge f_{\Dee}(1, \sigp, \beta') \text{ and }\widehatp_i-p_i\ge f_{\Bin}(k_i, p+\alpha, \beta')\right)\\
&= 1-\Pr\left(p_i-p\ge f_{\Dee}(1, \sigp, \beta')\right)\Pr\left(\widehatp_i-p_i\ge f_{\Bin}(k_i, p+\alpha, \beta')\;|\; p_i-p\ge f_{\Dee}(1, \sigp, \beta')\right)\\
&\le 1-\Pr\left(p_i-p\ge f_{\Dee}(1, \sigp, \beta')\right)\Pr\left(\widehatp_i-p_i\ge f_{\Bin}(k_i, p_i, \beta')\;|\; p_i-p\ge f_{\Dee}(1, \sigp, \beta')\right)\\
&\le 1-\Pr\left(p_i-p\ge f_{\Dee}(1, \sigp, \beta')\right)\Pr\left(\widehatp_i-p_i\ge f_{\Bin}(k_i, p_i, \beta')\right)\\
&\le 1-\frac{1}{8}(\beta')^2
\end{align*}
where the first inequality comes from $p_i\ge p+\alpha$, so $f_{\Bin}(k_i, p+\alpha, \beta')\le f_{\Bin}(k_i, p_i, \beta')$ Finally, 
\begin{align*}
\Pr\left(\max_i |\widehatp_i-p|\ge f_{\Dee}(1, \sigp, \beta')+f_{\Bin}(k_i, p+\alpha, \beta')\right)
\ge 1-(1-(\beta'/\sqrt{8})^2)^{n} = \beta,
\end{align*}
which implies the result.
\end{proof}

\subsection{Extending Our Results to Unknown $f_{\Dee}^{k_i}(n, \sigp, \beta)$ settings}

Lemma~\ref{concentrationimpliesconcentration} gives both upper bound and lower bounds on $\concentrationbound{k_i}{\Dee}{n}{\sigma_p^2}{\beta}$, which can be used to modify Algorithm \ref{alg.dp} and extend Theorem \ref{metatheorem} to apply in the setting where $\concentrationbound{k_i}{\Dee}{n}{\sigp}{\beta}$ is unknown, but $\concentrationbound{}{\Dee}{n}{\sigp}{\beta}$ is known instead.

Recall that the concentration bound $\concentrationbound{k_i}{\Dee}{n}{\sigp}{\beta}$ is used in Algorithm \ref{alg.dp} to define the truncation parameters $\widehat{a}_i$ and $\widehat{b}_i$, and that 
we would like to define a truncation window $[\widehat{a_i}, \widehat{b_i}]$ that both contains $[a_i,b_i]$ (so that with high probability none of the $\widehatp_i$ are truncated), and is not too wide, so $|\widehat{b_i}-\widehat{a_i}|\le 6|b_i-a_i|$ (in order to invoke Lemma \ref{lem.finalvariance}).

The following lemma proposes new values for $\widehat{a_i}$ and $\widehat{b_i}$ for the setting where only $\concentrationbound{}{\Dee}{n}{\sigma_p^2}{\beta}$ is known, but not $\concentrationbound{k_i}{\Dee}{n}{\sigma_p^2}{\beta}$. It combines the bounds on $\concentrationbound{k_i}{\Dee}{n}{\sigma_p^2}{\beta}$ from Lemma \ref{concentrationimpliesconcentration}, with the bounds on $\pinitial$ from Lemma~\ref{initialmeanestimate} to show that $|\widehat{b_i}-\widehat{a_i}|\le 6|b_i-a_i|$, as desired.

\begin{lemma}\label{lem.newlemma}
For $\alpha>0$,
let \[\widehat{a_i} = \max\left\{0, \widehat{p}-\alpha-f_{\Dee}(n, \widehat{\sigp}, \beta/2)-f_{\Bin}(k_i,\widehat{p}+\alpha+f_{\Dee}(n, \widehat{\sigma_p}, \beta/2),\beta/n)\right\} \]
and \[\widehat{b_i} = \min\left\{1, \widehat{p}+\alpha+f_{\Dee}(n, \widehat{\sigp}, \beta/2)+f_{\Bin}(k_i,\widehat{p}+\alpha+f_{\Dee}(n, \widehat{\sigp}, \beta/2),\beta/n)\right\}.\]
If 
$\widehat{\sigma_p^2}\ge\sigma_p^2$, and $|p-\widehat{p}|\le\alpha$, then for all $i\in[n]$, \[[a_i, b_i]\subset[\widehat{a_i}, \widehat{b_i}].\]
Further, if $\alpha\le \concentrationbound{k_i}{\Dee}{n}{\sigma_p^2}{\beta}$ and $\concentrationbound{k_i}{\Dee}{n}{\sigma_p^2}{\beta} \ge \Omega(f_{\Dee}(n, \sigp, \beta)+f_{\Bin}(k_i, \min\{1/2, p+f_{\Dee}(n, \sigma_p, \beta/2)\}, \beta/n))$
then \[|\widehat{b_i}-\widehat{a_i}|\le 6|b_i-a_i|.\] 
\end{lemma}

\showproofs{
\begin{proof}[Proof of Lemma \ref{lem.newlemma}]
Let us first show that $[\widehat{a}_i, \widehat{b}_i]\subset[a_i,b_i]$. 
Using our modified definition of $\widehat{a_i}$ given above, we have,
\begin{align*}
\widehat{a_i} &= \pinitial-\alpha-f_{\Dee}(n, \widehat{\sigp}, \beta/2)-f_{\Bin}(k_i,\pinitial+\alpha+f_{\Dee}(n, \widehat{\sigp}, \beta/2),\beta/n) \\
&\le p-f_{\Dee}(n, \widehat{\sigp}, \beta/2)-f_{\Bin}(k_i,p+f_{\Dee}(n, \widehat{\sigp}, \beta/2),\beta/n)\\
&\le p-f_{\Dee}(n, \sigp, \beta/2)-f_{\Bin}(k_i,p+f_{\Dee}(n, \sigma_p, \beta/2),\beta/n)\\
&\le p-\concentrationbound{k_i}{\Dee}{n}{\sigma_p^2}{\beta}\\
&= a_i.
\end{align*}
The first two inequalities respectively follow from the accuracy conditions on $\meanest$ and $\varianceest$ in Theorem \ref{metatheorem}; the third inequality comes from Lemma~\ref{concentrationimpliesconcentration}; and the final equality is by the definition of $a_i$. A symmetric result that $\widehat{b_i} \geq b_i$ follows similarly.

The second statement of this lemma ensures that the width of the truncation parameter is not more than a constant factor larger than the ideal.
Specifically, 
\begin{align*}
|\widehat{b_i}-\widehat{a_i}| &\le 2\alpha+2\left(f_{\Dee}(n, \widehat{\sigp}, \beta/2)+f_{\Bin}(k_i,\widehat{p}+\alpha+f_{\Dee}(n, \widehat{\sigp}, \beta/2),\beta/n)\right)\\
&\le 2\concentrationbound{k_i}{\Dee}{n}{\sigma_p^2}{\beta}+O(f_{\Dee}(1, \sigp, \beta')-f_{\Bin}(k_i, p+\alpha, \beta'))\\
&\le 2\concentrationbound{k_i}{\Dee}{n}{\sigma_p^2}{\beta}+2\left(2\concentrationbound{k_i}{\Dee}{n}{\sigma_p^2}{\beta}\right)\\
&\le 6\concentrationbound{k_i}{\Dee}{n}{\sigma_p^2}{\beta}\\
&= 6|b_i-a_i|
\end{align*}
\end{proof}}

We note that Lemma \ref{lem.finalvariance} as stated requires $|\widehat{b_i}-\widehat{a_i}| \leq 4|b_i-a_i|$, rather than $6|b_i-a_i|$, this difference of constants will only affect the constant $C$ in Theorem \ref{metatheorem}, and the main claim of a constant approximation in variance will still hold with these new $\widehat{a_i}$ and $\widehat{b_i}$ values.  

We will, however, have to add an additional assumption to Theorem \ref{metatheorem} in this setting. We will need to assume that $\Dee$ is s.t. $\concentrationbound{k_i}{\Dee}{n}{\sigma_p^2}{\beta} \ge \Omega(f_{\Dee}(n, \sigp, \beta)+f_{\Bin}(k_i, \min\{1/2, p+f_{\Dee}(n, \sigma_p, \beta/2)\}, \beta/n))$, to satisfy the condition of Lemma \ref{lem.newlemma}.  This condition is related to the high probability bound on $\bersumP$. 
The right hand side of this condition is the high probability bound on $\bersumP$ that is inherited directly from the high probability bounds on $\Dee$ and $\Bin(k,p)$. Without further assumptions on $\Dee$, this is the best upper bound on $\concentrationbound{k_i}{\Dee}{n}{\sigma_p^2}{\beta}$ that we can obtain, and hence is the bound used in the truncation in $\prealistic$. The condition states that this upper bound is within a constant multiplicative factor of the true value $\concentrationbound{k_i}{\Dee}{n}{\sigma_p^2}{\beta}$. We note that this condition is guaranteed by the lower bound on $\concentrationbound{k_i}{\Dee}{n}{\sigma_p^2}{\beta}$ in Lemma \ref{concentrationimpliesconcentration} for $\Dee$ with support on $[0,1/2]$, and we conjecture that it holds more broadly.

\end{document}